\documentclass[11pt,reqno]{amsart}

\usepackage{amsmath,amsfonts,amsthm,amssymb,amsxtra}
\usepackage[colorlinks,citecolor=red,pagebackref,hypertexnames=false]{hyperref} % for cited ref

\usepackage{color}
\usepackage[shortlabels]{enumitem}
\usepackage{bbm} % get mathbb1
\usepackage{stmaryrd}
\usepackage{nicematrix}
\usepackage{mathrsfs} % get script L, H
\usepackage{float}
\usepackage{graphicx}

\usepackage[margin=1.1in]{geometry}
\usepackage{extarrows}

\usepackage{subcaption}

\usepackage{tikz}
\usepackage{pgfplots}
\pgfplotsset{compat=1.18} % or 1.17 if your TeX install prefers that
\usetikzlibrary{calc}
\newtheorem{theorem}{Theorem}[section]

\newtheorem{proposition}[theorem]{Proposition}
\newtheorem{lemma}[theorem]{Lemma}
\newtheorem{corollary}[theorem]{Corollary}

\theoremstyle{definition}

\theoremstyle{remark}

\newtheorem{remark}{Remark}[section]

\numberwithin{equation}{section}
\allowdisplaybreaks

\newcommand{\A}{\mathbf{A}}

\newcommand{\C}{\mathbb{C}}

\renewcommand{\epsilon}{\varepsilon}

\newcommand{\cN}{\mathcal{N}}

\renewcommand{\phi}{\varphi}
\newcommand{\R}{\mathbb{R}}

\newcommand{\w}{\mathrm{weak}}
\newcommand{\Z}{\mathbb{Z}}

\newcommand{\E}{\mathbb{E}}

\DeclareFontFamily{U}{mathx}{}
\DeclareFontShape{U}{mathx}{m}{n}{<-> mathx10}{}
\DeclareSymbolFont{mathx}{U}{mathx}{m}{n}
\DeclareMathAccent{\widehat}{0}{mathx}{"70}
\DeclareMathAccent{\widecheck}{0}{mathx}{"71}

  \renewcommand{\mod}{{\rm \, mod\, }}

\newcommand{\ipc}[2]{ \langle #1 , #2  \rangle }

\renewcommand{\P}{\mathbb{P}}

\newcommand{\eps}{\varepsilon}

\newcommand{\vertiii}[1]{{\left\vert\kern-0.25ex\left\vert\kern-0.25ex\left\vert #1
    \right\vert\kern-0.25ex\right\vert\kern-0.25ex\right\vert}}

\DeclareMathOperator{\im}{Im}
\newcommand{\comments}[1]{}

\begin{document}

 \title[Div Grad]{Upper and Lower Bounds for the Quantum Dynamics of One-Dimensional Divergence-Type Random Jacobi Operators}

\author[L. Li, W. Wang, S. Zhang]{Long Li, Wei Wang, Shiwen Zhang}

\begin{abstract}
We study quantum transport for the discrete one-dimensional random Jacobi operator of divergence–gradient type. For strictly positive and bounded random variables, we analyze the $q$-moments of the position operator and establish both upper and lower power-law bounds on their growth. Our approach relies on the asymptotic behavior of the integrated density of states and the Lyapunov exponent near the critical energy $0$, previously obtained by Pastur and Figotin \cite{pastur1992book}. A key ingredient in our analysis is the large deviation-type estimates explored via the phase formalism, which play a central role in deriving bounds on the growth of the transfer matrices.
\end{abstract}

  \maketitle

\tableofcontents

\section{Introduction and Main Results}
Pioneered by the work of P. W. Anderson \cite{anderson58}, random operators have been extensively studied over the past several decades. In particular, one-dimensional discrete random operators are now well understood; see the most recent textbooks for the general one-dimensional ergodic case \cite{damanikfillman1,damanikfillman2}, an earlier textbook for random and almost-periodic operators \cite{pastur1992book}, as well as broader treatments of random operators in \cite{aizenman2015random}. In one dimension, random operators exhibit  spectral \emph{Anderson localization}-that is, with probability one, a pure point spectrum with exponentially localized eigenfunctions-at all energies under very general conditions. This phenomenon occurs not only for the standard Anderson model but also for operators with random hopping terms and more general Jacobi matrices (see, e.g., \cite{delyon83}). By contrast, dynamical localization, a stronger form of localization concerning quantum transport, requires additional assumptions even in one dimension.  A notable example is the random dimer model, which can be viewed as a correlated variant of the one-dimensional Anderson–Bernoulli model, where potential values occur in pairs rather than independently at each site. This model exhibits spectral localization at all energies \cite{debievre00}, yet allows super-diffusive quantum transport \cite{jitomirskaya2003deloc,jitomirskaya2007upper}. In this work, we consider another one-dimensional random operator that displays spectral localization but fails to exhibit dynamical localization.

More precisely, given a probability space $(\Omega,{\mathcal F},\P)$, let $\{a_n(\omega)\}_{n\in \Z,\omega \in \Omega }$ be a sequence of independent and identically distributed (i.i.d.) random variables with common distribution  $P_0$ on $\R$, induced by $\P$.  We study the one-dimensional random Jacobi operator on  $\ell^2(\Z)$, defined by:
\begin{align}\label{eqn:div-grad}
    (H_\omega\phi)_n=-a_{n+1}\phi_{n+1}+(a_{n+1}+a_n )\phi_n -a_n \phi_{n-1}.
\end{align}
We refer to $H_\omega$    a divergence-gradient-type operator, or simply {\emph{a div-grad model}}, since it can be rewritten as
$  (H_\omega\phi)_n=-[a_{n+1}(\phi_{n+1}-\phi_n)- a_n (\phi_n -\phi_{n-1})],$
which is the discrete analogue of the one-dimensional divergence-form differential operator on $L^2(\R)$
\begin{align}\label{eqn:div-grad-conti}
    {\mathcal L}\psi(x)=-\frac{d}{dx}\big(a(x)\frac{d \psi}{dx}\big),  \ a:\R\to \R.
\end{align}
These operators, along with their higher-dimensional generalizations (see \eqref{eqn:div-grad-high} and further discussion later), naturally arise in models of elasticity tensors for random structures, such as disordered lattices or inhomogeneous media. Their spectral properties are fundamental for understanding phenomena such as elastic and acoustic wave propagation in complex environments; see, for example, \cite{AizenmanMolchanov, FigotinKlein, figotin96loca}.

 We will focus on the one-dimensional case \eqref{eqn:div-grad}. We denote by ${\rm supp}P_0$ the (essential) support of   $P_0$, defined as $ {\rm supp}P_0=\big\{\, x \in \R\, |\, P_0(x-\eps,x+\eps)>0 \ {\rm for\ all\ } \eps>0\big\}.$
We assume that ${\rm supp}P_0$ contains more than one point (i.e., it is non-trivial) and is bounded away from both $0$ and $\infty$: 
\begin{align}\label{eqn:an-supp}
   0<a_-:= \inf {\rm supp}P_0\, < \, \sup {\rm supp}P_0=:a_+\, < \infty , 
     \end{align}
which ensures that  almost surely, 
\begin{align}\label{eqn:an-bound}
    a_- \le a_n \le a_+ \quad \text{for all } n \in \Z.
\end{align}
We refer to this property as \emph{uniform ellipticity}, in the sense commonly used for divergence-gradient-type operators.

It was shown in \cite{delyon83} (see also Proposition~\ref{prop:rev-spe} in Appendix~\ref{sec:spe-pf}) that the spectrum of $H_\omega$ is almost-surely a nonrandom set given by
\begin{align}\label{eqn:spe}
\Sigma:=    \sigma(H_\omega)=[0,4]\cdot{\rm supp}P_0,  \ \ \P{\rm -a.s.}, 
\end{align} 
where the product denotes the set of all element-wise products between $[0,4]$ and $\mathrm{supp}P_0$. Furthermore, \cite{delyon83} established that $H_\omega$  exhibits  spectral  \emph{Anderson localization}. 
     
In this work, we investigate the time-averaged $q$-moments of the position operator associated with $H_\omega$. We derive both upper and lower bounds of power-law type for these moments.

To be  precise, we consider the time-averaged
$q$-th
moments:
\begin{align}\label{eqn:Mtq}
    M_T^q=\int_0^\infty \frac{dt}{T}e^{-t/T}\, \sum_{n\in \Z} |n|^q\, \big|\ipc{\delta_n}{e^{-itH_{\omega}}\delta_0}\big| ^2 , ~q>0,  
\end{align}
where $\delta_n(j)=0$ iff $j=n$ is the  Kronecker  delta function and $\ipc{\cdot}{\cdot}$ denotes the standard Euclidean inner product on $\ell^2(\Z)$.  Let $\E(\cdot)$ denote the expectation with respect to  $\omega$. Our main result is the following:
\begin{theorem}\label{thm:main} 
   Let $H_\omega$ be as in \eqref{eqn:div-grad} satisfying \eqref{eqn:an-bound}. Then:
\begin{enumerate}
    \item  For  $q\ge 4$, 
   \begin{align}\label{eqn:beta-ave}
\beta^-_q:= \liminf_{T\to \infty} \frac{\log \E M_T^q}{q \log T } \ge \frac{1}{2}-\frac{2}{q} 
\end{align}   
    \item For   $q\ge  1  $ ,
     \begin{align} \label{eqn:beta+ave}
\beta^+_q:= \limsup_{T\to \infty} \frac{\log \E M_T^q}{q \log T } \le 1-\frac{1}{5q}.
\end{align}
\end{enumerate}

\end{theorem} 

This theorem implies that, on average, for large $q$ and $T$, 
\begin{align}\label{eqn:Mqt-asym}
T^{\frac{q}{2}-2} \lesssim   \E M_T^q \lesssim T^{q-\frac{1}{5}}. 
\end{align}

With the aid of a large deviation-type estimate, we also obtain an almost sure lower bound of slightly worse order:
\begin{theorem}\label{eqn:main-as}
   For  $q\ge 11/2$,  almost surely,   
\begin{align}\label{eqn:beta-as}
 \beta_q^{-,a.s.}:=   \liminf_{T\to \infty} \frac{\log M_T^q}{q\log T }\ge \frac{2}{5} - \frac{11}{5q}. 
\end{align}
\end{theorem}
Neither the upper nor the lower bound is sharp. Numerical evidence suggests that the quantum dynamics behaves nearly diffusively for large \(q\); specifically, the transport exponent appears to be approximately \(\frac{1}{2} + o_q(1)\), and the \(q\)-th moment at time \(t\) grows on the order of \(t^{\frac{q}{2}(1 + o_q(1))}\), where \(o_q(1) \to 0\) as \(q \to \infty\). See Figure~\ref{fig:mqt}.

\begin{figure}[tb]
  \centering

  % Row 1: two images
  \begin{subfigure}{0.32\linewidth}
    \centering
    \includegraphics[width=\linewidth]{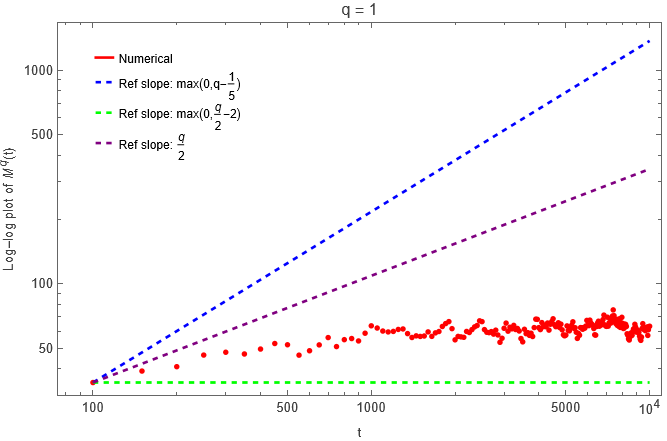}
    \label{fig:q1}
  \end{subfigure}\hfill
  \begin{subfigure}{0.32\linewidth}
    \centering
    \includegraphics[width=\linewidth]{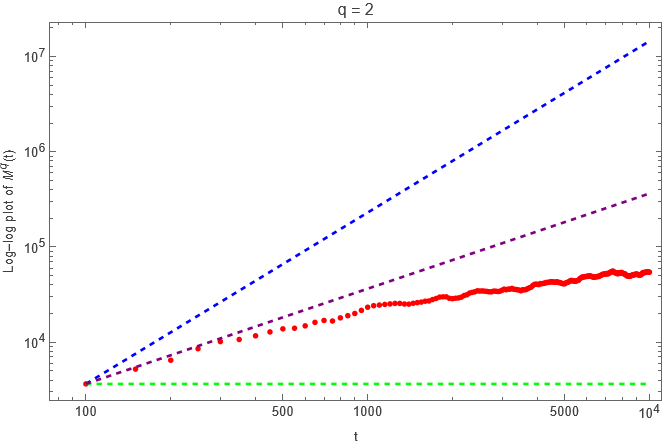}
    \label{fig:q2}
  \end{subfigure}
   \begin{subfigure}{0.32\linewidth}
    \centering
    \includegraphics[width=\linewidth]{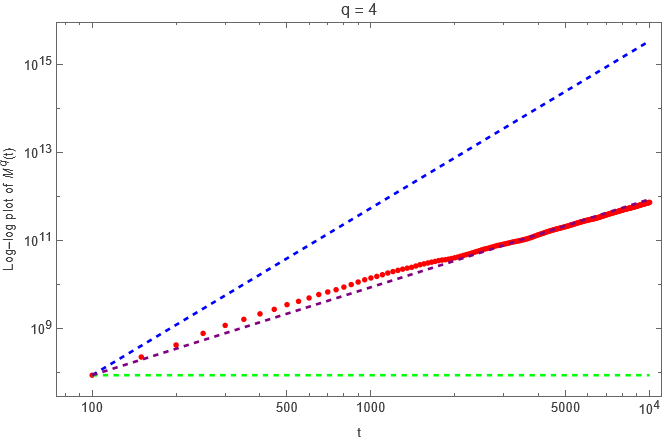}
    \label{fig:q4}
  \end{subfigure}

  \vspace{0.8em}

  % Row 2: three images
  \begin{subfigure}{0.32\linewidth}
    \centering
    \includegraphics[width=\linewidth]{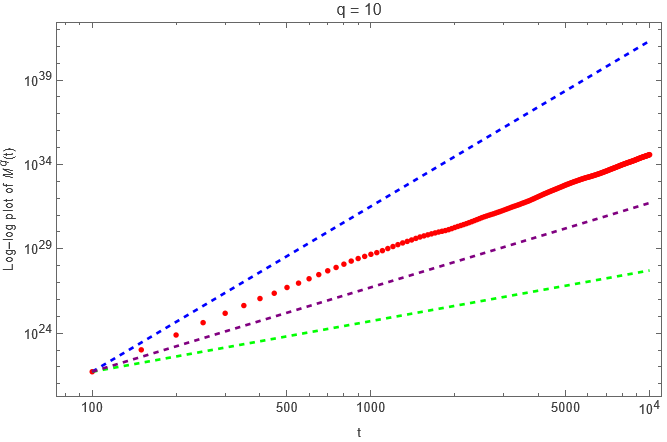}
    \label{fig:q6}
  \end{subfigure}\hfill
  \begin{subfigure}{0.32\linewidth}
    \centering
    \includegraphics[width=\linewidth]{fig/q10.png}
    \label{fig:q10}
  \end{subfigure}\hfill
  \begin{subfigure}{0.32\linewidth}
    \centering
    \includegraphics[width=\linewidth]{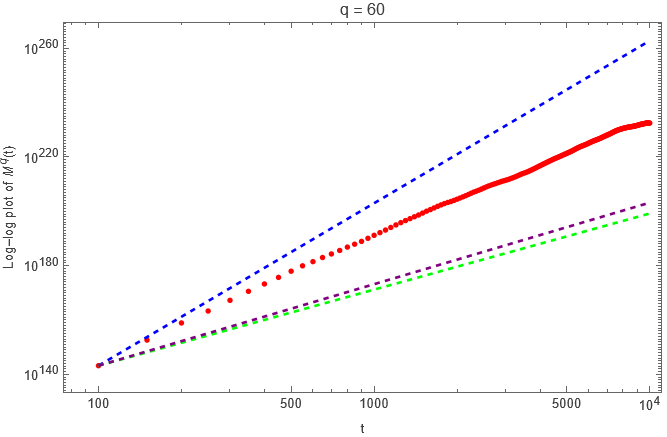}
    \label{fig:q60}
  \end{subfigure}

 \caption{
     Log–log plots of the non-averaged $q$-th moment
      $ M^q(t) = \sum_{k=0}^{n-1} k^q \,\big|\langle \delta_k, e^{-it H_{n,\omega}} \delta_0 \rangle\big|^2, $
        where $H_{n,\omega}$ is the restriction of the random operator $H_{\omega}$ to $\{0,1,\dots,n-1\}$.
        The random coefficients satisfy $a_i \sim \text{Uniform}[1,2]$.
        System size: $n = 10{,}000$.
        Panels correspond to $q =1, 2, 4, 6, 10, 60$ (top-left to bottom-right).
        Red solid line: numerical data; blue dashed: reference slope $\max(0, q - \tfrac{1}{5})$;
        green dashed: reference slope $\max(0, \tfrac{q}{2} - 2)$;
        purple dashed: reference slope $\tfrac{q}{2}$.
        Time grid: $t \in \{ 100, 150, \cdots, 10000 \}$. % with $t_0 = 100$, step $s = 50$.
    }
  \label{fig:mqt}
\end{figure}

For spectral problems in dimensional one, there is a unique but powerful dynamical systems approach through the $SL(2,\C)$ cocycles. Let  \(T_0^z=Id\), and for $n\ge 1$, 
\begin{align}\label{eqn:Tn-intro}
  T_n^z= A^z_{n-1} \cdots A^z_0  \quad {\rm and} \quad  T_{-n}^z=\big[A^z_{-1} \cdots A^z_{-n}\big]^{-1}, 
\end{align}
where 
\begin{align}\label{eqn:Aj-intro}
    A_j^z = \frac{1}{a_j} \begin{pmatrix} a_{j+1} + a_j - z & -a_j^2 \\ 1 & 0 \end{pmatrix}, \quad j\in \Z, \quad z\in\C
\end{align}
be the  transfer matrices associated with the (generalized) eigenvalue equation $H_\omega \varphi=z\varphi$ at some complex energy $z\in\C$. For i.i.d. coefficients \(a_n(\omega)\), we write \(T_n^z=T_n^z(\omega)\) to emphasize its dependence on the random variables.  
The \emph{Lyapunov exponent} is defined as
\begin{align}\label{eqn:Lyp}
    L(z) = \inf_{n\in\Z} \frac{1}{|n|} \mathbb{E}\!\left( \log \|T_n^z(\omega)\| \right)
    \xlongequal{\text{\rm a.s.}} \lim_{n \to \infty} \frac{1}{n}    \log \|T_n^z(\omega)\|.
\end{align}
The Lyapunov exponent characterizes the exponential growth rate of the transfer matrix and plays a crucial role in describing the growth or decay of solutions to \(H_\omega \varphi = z \varphi\). It serves as a measure of the localization length and is often referred to as the ``inverse localization length'' in physical contexts. For a real energy \(E\), a positive \(L(E)\) is a key indicator of possible localization or an upper bound on various transport exponents, while a vanishing exponent is often associated with potential delocalization or a lower bound on transport exponents.

A key ingredient that inspires our work is the following asymptotic expansion of \(L(E)\) near \(E = 0\), previously obtained by Pastur and Figotin \cite{pastur1992book}; see Figure~\ref{fig:lya_sing} for a numerical illustration.
\begin{theorem}[ {\cite[Theorem 14.6, Part (ii)]{pastur1992book}}]\label{thm:LE-linear}
Let \(H_\omega\) be the random div-grad model \eqref{eqn:div-grad} with coefficients \(a_n\) satisfying \eqref{eqn:an-bound}. Let \(L(z)\) be the associated Lyapunov exponent as in \eqref{eqn:Lyp}. Then \(L(z) \ge 0\) for all \(z \in \C\) and vanishes if and only if \(z = 0\). Moreover,
\begin{align}\label{eqn:linearLE}
    L(E) = \frac{\kappa E}{8} \, \mathbb{E} \big\{(  {a_0}^{-1} -  {\kappa}^{-1})^2 \big\} \, \big(1 + O(E^{1/2})\big)
\end{align}
as \(E \to 0^+\) with \(E \in \R\), where \(\kappa = \big[\mathbb{E}(a_0^{-1})\big]^{-1}\).
\end{theorem}
\begin{remark}\label{rmk:lyp-asym}
For a discrete Schr\"odinger operator \(-\Delta + g V_\omega\) with a small coupling constant \(g > 0\) and an i.i.d. random potential \(V_\omega=\{v_n\}_{n\in\Z}\), Figotin–-Pastur \cite[Theorem 14.6, Part (i)]{pastur1992book} employed phase formalisms, also known as modified Pr\"ufer variables, to derive the asymptotic formula for the Lyapunov exponent:
\[
L(g,E) = \frac{g^2 \, \E(v_0^2)}{2(4 - E^2)} + O(g^3), \qquad g \to 0.
\]
The remainder term \(O(g^3)\) depends on \(E\) but remains uniformly bounded for \(\delta < |E| < 2 - \delta\) for any \(\delta > 0\). In \cite[Theorem 14.6, Part (ii)]{pastur1992book}, the corresponding expansion for the div-grad operator \(H_\omega\) in \eqref{eqn:linearLE} was obtained by replacing \(g\), \(E\), and \(v_n\) with the corresponding terms for the div-grad model in the Schr\"odinger case formulas. Since this substitution was presented briefly, we provide supplementary details in Appendix~\ref{sec:lyp-asym} to clarify the dependence on the small energy parameter \(E\) at each step; see Lemma~\ref{lem:rho2-est}. These technical details supplement the original argument and may also be useful for deriving asymptotic formulas for the Lyapunov exponent in other related models.

Similar asymptotic results for the Lyapunov exponent in one-dimensional random isotopic chains were established in \cite{matsuda,Oconnor} using a different approach based on Furstenberg’s ergodic theorem \cite{furstenberg1963} for products of random matrices.
\end{remark}
\begin{remark}
A direct consequence of \eqref{eqn:linearLE} is that there exist constants \(D_0, D_1, E_0 > 0\) such that for \(0 < E < E_0\),
\begin{align}\label{eqn:LE-bound}
    D_0 E \le L(E) \le D_1 E.
\end{align}
The asymptotic behavior in \eqref{eqn:linearLE} applies as \(E \to 0^+\) within the spectrum. For \(E < 0\), outside the spectrum, the corresponding cocycle system \eqref{eqn:Tn-intro} is uniformly hyperbolic with a positive Lyapunov exponent. A straightforward computation (see Corollary~\ref{cor:Lyp-hyper} in Appendix~\ref{sec:lyp-hyp}) shows that there exist constants \(D_0', D_1', E_0' > 0\) such that for \(-E_0' < E < 0\),
\begin{align}\label{eqn:LE-bound-negative}
    D_0' \sqrt{-E} \le L(E) \le D_1' \sqrt{-E}.
\end{align}
\end{remark}

The asymptotic formula \eqref{eqn:linearLE} suggests that, although all eigenfunctions of \(H_\omega\) decay exponentially (as shown in \cite{delyon83}), the localization length grows like  \(1/E\) as \(E \to 0^+\). See Figure~\ref{fig:ef}. \begin{figure}[htbp]
    \centering
    \includegraphics[width=0.9\textwidth]{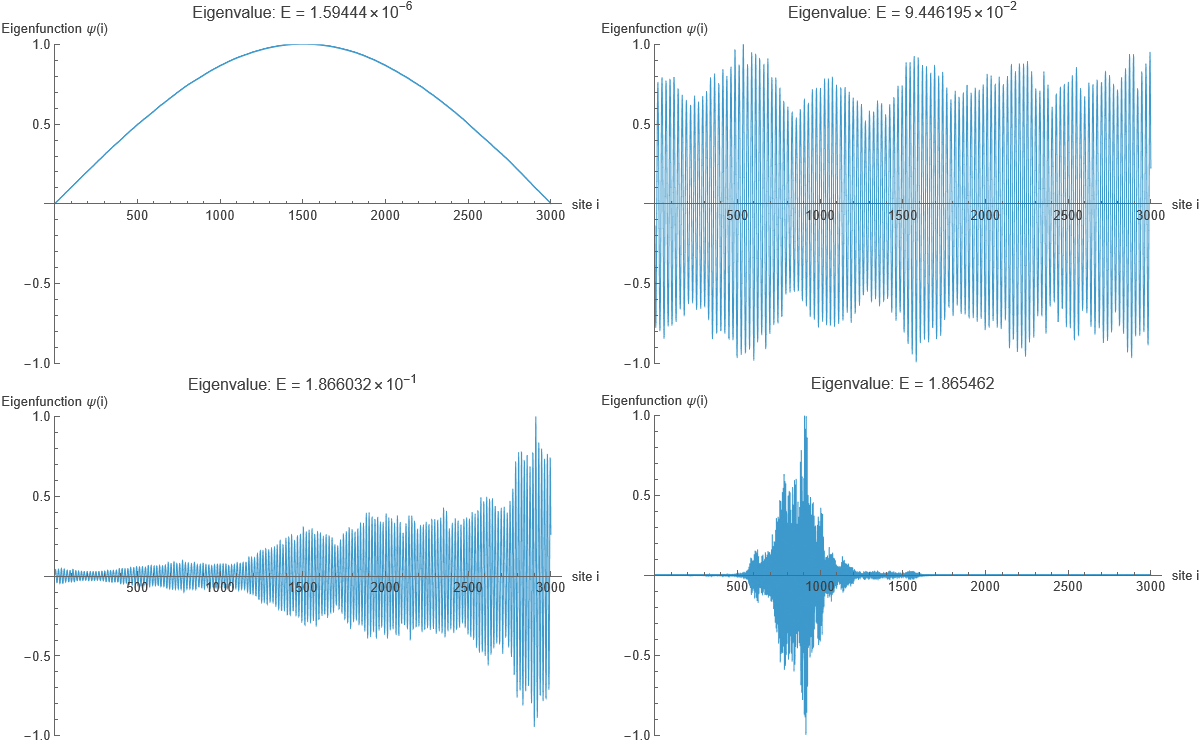}
    \caption{(Normalized) Eigenfunction plots \(\psi(i)\) for the div-grad model \(H_\omega\) with random entries \(a_j \sim \text{Uniform}[1,2]\), restricted to a finite size \(n = 3000\). The first panel corresponds to the ground state near \(E_0 = 0\). Panels 2--4 show eigenfunctions near target energies \(E_1 = 1/(kn)\), \(E_2 = 2/(kn)\), and \(E_3 = 20/(kn)\), where the localization length \(\ell \approx 1/(kE)\) is roughly \(n\), \(n/2\), and \(n/20\), respectively. Here \(k\) is the linear constant in the asymptotic formula \eqref{eqn:linearLE}, computed explicitly from \(\kappa\) and \(a_j\). As \(E\) decreases to \(0\), the eigenfunctions become less localized, illustrating the growth of the localization length predicted by the Lyapunov exponent in \eqref{eqn:linearLE}.}
    \label{fig:ef}
\end{figure} This behavior suggests the possibility of dynamical delocalization and nontrivial lower bounds on quantum transport generated by \(e^{-itH_\omega}\). Conversely, the fact that \(L(E)\) remains small yet strictly positive near zero indicates potential upper bounds on quantum dynamics. However, estimates such as \eqref{eqn:LE-bound} do not directly translate into transport bounds \eqref{eqn:beta-ave}, \eqref{eqn:beta+ave}, or \eqref{eqn:beta-as}.

A key technical challenge lies in the convergence rate of \(\frac{1}{n}\log\|T_n^E(\omega)\|\to L(E)\). One of the main technical accomplishment of this work is to establish large deviation estimates for the norm of transfer matrices \(T_n^z\) at complex energies \(z\) near \(0\); see Theorem~\ref{thm:ldt-Tn-norm} following a review of the phase formalism. These probabilistic bounds are essential for deriving both upper and lower transport estimates.

The critical energy \(E = 0\) is parabolic, lying at the boundary between the elliptic region (the spectrum) and the uniformly hyperbolic region (the resolvent set). Intuitively, one expects linear growth of the transfer matrix norm, \(\|T_n^0\|\lesssim |n|\) for \(n\neq 0\). By a telescoping argument from \cite[Theorem 2J]{simon96bounded}, this implies \(\|T_n^z\|\lesssim |n|\) whenever \(|n|^2|z|\lesssim 1\). These deterministic bounds, valid under \eqref{eqn:an-bound}, will be computed explicitly in Section~\ref{sec:transfer-bound}. They show that \(\|T_n^E\|\) remains bounded for \(|n|\le 1/\sqrt{E}\), but such estimates are too coarse to prove \eqref{eqn:beta+ave} and \eqref{eqn:beta-as}. In light of \eqref{eqn:linearLE}, one expects \(\|T_n^E\|\approx e^{(cE+o(1))|n|}\) with high probability, providing information for scales near and beyond the localization length \(|n|\lesssim 1/E\). Establishing these large deviation bounds is the focus of Section~\ref{sec:transfer-bound}; they may be further refined and prove useful for studying other models and related problems.

The rest of the paper is organized as follows. Section~\ref{sec:pre} presents basic facts about divergence–gradient random operators, the spectrum, the integrated density of states, and Lyapunov exponents. In Section~\ref{sec:transfer-bound}, we study the growth of the transfer matrix, review modified Pr\"ufer variables, and establish a large deviation theorem for the norm of the transfer matrix at complex energies. In Section~\ref{sec:lowerBDq}, we prove lower bounds on quantum dynamics in both expectation and almost sure cases. Section~\ref{sec:upper} bootstraps the large deviation estimates for transfer matrices and establishes upper bounds on quantum dynamics. The appendices provide supplementary material, including a proof of the deterministic spectrum, a refined analysis of asymptotic formulas for the integrated density of states and Lyapunov exponents via Pr\"ufer variables, estimates of quantum transport at resolvent energies, and bounds on the Borel transform that support the main results.

Throughout the paper, constants such as \(C\), \(c\), and \(c_i\) may change from line to line. We use the notation \(X \lesssim Y\) to mean \(X \le cY\), and \(X \gtrsim Y\) to mean \(X \ge cY\), for some constant \(c\) independent of \(n\) and \(E\) (usually either an abstract constant or depending only on the random distribution of \(a_n\)). If \(X \lesssim Y \lesssim X\), we may also write \(X \approx Y\). For \(\tau > 0\) and \(E \to 0^+\), we write \(X = O(E^\tau)\) as shorthand for \(X \lesssim E^\tau\).

\section{Basic Facts on Div-Grad Random operators}\label{sec:pre}

We begin with a brief review of some fundamental aspects of the spectral theory of random operators. For simplicity, most of the results in this subsection are presented in the context of the discrete one-dimensional model \eqref{eqn:div-grad}, although they apply more broadly to a wide range of other models. Our setting is the Hilbert space \(\ell^2 = \ell^2(\mathbb{Z}; \mathbb{C})\), consisting of square-summable, complex-valued sequences over the one-dimensional lattice, equipped with the standard inner product:
\[
\ipc{u}{v} = \sum_{n \in \mathbb{Z}} \bar{u}_n v_n.
\]
We consider the natural and convenient choice of the probability space \(\Omega = \mathbb{R}^{\mathbb{Z}}\), equipped with the product topology, the corresponding product \(\sigma\)-algebra and the product measure by cylinder sets. In this setting, the random coefficients at site \(n \in \mathbb{Z}\) are given by the \(n\)-th component of \(\omega\), i.e., \(a_n(\omega) = \omega(n)\). Measure-preserving, ergodic transformations on \(\Omega\) are induced by lattice shifts, defined as
\begin{align}\label{eqn:shift}
    (S\omega)(n) = \omega(n+1).
\end{align}

Let \(H_\omega\) be defined as in \eqref{eqn:div-grad}. For each realization \(\omega \in \Omega\), the operator \(H_\omega\) is self-adjoint on \(\ell^2\). We say that \(H_\omega\) is an ergodic operator in the sense that \(H_{S\omega} = U H_\omega U^\dagger\), where \((U\psi)_n = \psi_{n+1}\) is the unitary shift operator on \(\ell^2\). That is, for every \(\omega \in \Omega\), \(H_{S\omega}\) is unitarily equivalent to \(H_\omega\).

Birkhoff's ergodic theorem (see, e.g., \cite{krengel85ergodic} for a modern proof) implies many self-averaging properties for standard ergodic operators. In particular, quantities influenced by disorder (e.g., random variables) often converge almost surely to deterministic values. A classical application of ergodic theory to random operators—originating with Pastur \cite{pastur80spectralprop}—shows that the spectrum of the family \((H_\omega)_{\omega \in \Omega}\) is \(\mathbb{P}\)-almost surely a non-random set, denoted by \(\Sigma = \sigma(H_\omega)\).

 The connection between the div-grad model \eqref{eqn:div-grad} and its continuous analogue \(\mathcal{L}\) in \eqref{eqn:div-grad-conti} becomes clearer when we express the Hamiltonian \(H_\omega\) in its Dirichlet energy form:
\begin{align}\label{eqn:Diri-en}
    \ipc{\varphi}{H_\omega \varphi} = \sum_{n \in \mathbb{Z}} a_n \, |\varphi_{n} - \varphi_{n-1}|^2.
\end{align}
The higher-dimensional version of \eqref{eqn:div-grad}, denoted by \(H_\omega^d\) on \(\ell^2(\mathbb{Z}^d)\) for \(d \ge 2\), is defined via non-negative quadratic forms, analogous to \eqref{eqn:Diri-en}, as
\begin{align}\label{eqn:div-grad-high}
    \ipc{f}{H_\omega^d f} = \frac{1}{2} \sum_{\substack{n,m \in \mathbb{Z}^d \\ \|n - m\| = 1}} K_{n,m} \, |f_n - f_m|^2,
\end{align}
which can be interpreted as the discrete analogue (lattice approximation) of the higher-dimensional divergence–gradient differential operator on \(L^2(\mathbb{R}^d)\) in the form
\begin{align*}
    \mathcal{L}^d := -\nabla \cdot \big(K(x) \nabla \big), \quad K : \mathbb{R}^d \to \mathbb{R}.
\end{align*}
These operators describe fundamental aspects of wave propagation in inhomogeneous media. Aizenman--Molchanov \cite{AizenmanMolchanov} studied the discrete case and proved localization at extreme or high energies under certain regularity and decay conditions on the random coefficients. Figotin--Klein investigated the localization of classical acoustic waves modeled by such random operators, both in the discrete setting \cite{FigotinKlein} and in the continuum setting \cite{figotin96loca}, assuming the random vector field is a small perturbation of a periodic background.

From \eqref{eqn:Diri-en}, it follows that
\begin{align*}
    0 \le \ipc{\varphi}{H_\omega \varphi} \le 4 \max_n a_n \, \ipc{\varphi}{\varphi},
\end{align*}
which yields the one-sided spectral inclusion
\begin{align}\label{eqn:spe2}
    \sigma(H_\omega) \subset [0, 4 \max_n a_n] = [0,4] \cdot \mathrm{supp}\, P_0.
\end{align}
It turns out that the reverse inclusion also holds; that is, the spectrum \(\Sigma\) can be explicitly determined by the right-hand side of \eqref{eqn:spe2}; see \eqref{eqn:spe}. Note that the interval \([0,4]\) in \eqref{eqn:spe2} is precisely the spectrum of the discrete negative Laplacian \(-\Delta\). This result is based on a singular transformation that conjugates \(H_\omega\) to a disordered harmonic chain related to \(-\Delta\). Techniques developed by Kunz and Souillard \cite{kunz80Surle} for random Schr\"odinger operators—relying essentially on Weyl sequences and the Weyl criterion—can then be used to prove the reverse inclusion.  The precise expression \eqref{eqn:spe} is stated in \cite[Theorem~1,(4)]{delyon83} without proof. We will introduce a singular transform that will also be used for other estimates of quantum transport for \(H_\omega\) in the next section. We will also provide a proof of \eqref{eqn:spe}, based on this transformation, in Proposition~\ref{prop:rev-spe} in Appendix~\ref{sec:spe-pf}.

The \emph{density of states measure} (DOS), roughly speaking, counts the ``number of states per unit volume'' in a finite-volume system. The existence of the thermodynamic limit of the DOS for an ergodic operator can be established in various ways; see, e.g., \cite{Pastur1973,CL1990}. Below, we review a convenient definition for the specific one-dimensional div-grad model \(H_\omega\) in \eqref{eqn:div-grad}.

Let \(H_{N,\omega}\) denote the restriction of \(H_\omega\) to the interval \([0, N-1]\) with Dirichlet boundary conditions \(\psi_{-1} = \psi_N = 0\). Let
\[
E_1^N(\omega) \le E_2^N(\omega) \le \cdots \le E_N^N(\omega)
\]
be the eigenvalues of \(H_{N,\omega}\). Then the \emph{integrated density of states} (IDS) of \(H_\omega\) is defined by
\begin{align}\label{eqn:ids}
    \mathcal{N}(E) \xlongequal{\text{\rm a.s.}} \lim_{N \to \infty} \frac{1}{N} \#\left\{\, j : E_j^N(\omega) \le E \right\},
\end{align}
where \(\#\) denotes cardinality. The limit exists almost surely and is independent of the boundary conditions of \(H_{N,\omega}\). It is well known that for one-dimensional discrete random models such as \(H_\omega\) (see, e.g., \cite{pastur1992book,kirsch2007invitation,aizenman2015random}), the function \(\mathcal{N}(E)\) is a non-random, continuous\footnote{Without an ergodic setting, the existence of \(\mathcal{N}\) as a limit is non-trivial. Continuity of \(\mathcal{N}\) relies on the discrete setting; for example, there is no analogous result for Schr\"odinger operators on \(L^2(\mathbb{R}^d)\).} distribution function. The associated measure, denoted by \(d\mathcal{N}(E)\), is called the density of states measure. The support of this measure determines the spectrum: \(\sigma(H_\omega) = \mathrm{supp}\, d\mathcal{N}(E)\). Together with \eqref{eqn:spe}, we have \(\mathcal{N}(E) = 0\) for \(E \le 0\), \(\mathcal{N}(E) = 1\) for \(E \ge 4 \sup \mathrm{supp}\, P_0\), and \(\mathcal{N}(0) = 0\), where \(E = 0\) marks the bottom of the spectrum.

 Unlike the so-called fluctuation boundary—such as in the case of random Schr\"odinger operators on \(\ell^2(\mathbb{Z}^d)\)—the div-grad model \(H_\omega\) exhibits a stable spectral boundary, characterized by the following asymptotic behavior of its IDS near \(E = 0\); see Figure~\ref{fig:NE}.

\begin{figure}[htbp]
    \centering
    \includegraphics[width=0.42\linewidth]{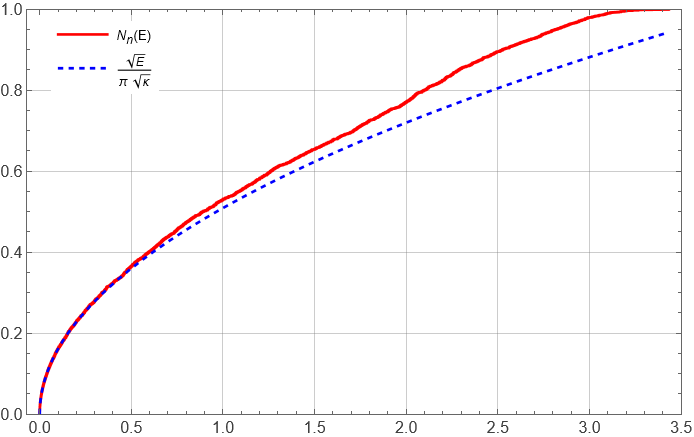}
    \includegraphics[width=0.42\linewidth]{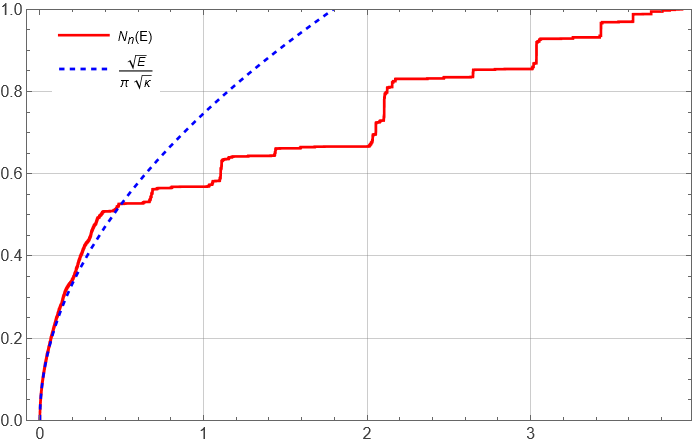}
    \caption{Finite volume IDS $N_n(E)$ with $n=3000$ for two cases: (left) i.i.d.\ $a_i \sim \text{Uniform}[0.1,1]$ and (right) i.i.d.\ $a_i \sim \text{Bernoulli}\{0.1,1\},\P(0.1)=0.5$. The plots show the integrated density of states with reference curves $y=\frac{\sqrt E}{\pi\sqrt \kappa}$, with $\kappa$ given by \eqref{eqn:NE-root}.}
    \label{fig:NE}
\end{figure}

\begin{theorem} 
\label{thm:NE-root}
For $H_\omega$ satisfying \eqref{eqn:an-bound}, we have
\begin{align}\label{eqn:NE-root}
    {\mathcal N}(E) = \frac{1}{\pi \sqrt{\kappa}} \sqrt{E} +  O(E) 
\end{align}
as $E \to 0^+$, where $ \kappa = \left[\mathbb{E}\left( a_0^{-1}\right)\right]^{-1}$ is the same constant as in \eqref{eqn:linearLE}. 
\end{theorem}

\begin{remark}
Similar to \eqref{eqn:LE-bound}, there exist constants \(E_0, D_0, D_1 > 0\) such that for all \(0 < E < E_0\),
\begin{align}\label{eqn:NE-bound}
    D_0 \sqrt{E} \le \mathcal{N}(E) \le D_1 \sqrt{E}.
\end{align}
For simplicity, we may assume \(0 < E_0 < 1\), \(D_0 < 1\), and \(D_1 > 1\), choosing them to coincide with the constants in \eqref{eqn:LE-bound} (by taking the appropriate maximum or minimum of the corresponding values) to reduce the number of distinct constants used.
The asymptotic formula \eqref{eqn:NE-root} was first proved in \cite[Theorem 6.6]{pastur1992book} for the continuous random model \(\mathcal{L}\) in \eqref{eqn:div-grad-conti}. The same asymptotic behavior holds for the discrete model \(H_\omega\), though it was left as an exercise in \cite[Problem 18, Page 183]{pastur1992book}. While the proof is conceptually similar to the continuous case, it involves additional technical challenges due to the discrete phase formalism. For completeness, we provide a detailed proof in Appendix~\ref{sec:prufer-app}.
\end{remark}

\begin{remark}
 Assume the edge weights satisfy \(0 < K_- \le K_{n,m} \le K_+ < \infty\) for all \(n,m\in\Z^d\). Denote by \(\mathcal{N}^d\) the IDS of the \(d\)-dimensional div-grad model \(H_\omega^d\) determined by such \(K_{n,m}\) on \(\ell^2(\mathbb{Z}^d)\) as in \eqref{eqn:div-grad-high}. Under these assumptions, we have
\[
    K_- \ipc{f}{-\Delta f} \le \ipc{f}{H_\omega^d f} \le K_+ \ipc{f}{-\Delta f}.
\]
By the min-max principle, the IDS of \(H_\omega^d\) is bounded between scaled versions of the free Laplacian IDS:
\[
    \mathcal{N}_0^d\!\left(\frac{E}{K_+}\right) \le \mathcal{N}^d(E) \le \mathcal{N}_0^d\!\left(\frac{E}{K_-}\right),
\]
where \(\mathcal{N}_0^d\) denotes the IDS of the \(d\)-dimensional free (negative) Laplacian \(-\Delta\). It is well known that \(\mathcal{N}_0^d(E)\) behaves like \(O(E^{d/2})\) as \(E \downarrow 0\), which in turn implies the same order of asymptotic behavior for \(\mathcal{N}^d(E)\). This generalizes the one-dimensional case in \eqref{eqn:NE-bound} to higher dimensions.

More precise asymptotic formulas were established by Anshelevich et al. \cite{Anshelevich1981}, Figari et al. \cite{Figari1982}, and Kozlov and Molchanov \cite{Kozlov1984}:
\[
    \mathcal{N}^d(E) = \mathcal{N}_{H_0}^d(E)\big(1 + o(1)\big), \quad \text{as } E \downarrow 0,
\]
where \(\mathcal{N}_{H_0}^d\) is the IDS of a deterministic operator \(H_0\) on \(\ell^2(\mathbb{Z}^d)\). Here, \(H_0\) is implicitly defined by a variational problem involving the parameters \(K_{n,m}\) from \(H_\omega^d\). In one dimension, \(H_0\) can be computed explicitly, leading to an equivalent form of the formula in \eqref{eqn:NE-root}. In higher dimensions, however, \(H_0\) generally cannot be determined in closed form. 
\end{remark}

The Lyapunov exponent, as defined in \eqref{eqn:Lyp}, is a dual quantity to the IDS. They are related via the well-known Thouless formula:
\begin{align}\label{eqn:thouless}
    L(z) = -\mathbb{E}(\log a_0) + \int_{\mathbb{R}} \log |z - E'| \, d\mathcal{N}(E'), \quad z \in \mathbb{C}.
\end{align}

A direct consequence is the following:
\begin{proposition}
    Let \(z = E + i\varepsilon \in \mathbb{C}\). For any \(E, \varepsilon > 0\),
\begin{align}\label{eqn:Lz-lower}
    L(z) \ge L(E) + (\ln 2) \cdot \big[ \mathcal{N}(E + \varepsilon) - \mathcal{N}(E - \varepsilon) \big]\ge L(E).
\end{align}
\end{proposition}

\begin{proof}
A direct computation using the Thouless formula \eqref{eqn:thouless} gives
\begin{align*}
    L(z) - L(E) 
    \ge \int_{|E - E'| \le \varepsilon} \log \!\left(1 + \frac{\varepsilon^2}{|E - E'|^2} \right) \, d\mathcal{N}(E') 
    \ge (\ln 2) \cdot \mathcal{N}\big(E' : |E - E'| \le \varepsilon\big).  
\end{align*}
\end{proof}

\begin{remark}
Both the asymptotic behaviors of the Lyapunov exponent in \eqref{eqn:linearLE} and the IDS in \eqref{eqn:NE-root} require that \(a_n\) be uniformly bounded away from zero, i.e., \(a_-=\inf \mathrm{supp}\, P_0 > 0\) as in \eqref{eqn:an-bound}. When \(a_- = 0\), numerical evidence indicates deviations from the prediction in \eqref{eqn:linearLE}; see Figure~\ref{fig:lya_sing} for examples of different asymptotic behaviors of \(L(E)\).
\end{remark}
\begin{figure}[htbp]
    \centering
    \begin{tabular}{cc}
        \includegraphics[width=0.45\textwidth]{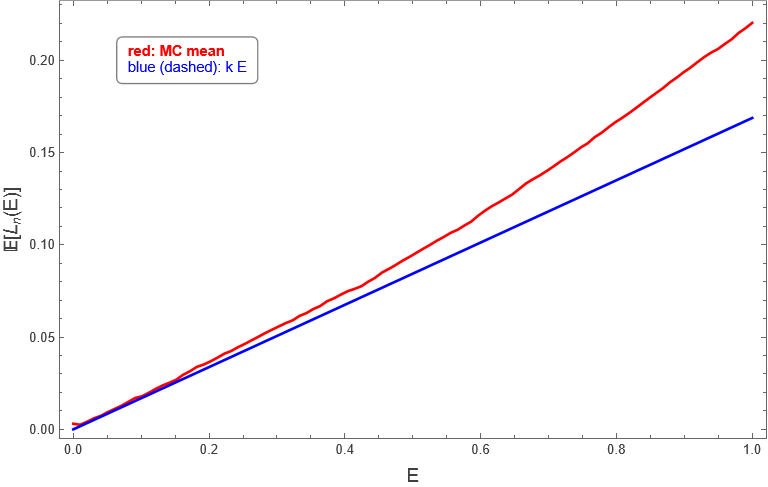} &
        \includegraphics[width=0.45\textwidth]{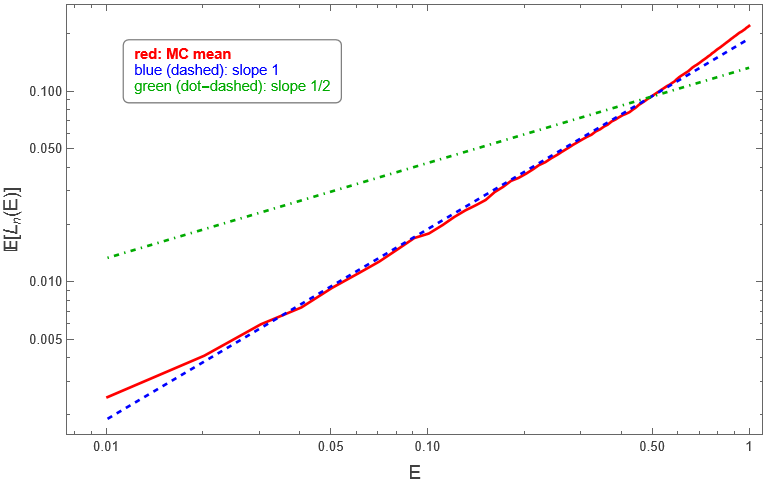} \\
        \includegraphics[width=0.45\textwidth]{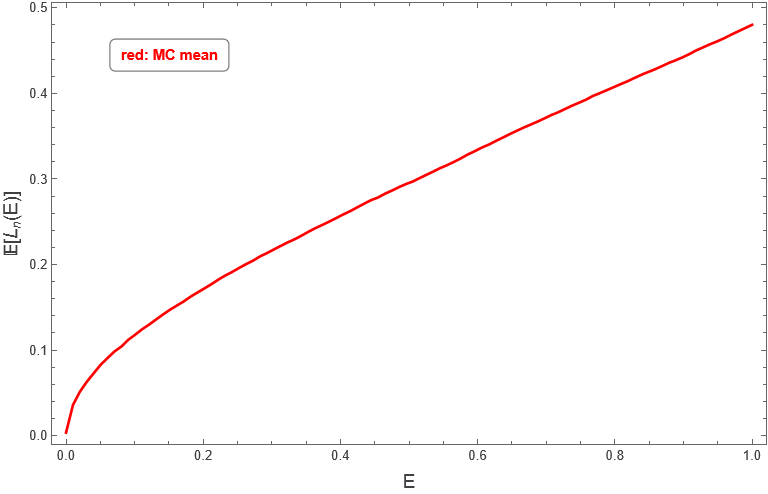} &
        \includegraphics[width=0.45\textwidth]{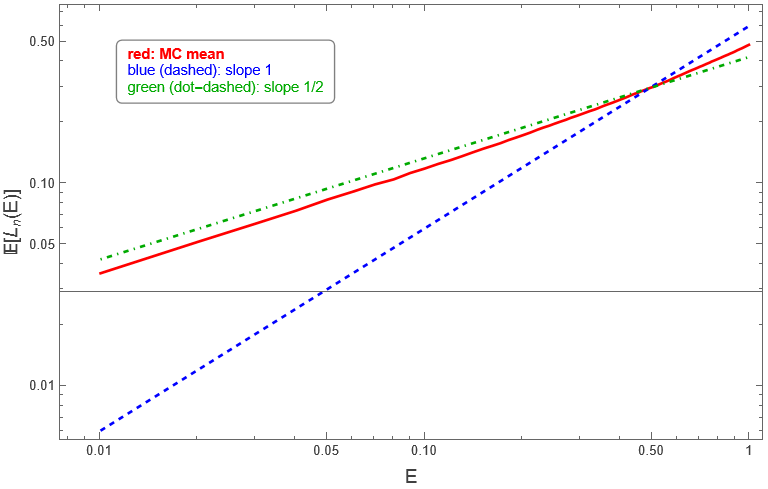} \\
    \end{tabular}
  \caption{Monte Carlo mean estimates of the Lyapunov exponent $L_n(E)=\frac{1}{n}\log \|T_n^E\|$ for $n = 3000$ under different uniform distributions of $a_j$. 
Top row: $a_j \sim \mathrm{Uniform}[0.1, 1]$, averaged over 100 replicates (samples). 
Bottom row: $a_j \sim \mathrm{Uniform}[0, 1]$, averaged over 100 replicates (samples). 
Each column shows: (a) linear scale with reference slope $kE$ from \eqref{eqn:linearLE}, where 
    $ k = \frac{\kappa }{8} \, \mathbb{E}  \{  ( 1/a_j - 1/\kappa  )^2  \} $
  and (b) log--log scale with slope~1 (blue, dashed) and slope~$\tfrac{1}{2}$ (green, dot-dashed) reference lines. 
Red curves represent numerical estimates. Note: When $\inf \operatorname{supp} P_0 = 0$, we have $\kappa = 0$ in \eqref{eqn:linearLE}.}
    \label{fig:lya_sing}
\end{figure}

\section{Upper Bounds on the Transfer Matrices}\label{sec:transfer-bound}

In this section, we establish bounds on \(\|T_n^z\|\), either on a full-measure set or on a set of large probability. These bounds on the norm of the transfer matrices will play a central role in the subsequent analysis, providing key ingredients for both the upper and lower bounds on quantum dynamics.

Recall the notation \(A_j^z\) and \(T_n^z\) from \eqref{eqn:Tn-intro}, where \(T_0^z = \mathrm{Id}\), and for \(n \ge 1\),
\begin{align*}
    T_n^z = A^z_{n-1} \cdots A^z_0, \quad \text{and} \quad T_{-n}^z = \big[A^z_{-1} \cdots A^z_{-n}\big]^{-1}.
\end{align*}
For \(z \in \mathbb{C}\), a sequence \(u = \{u_n\}_{n \in \mathbb{Z}}\) solves \(H_\omega u = z u\) if and only if
\begin{align}\label{eqn:u-cocycle}
    \begin{pmatrix}
        a_{n+1} u_{n+1} \\ u_n
    \end{pmatrix}
    = A_n^z
    \begin{pmatrix}
        a_n u_n \\ u_{n-1}
    \end{pmatrix}
    = T_n^z
    \begin{pmatrix}
        a_0 u_0 \\ u_{-1}
    \end{pmatrix}, \quad n \in \mathbb{Z}.
\end{align}
We write \(A_n^z(\omega)\) and \(T_n^z(\omega)\) when emphasizing dependence on the random variables. Let \(S\) be as in \eqref{eqn:shift}. Together with the definition \(a_n(\omega) = \omega(n)\), we have
\begin{align*}
    A_n^z(S^j \omega) = A_{n+j}^z(\omega), \quad \text{and} \quad
    T_n^z(S^j \omega) = A_{n-1+j}^z(\omega) \cdots A_j^z(\omega), \quad j \in \mathbb{Z}, \; n \ge 1.
\end{align*}

\subsection{Telescoping Argument and Deterministic Bound}

We first establish deterministic bounds on the norm \(\|T_n^z\|\). These bounds hold whenever \(a_n\) satisfies \eqref{eqn:an-bound}, and therefore apply on a full-measure set. 
\begin{lemma}\label{lem:telescope}
Suppose \(0 < a_- \le a_n \le a_+ < \infty\) for all \(n \in \mathbb{Z}\).
Define \(C_1 = \frac{8 a_+^2}{a_-}\) and \(C_2 = \frac{8 a_+^2}{a_-^2}\). Then, for all \(n \neq 0\) and \(z \in \mathbb{C}\),
\begin{align}\label{eqn:Tn0-linear}
    \|T_n^0\| &\le C_1 |n|,
\end{align}
and
\begin{align}\label{eqn:Tnz-exp}
    \|T_n^z\| &\le C_1 |n| e^{C_2 n^2 |z|}.
\end{align}
\end{lemma}

\begin{proof}
It suffices to prove the lemma for \(n > 0\), as the case \(n < 0\) can be handled similarly. For \(n > 0\), one can compute \(T_n^z\) at \(z = 0\) iteratively as
\begin{align*}
    T_1^0 = \begin{pmatrix}
        a_2\left( \frac{1}{a_1} + \frac{1}{a_2} \right) & -a_1 a_2 \cdot \frac{1}{a_2} \\
        \frac{1}{a_1} & 0
    \end{pmatrix}, 
    \  \cdots, \
    T_n^0 = \begin{pmatrix}
        a_{n+1} K(1, n+1) & -a_1 a_{n+1} K(2, n+1) \\
        K(1, n) & a_1 K(2, n)
    \end{pmatrix},
\end{align*}
where \(K(m, n) = \frac{1}{a_m} + \cdots + \frac{1}{a_n}\) for \(1 \le m \le n\).

It follows from the bounds on \(a_n\) that
\begin{align}\label{eqn:Tn0-pf}
    \|T_n^0\| \le 4 a_+^2 K(1, n+1) \le 4 a_+^2 \cdot \frac{1}{a_-}(n+1) \le C_1 n, \quad C_1 = \frac{8 a_+^2}{a_-}.
\end{align}
The same bound holds for \(\|T_n^0(S^j)\|\) for any \(j \in \mathbb{Z}\), since \(T_{n}^0(S^j) = A^0_{n+j-1} \cdots A^0_{j}\).

Next, by the telescoping argument from the proof of \cite[Theorem 2J]{simon96bounded}, for \(z \in \mathbb{C}\),
\begin{align*}
    T_n^z = T_n^0 - \sum_{j=1}^n \frac{z}{a_j} T_{n-j}^0(S^j) \begin{pmatrix}
        1 & 0 \\
        0 & 0
    \end{pmatrix} T_{j-1}^z.
\end{align*}
Direct iteration using \eqref{eqn:Tn0-pf} yields
\begin{align*}
    \|T_n^z(\omega)\| &\le C_1 n \sum_{k=0}^n \binom{n}{k} \left( \frac{C n |z|}{a_-} \right)^k 
    = C_1 n \left(1 + \frac{C_1 n |z|}{a_-} \right)^n 
    \le C_1 n \exp\left( \frac{C_1}{a_-} n^2 |z| \right).
\end{align*}
\end{proof}

\subsection{Modified  Pr\"ufer Variables and probabilistic bounds}\label{sec:prufer-ldt}
The estimate \eqref{eqn:Tnz-exp} implies that for \(z=E+i/T\), if \(|n|\lesssim  E^{-1/2}\) and \(|n|\lesssim   T^{-1/2}\), then \(\|T_n^z\|\lesssim E^{-1/2}\). As we discussed in the introduction, such estimates are too coarse and are not sufficient for proving \eqref{eqn:beta+ave} and \eqref{eqn:beta-as}. We need to extend the bounds of \(\|T_n^z\|\) for frequencies \(n\) up to the localization length \(|n| \lesssim E^{-1}\). The goal of this subsection is to obtain the following  large deviation estimates for the norm of transform matrices.  
\begin{theorem}\label{thm:ldt-Tn-norm}
There exist constants \(C, E_0 > 0\), depending only on \(\kappa, a_-, a_+\), such that for any \(\alpha > 0\) and \(z = E + i T^{-1}\) with \(0 < E < E_0\), if \(n^{1 + 2\alpha} E \le 1\) and \(n \le E^{\frac{3}{2}} T\), then
\begin{align}\label{eqn:ldt-TnZ}
    \mathbb{P}\Big( \|T_n^z(\omega)\| \le C E^{-\frac{3}{2}} \Big) \ge 1 - n e^{-n^\alpha}.
\end{align}
\end{theorem}

The proof relies on:
\begin{itemize}
    \item a singular transform that conjugates the divergence-gradient model to an isotopically disordered harmonic chain;
    \item the Figotin–-Pastur phase formalism for the isotopically disordered harmonic chain.
\end{itemize}

We first introduce the singular transform and explain how it conjugates the divergence–gradient model to an isotopically disordered harmonic chain. For a real energy \(E > 0\), let \(u = \{u_j\}_{j \in \mathbb{Z}}\) satisfy
\begin{align}\label{eqn:hu=eu}
    -a_{n+1} u_{n+1} + (a_{n+1} + a_n) u_n - a_n u_{n-1} = E u_n, \quad n \in \mathbb{Z}.
\end{align}
Define \(v_n = a_n (u_n - u_{n-1})\) for \(n \in \mathbb{Z}\). Then \(v = \{v_j\}_{j \in \mathbb{Z}}\) satisfies
\begin{align}\label{eqn:hv=ev}
    -v_{n+1} + 2 v_n - v_{n-1} = \frac{E}{a_n} v_n.
\end{align}
The transformed equation \eqref{eqn:hv=ev} is commonly referred to as an isotopically disordered harmonic chain, which has been studied in, e.g., \cite{matsuda,Oconnor}; see also \cite[\S7.1]{lifshits88book}.

The change of variables between \((u_n, u_{n-1})\) and \((v_n, v_{n-1})\), depending on \(E\), can be expressed as
\begin{align}\label{eqn:u-v-change}
    \begin{cases}
        v_n = a_n u_n - a_n u_{n-1}, \\
        v_{n-1} = a_n u_n + (E - a_n) u_{n-1},
    \end{cases}
    \quad \Longleftrightarrow \quad
    \begin{pmatrix}
        v_n \\
        v_{n-1}
    \end{pmatrix}
    = W_n
    \begin{pmatrix}
        a_n u_n \\
        u_{n-1}
    \end{pmatrix},
\end{align}
where \(W_n\) is invertible for any \(E \neq 0\), and
\begin{align}\label{eqn:W}
    W_n =
    \begin{pmatrix}
        1 & -a_n \\
        1 & E - a_n
    \end{pmatrix},
    \qquad
    W_n^{-1} = \frac{1}{E}
    \begin{pmatrix}
        E - a_n & a_n \\
        -1 & 1
    \end{pmatrix}.
\end{align}

Similarly to \eqref{eqn:u-cocycle}, we can rewrite \eqref{eqn:hv=ev} in cocycle form: for \(n \ge 1\),
\begin{align}\label{eqn:v-cocycle}
    \begin{pmatrix}
        v_n \\ v_{n-1}
    \end{pmatrix}
    = B^E_{n-1}
    \begin{pmatrix}
        v_{n-1} \\ v_{n-2}
    \end{pmatrix}
    = F_n^E
    \begin{pmatrix}
        v_0 \\ v_{-1}
    \end{pmatrix},
\end{align}
where
\[
    B^E_j =
    \begin{pmatrix}
        2 - \frac{E}{a_j} & -1 \\
        1 & 0
    \end{pmatrix},
    \quad \text{and} \quad
    F_n^E = B^E_{n-1} \cdots B^E_0, \quad n \ge 1, \; j \in \mathbb{Z}.
\]
The cocycle for \(n \le 0\) is defined similarly. Note that from \eqref{eqn:v-cocycle}, for \(n \ge 1\), the term \(v_n\) depends only on the random variables \(a_0, \dots, a_{n-1}\), and is independent of \(a_n\). It follows from \eqref{eqn:u-cocycle}, \eqref{eqn:W}, and \eqref{eqn:v-cocycle} that
\begin{align}\label{eqn:A-B}
    B_n^E = W_{n+1} A_n^E W_n^{-1}, \quad \text{and} \quad F_n^E = W_n T_n^E W_0^{-1}.
\end{align}
\begin{proposition}
There exists a constant \(c > 0\), depending only on \(a_+\) in \eqref{eqn:an-bound}, such that for any \(n \ge 0\) and \(0 < E \le 4 a_+\),
\begin{align}\label{eqn:F-Tn-norm}
    \frac{E}{c} \|F_n^E\| \le \|T_n^E\| \le \frac{c}{E} \|F_n^E\|.
\end{align}
\end{proposition}
\begin{proof}
This follows directly from the inequality
\[
    \|F_n^E\| \le \|W_n\| \cdot \|T_n^E\| \cdot \|W_0^{-1}\|,
\]
and the bounds
\[
    \|W_n\| \le |E| + 2|a_n| + 2, \qquad \|W_n^{-1}\| \le \frac{\|W_n\|}{E}
\]
for any \(n\). Similar arguments yield the corresponding bound for \(\|T_n^E\|\).
\end{proof}

We now introduce the phase formalism for the new coordinates \(\{v_n\}_{n \ge 0}\).
Let \(u, v\) be solutions of \eqref{eqn:hu=eu} and \eqref{eqn:hv=ev} respectively for \(n \ge 0\) and \(E > 0\). The initial condition for \(u\) is normalized as
\begin{align}\label{eqn:u-initial}
    |a_0 u_0|^2 + |u_{-1}|^2 = 1,
\end{align}
and the corresponding initial condition for \(v\) is given by \eqref{eqn:u-v-change}. 

Let \(\kappa = \big[\mathbb{E}(a_0^{-1})\big]^{-1}\) be as in \eqref{eqn:LE-bound}. Define
\begin{align}\label{eqn:eta}
    P = \begin{pmatrix}
        1 & -\cos \eta \\
        0 & \sin \eta
    \end{pmatrix}, 
    \quad \text{where} \quad 
    \eta(E) = \cos^{-1}\!\Big(1 - \frac{1}{2\kappa} E\Big) = \frac{\sqrt{E}}{\sqrt{\kappa}} + O(E^{3/2}).
\end{align}
 The matrix \(P\) is invertible with inverse
\[
    P^{-1} = \frac{1}{\sin \eta} 
    \begin{pmatrix}
        \sin \eta & \cos \eta \\
        0 & \sin \eta
    \end{pmatrix},
\]
and satisfies
\begin{align}\label{eqn:P-norm}
    \|P\| \le 2, \quad 
    \|P^{-1}\| \le \frac{2}{\sin \eta} \le \frac{2\sqrt{2\kappa}}{\sqrt{E}}, 
    \quad \text{for } 0 < E < 2\kappa.
\end{align}

We define the (modified) Pr\"ufer variables \(\big(\rho_n(E), \chi_n(E)\big)\) with respect to the matrix \(P\) for \((v_n, v_{n-1})\) in \eqref{eqn:v-cocycle} as
\begin{align}\label{eqn:modi-prufer}
    \rho_n(E)
    \begin{pmatrix}
        \cos \chi_n(E) \\
        \sin \chi_n(E)
    \end{pmatrix}
    = P
    \begin{pmatrix}
        v_n \\ v_{n-1}
    \end{pmatrix}, \quad n \ge 0,
    \quad \text{with} \quad
    \begin{pmatrix}
        v_0 \\ v_{-1}
    \end{pmatrix}
    = W_0
    \begin{pmatrix}
        a_0 u_0 \\ u_{-1}
    \end{pmatrix}.
\end{align}

The following iteration is obtained by direct computation from this definition.
\begin{proposition}\label{prop:rho-chi-Qn}
For any \(n \in \mathbb{Z}_{\ge 0}\) and \(0 < E < 2\kappa\),
\begin{align}\label{eqn:rho-n-iter}
\begin{cases}
    \rho_{n+1} \cos \chi_{n+1} = \rho_n \Big[\cos(\chi_n + \eta) + Q_n \sin(\eta + \chi_n)\Big], \\
    \rho_{n+1} \sin \chi_{n+1} = \rho_n \sin(\chi_n + \eta),
\end{cases}
\end{align}
where
\begin{align}\label{eqn:Qn}
    Q_n = \frac{E}{\sin \eta(E)} \big(\kappa^{-1} - a_n^{-1}\big)
    = \big(\kappa^{-1} - a_n^{-1}\big)\Big(\sqrt{\kappa E} + O(E^{3/2})\Big).
\end{align}

As a consequence,
\begin{align}\label{eqn:Qn-bound}
    \mathbb{E}[Q_n] = 0, \quad \text{and} \quad |Q_n| \le \sqrt{2\kappa^{-1} E}.
\end{align}
In addition, for \(n \ge 0\), the random variable \(\chi_n\) depends only on \(a_0, \dots, a_{n-1}\) and is independent of \(a_n\), while \(Q_n\) depends only on \(a_n\) and is therefore independent of \(\chi_0, \dots, \chi_n\).
\end{proposition}

\begin{proof}
The matrix \(P\) and the recurrence relation \eqref{eqn:rho-n-iter} originate from Figotin–-Pastur \cite[Theorem 14.6, Part (i)]{pastur1992book} for a Schr\"odinger operator \(-\Delta + g V_\omega\) with a small coupling constant \(g > 0\) and an i.i.d. random potential. One obtains \eqref{eqn:rho-n-iter} by substituting \(g \mapsto \frac{E}{\sin \eta(E)}\) and \(V_\omega \mapsto \kappa^{-1} - a_n^{-1}\) in the corresponding formulas of Figotin–-Pastur. Additional details of \eqref{eqn:rho-n-iter} are provided in Appendix~\ref{sec:prufer-app} for completeness; see Lemma~\ref{lem:zeta-order}. The main difference lies in the expression of \(Q_n\), particularly its asymptotic behavior in the small parameter \(E\).

  In \eqref{eqn:Qn}, we used the analytic expression for \(\eta(E)\) in \eqref{eqn:eta}, which implies that \(Q_n\) is analytic as a function of \(\sqrt{E}\) at \(\sqrt{E} = 0\), with the expansion 
  \[ Q_n =\sqrt E\Big(\frac{1}{\kappa}-\frac{E}{4\kappa^2}\Big)^{-1/2}(\kappa^{-1}-a_n^{-1})=\sqrt E\Big[\kappa^{1/2}+O(E)  \Big](\kappa^{-1}-a_n^{-1}), \quad {\rm as}\quad E\to 0^+. \]
Note that the \(O(E)\) term does not depend on \(n\), nor does the \(O(E^{3/2})\) term in \eqref{eqn:Qn}. We will use the fact that \(Q_n = O(\sqrt{E})\) uniformly in \(n\) throughout. The expectation \(\mathbb{E}[Q_n]\) vanishes by direct computation:
\[
    \mathbb{E}[Q_n] = \frac{E}{\sin \eta(E)} \mathbb{E}[\kappa^{-1} - a_n^{-1}] = 0.
\]
The upper bound in \eqref{eqn:Qn-bound} follows from the explicit estimate: for \(0 < E \le 2\kappa\),
\[
    \Big(\frac{1}{\kappa} - \frac{E}{4\kappa^2}\Big)^{-1/2}|\kappa^{-1} - a_n^{-1}|
    \le \sqrt{2\kappa}\,|\kappa^{-1}(1 - \kappa a_n^{-1})|
    \le \sqrt{2\kappa^{-1}}.
\]
Finally, recall from \eqref{eqn:v-cocycle} that
\begin{align}\label{eqn:vn-dep}
    \begin{pmatrix}
        v_n \\ v_{n-1}
    \end{pmatrix}
    =
    \begin{pmatrix}
        2 - \frac{E}{a_{n-1}} & -1 \\
        1 & 0
    \end{pmatrix}
    \cdots
    \begin{pmatrix}
        2 - \frac{E}{a_0} & -1 \\
        1 & 0
    \end{pmatrix}
    \begin{pmatrix}
        v_0 \\ v_{-1}
    \end{pmatrix}.
\end{align}
This shows that \(v_n, v_{n-1}\) depend only on \(a_0, \dots, a_{n-1}\) and are independent of \(a_n\). It then follows from \eqref{eqn:modi-prufer} that \(\chi_n\) depends on \(v_n, v_{n-1}\), and hence on \(a_0, \dots, a_{n-1}\), but is independent of \(a_n\). On the other hand, \(Q_n\), as given in \eqref{eqn:Qn}, depends only on \(a_n\). Therefore, since \(\{a_n\}_{n \in \mathbb{Z}}\) are i.i.d., \(Q_n\) is independent of \(\chi_0, \dots, \chi_n\). 
\end{proof}

Next, we establish large deviation–type estimates for the radial variables \(\rho_n\) and the transfer matrices \(T_n^z\). Iterating using \eqref{eqn:v-cocycle} and \eqref{eqn:modi-prufer} gives, for \(n \ge 0\),
\begin{align}\label{eqn:rho-F}
    \rho_n
    \begin{pmatrix}
        \cos \chi_n(E) \\
        \sin \chi_n(E)
    \end{pmatrix}
    = P
    \begin{pmatrix}
        v_n \\ v_{n-1}
    \end{pmatrix}
    = P F_n^E P^{-1} \rho_0
    \begin{pmatrix}
        \cos \chi_0(E) \\
        \sin \chi_0(E)
    \end{pmatrix}.
\end{align}

A direct computation shows that for any \(E > 0\), there is a one-to-one correspondence between the normalized initial condition \((a_0 u_0, u_{-1})\) and \(\chi_0 \in [0,\pi)\). In other words, by varying \((a_0 u_0, u_{-1})\) in \eqref{eqn:u-initial}, the parameter \(\chi_0\) can attain all values in \([0,\pi)\); see Appendix~\ref{sec:prufer-app}. Consequently,
\begin{align}\label{eqn:Fn-rhon}
    \|P F_n^E P^{-1}\|
    = \max_{\chi_0 \in [0,\pi)} \Big\| P F_n^E P^{-1}
    \begin{pmatrix}
        \cos \chi_0 \\
        \sin \chi_0
    \end{pmatrix} \Big\|
    = \frac{\rho_n}{\rho_0}.
\end{align}
Squaring the two equations in \eqref{eqn:rho-n-iter} and adding both sides together gives
\begin{align*}
    \rho_{n+1}^2
    = \rho_n^2 \Big[ 1 + Q_n \sin 2(\chi_n + \eta) + Q_n^2 \sin^2(\chi_n + \eta) \Big].
\end{align*}
Inductively, one obtains
\begin{align*}
    \log \rho_n^2
    = \log \rho_0^2 + \log \prod_{i=0}^{n-1} \Big[ 1 + Q_i \sin 2(\chi_i + \eta) + Q_i^2 \sin^2(\chi_i + \eta) \Big].
\end{align*}
Using \(Q_n \sim \sqrt{E}\) (uniformly in \(n\)) and the expansion of \(\log(1+x)\) near zero, we have
\begin{align}\label{eqn:rhon-rho0}
    \frac{1}{n} \log \frac{\rho_n}{\rho_0}
    &= \frac{1}{2n} \sum_{i=0}^{n-1} \Big[ Q_i \sin 2(\chi_i + \eta)
    + Q_i^2 \sin^2(\chi_i + \eta)
    - \frac{1}{2} Q_i^2 \sin^2 2(\chi_i + \eta)
    + O(Q_n^3) \Big] \\
    &= \frac{1}{8n} \sum_{i=0}^{n-1} Q_i^2 \label{eqn:rho-0} \\
    &\quad + \frac{1}{2n} \sum_{i=0}^{n-1} Q_i \sin 2(\chi_i + \eta) \label{eqn:rho-1} \\
    &\quad + \frac{1}{8n} \sum_{i=0}^{n-1} \Big[ -2 Q_i^2 \cos 2(\chi_i + \eta)
    + Q_i^2 \cos 4(\chi_i + \eta) \Big] + O(E^{3/2}).\label{eqn:rho-2}  
\end{align}
Due to the uniform bound on \(Q_n=O(\sqrt{E})\) in \eqref{eqn:Qn-bound}, there exist constants \(C > 0\) and \(E_0 > 0\), depending on \(\kappa\), such that for \(0 < E < E_0\) and any \(n > 0\),
\begin{align}
    \big| \eqref{eqn:rho-0} + \eqref{eqn:rho-2} \big| \le C E,\label{eqn:rho-CE}
\end{align}
uniformly on a full-measure set. On the other hand, by the same rough bound, we see that \eqref{eqn:rho-1} has a uniform order of \(O(\sqrt{E})\). The following lemma states that it is also of order \(O(E)\) on a large-probability set for \(n\) not too large. 

\begin{lemma}\label{lem:azuma-Qnsin}
   For any \(\alpha > 0\), \(1 \le m \le n\), and \(0<E\le 2\kappa\), 
\begin{align}\label{eqn:azuma-Qnsin}
    \mathbb{P}\Big( \sum_{i=0}^{m-1} Q_i \sin(\chi_i + \eta)
    < 2 \kappa^{-\frac{1}{2}} E^{\frac{1}{2}} n^{\frac{\alpha}{2}} m^{\frac{1}{2}} \Big)
    \ge 1 - e^{-n^{\alpha}}.
\end{align}
\end{lemma}
The proof of Lemma~\ref{lem:azuma-Qnsin} is deferred to the end of this section. We first use this lemma to prove:
    \begin{proof}[Proof of Theorem \ref{thm:ldt-Tn-norm}]
For any \(\alpha > 0\), assume that \(E < E_0\) as in \eqref{eqn:rho-CE}, \(nE \le n^{-2\alpha}\), and \(m \le n\). If
\[
    \sum_{i=0}^{m-1} Q_i \sin(\chi_i + \eta)
    < 2 \kappa^{-\frac{1}{2}} E^{\frac{1}{2}} n^{\frac{\alpha}{2}} m^{\frac{1}{2}},
\]
then by \eqref{eqn:rhon-rho0} and \eqref{eqn:rho-CE}, we have
\begin{align*}
    \log \frac{\rho_m}{\rho_0}
     \le   \frac{1}{2} \sum_{i=0}^{m-1} Q_i \sin(\chi_i + \eta) + C m E 
     \le \kappa^{-\frac{1}{2}} E^{\frac{1}{2}} n^{\frac{\alpha}{2}} m^{\frac{1}{2}} + C n E  
     \le \kappa^{-\frac{1}{2}} + C.  
\end{align*}
This implies, by \eqref{eqn:Fn-rhon}, that
\[
    \|P F_m^E P^{-1}\| = \frac{\rho_m}{\rho_0} \le e^{\kappa^{-\frac{1}{2}} + C}.
\]
 Combining this with \eqref{eqn:F-Tn-norm} and \eqref{eqn:P-norm},   we obtain
\begin{align*}
    \|T_m^E\|
     \le \frac{c}{E} \|F_m^E\|
    \le \frac{c}{E} \frac{4\sqrt{2\kappa}}{\sqrt{E}} \|P F_m^E P^{-1}\|  
     \le C_0 E^{-\frac{3}{2}}. 
\end{align*}
Here \(C_0 = 4c\sqrt{2\kappa} e^{\kappa^{-\frac{1}{2}} + C}\), where \(c, C > 0\) are constants from \eqref{eqn:F-Tn-norm} and \eqref{eqn:rho-CE}, independent of \(n\) and \(E\).

Hence, the deviation estimate \eqref{eqn:azuma-Qnsin} implies a deviation estimate for \(T_n^E\): for \(nE \le n^{-2\alpha}\), \(m \le n\), and \(E < E_0\),
\begin{align}
    \mathbb{P}\Big( \|T_m^E(\omega)\| \le C_0 E^{-\frac{3}{2}} \Big)
    > \mathbb{P}\Big( \sum_{i=0}^{m-1} Q_i \sin(\chi_i + \eta)
    < 2 \kappa^{-\frac{1}{2}} E^{\frac{1}{2}} n^{\frac{\alpha}{2}} m^{\frac{1}{2}} \Big)
    > 1 - e^{-n^{\alpha}}. \label{eqn:ldt-TmE}
\end{align}
For any \(1 \le j \le  n\), applying \eqref{eqn:ldt-TmE} with \(m = n - j \le n\) gives
\[
    \mathbb{P}\Big( \|T_{n-j}^E(S^j \omega)\| \le C_0 E^{-\frac{3}{2}} \Big)
    > 1 - e^{-n^{\alpha}}.
\]
By the same telescoping argument in Lemma~\ref{lem:telescope}, if
\(\|T_{n-j}^E(S^j \omega)\| \le C_0 E^{-\frac{3}{2}}\),
then for \(z = E + iT^{-1}\),
\begin{align*}
    T_{n}^z(\omega)
    = T_{n}^0 - \sum_{j=1}^{n} \frac{1}{a_j} T_{n-j}^E(S^j \omega)
    \begin{pmatrix}
        iT^{-1} & 0 \\
        0 & 0
    \end{pmatrix}
    T_{j-1}^z(\omega).
\end{align*}
which implies
\begin{align}
    \|T_{n}^z(\omega)\|
    \le C E^{-\frac{3}{2}} \sum_{j=0}^n \binom{n}{j}
    \left( \frac{C_0 E^{-\frac{3}{2}} T^{-1}}{a_-} \right)^j
    \le C_0 E^{-\frac{3}{2}} \exp\left( \frac{C_0 E^{-\frac{3}{2}} T^{-1}}{a_-} n \right).
\end{align}
Therefore, if \(E^{-\frac{3}{2}} T^{-1} n \le E^{-\frac{3}{2}} T^{-1} n \le 1\), then
\begin{align*}
    \|T_{n}^z(\omega)\| \le C_1 E^{-\frac{3}{2}}, \quad C_1 = C_0 e^{C_0 / a_-}.
\end{align*}

Hence,
\[
    \Big\{ \omega : \|T_{n}^z(\omega)\| > C_1 E^{-\frac{3}{2}} \Big\}
    \subset \bigcup_{j=1}^{n} \Big\{ \omega : \|T_{n-j}^E(S^j \omega)\| > C_0 E^{-\frac{3}{2}} \Big\},
\]
which implies \eqref{eqn:ldt-TnZ}.
\end{proof}

Now we return to prove \eqref{eqn:azuma-Qnsin}. The following is a   large-deviation bound for sums of martingale differences with bounded increments; see, e.g., \cite[Chapter 7, Theorem 7.2.1]{alon92book}.
\begin{theorem}[Azuma's inequality]
    Let \(0 = X_0, \dots, X_m\) be a martingale with \(|X_{i+1} - X_i| \le 1\) for all \(0 \le i < m\).
    Then for any \(\delta > 0\),
    \begin{align}\label{eqn:azuma}
        \mathbb{P}\big( X_m > \delta \sqrt{m} \big) < e^{-\delta^2 / 2}, \quad \text{and} \quad
        \mathbb{P}\big( |X_m| > \delta \sqrt{m} \big) < 2 e^{-\delta^2 / 2}.
    \end{align}
\end{theorem}
\begin{remark}
A martingale generalizes the concept of a sum of i.i.d.\ random variables with zero mean. If \(Y_0, Y_1, \dots\) are i.i.d.\ with \(\mathbb{E}[Y_j] = 0\), then
\[
    X_k = \sum_{j=0}^{k-1} Y_j
\]
is a martingale with respect to the natural filtration \(\mathcal{F}_k = \sigma(Y_0, \dots, Y_{k-1})\), since \(\mathbb{E}[X_{k+1} \mid \mathcal{F}_k] = X_k\).
Thus, martingales extend the idea of zero-mean i.i.d.\ sums by relaxing independence and identical distribution to a conditional mean-zero property. For the standard i.i.d.\ sum case, the large-deviation estimate is well known as Hoeffding's inequality \cite{chernoff1952measure,hoeffding1963probability}, from which Azuma's inequality \eqref{eqn:azuma} (also called the Azuma–Hoeffding inequality) was later developed by removing independence assumptions while retaining similar exponential tail behavior.
\end{remark}

Consequently, Azuma's inequality implies that
\begin{proof}[Proof of Lemma \ref{lem:azuma-Qnsin}]
Recall that, as discussed in Proposition~\ref{prop:rho-chi-Qn}, \(Q_n =O( \sqrt{E})\) (uniformly in \(n\)) depends on the random variable \(a_n\) and has zero expected value, \(\mathbb{E}[Q_n] = 0\), while \(\chi_n\) depends only on the random variables \(a_0, \dots, a_{n-1}\) and is independent of \(a_n\). Define
\begin{align*}
    X_0 = 0, \quad X_k = \sum_{i=0}^{k-1} Q_i \sin(\chi_i + \eta), \quad k = 1, \dots, n.
\end{align*}
Then \(X_1 = Q_0 \sin(\chi_0 + \eta)\) implies \(\mathbb{E}(X_1 \mid X_0) = \mathbb{E}[X_1] = 0 = X_0\), since \(\chi_0 + \eta\) is nonrandom. Moreover, for \(k = 1, \dots, n\), \(Q_k\) is independent of \(X_k\), and
\begin{align*}
    \mathbb{E}\big(X_{k+1} \mid X_k, \dots, X_0\big)
    &= \mathbb{E}\big(Q_k \sin(\chi_k + \eta) + X_k \mid X_k, \dots, X_0\big) \\
    &= \mathbb{E}[Q_k] \, \mathbb{E}\big(\sin(\chi_k + \eta) \mid X_k, \dots, X_0\big)
    + \mathbb{E}\big(X_k \mid X_k, \dots, X_0\big) = X_k.
\end{align*}
Hence, \(\{0 = X_0, \dots, X_n\}\) is a martingale. In addition, by \eqref{eqn:Qn},
\begin{align*}
    |X_{k+1} - X_k| = |Q_k \sin(\chi_k + \eta)| \le |Q_k| \le \sqrt{2\kappa^{-1}} \sqrt{E}, \quad k = 0, \dots, n.
\end{align*}

It suffices to apply Azuma's inequality to the rescaled martingale
\[
    \frac{X_k}{\sqrt{2\kappa^{-1}} \sqrt{E}}, \quad k = 0, \dots, m,
\]
for all \(1 \le m \le n\) with \(\delta = \sqrt{2} n^{\alpha/2}\). Then
\begin{align*}
    \mathbb{P}\Big( \frac{X_m}{\sqrt{2\kappa^{-1}} \sqrt{E}} > \sqrt{2} n^{\frac{\alpha}{2}} \sqrt{m} \Big)
     < \exp\Big\{-\frac{1}{2}(\sqrt{2} n^{\frac{\alpha}{2}})^2\Big\}  
    \Longleftrightarrow \mathbb{P}\Big( X_m > 2 \kappa^{-\frac{1}{2}} E^{\frac{1}{2}} n^{\frac{\alpha}{2}} m^{\frac{1}{2}} \Big)
     < e^{-n^{\alpha}},
\end{align*}
which proves \eqref{eqn:azuma-Qnsin}.
\end{proof}

%%%%%%%%%%%%%%%%%%%%%%%%%%%%%%%%%%%%%%%%%%%%

\section{Lower Bound on the Quantum Dynamics}\label{sec:lowerBDq}

Let \(\Sigma\) denote the almost-sure spectrum of \(H_\omega\) in \eqref{eqn:spe}. 
For \(z \in \mathbb{C} \setminus \Sigma\), the Green's function is defined as the kernel of the resolvent \((H_\omega - z)^{-1}\), given by
\begin{align}\label{eqn:green-def}
    G^z(n, m; \omega) = \langle \delta_n, (H_\omega - z)^{-1} \delta_m \rangle, \quad n, m \in \mathbb{Z}.
\end{align}
The following well-known identity, based on the Parseval formula and used earlier in \cite{killip03dyna}, connects the quantum transport properties of the wave packet in the one-dimensional model to the Green's function:
\begin{align*}
    \int_0^\infty e^{-t/T} \left| \langle \delta_n, e^{-itH_\omega} \delta_0 \rangle \right|^2 dt
    = \frac{1}{\pi} \int_{\mathbb{R}} \left| G^{E + \frac{i}{T}}(n, 0; \omega) \right|^2 dE.
\end{align*}
In view of \eqref{eqn:Mtq}, we obtain for any \(\omega \in \Omega\) and \(q > 0\),
\begin{align}\label{eqn:M-G}
    M_T^q = \frac{1}{\pi T} \int_{\mathbb{R}} \sum_{n \in \mathbb{Z}} |n|^q \left| G^{E + \frac{i}{T}}(n, 0; \omega) \right|^2 dE,
\end{align}
which will serve as the main tool for estimating the quantum transport exponent.

In this section, we study lower bounds on quantum dynamics. We begin by proving the averaged lower bound stated in \eqref{eqn:beta-ave}.

\begin{theorem}\label{thm:lowerBd} 
Let \(E_0, D_0 > 0\) be the constants in \eqref{eqn:NE-bound}. 
There exist constants \(c > 0\), depending explicitly on \(D_0\) and \(a_-,a_+\) in \eqref{eqn:an-bound}, such that for  \(q \ge 4\) and \(T > \max\big(\sqrt 2,\frac{1}{2E_0}\big)\), 
\begin{align}\label{eqn:lower-ave}
    \mathbb{E} M_T^q \ge c T^{\frac{q}{2} - 2}.
\end{align}
As a consequence,
\begin{align}
    \beta^-_q := \liminf_{T \to \infty} \frac{\log \mathbb{E} M_T^q}{q \log T } \ge \frac{1}{2} - \frac{2}{q}.
\end{align}   
\end{theorem}

The proof of Theorem~\ref{thm:lowerBd} relies on the following two technical lemmas. We first use these lemmas to establish the lower bound \eqref{eqn:lower-ave}, which concerns the expectation of the quantum dynamics. In the final subsection~\ref{sec:as-lower}, we also discuss the almost sure lower bound stated in \eqref{eqn:beta-as} based on the large deviation estimate \eqref{eqn:ldt-TnZ}.

\begin{lemma}\label{lem:tech-1}
  There exists a constant \(c > 0\), depending only on \(a_-,a_+\) in \eqref{eqn:an-bound}, such that for any \(q > 0\), and \(T > \sqrt 2\),
\begin{align}\label{eqn:Mtq-borel}
    \mathbb{E} M_T^q \ge c T^{\frac{q-3}{2} } \int_0^{\frac{1}{T}} \mathrm{Im} \, B_{\mathcal{N}}\left(E + \frac{i}{T}\right) dE,
\end{align}
where \(B_{\mathcal{N}}(z)\) is the Borel transform of the DOS \(d\mathcal{N}\), defined by
\begin{align}\label{eqn:borel-def}
    B_{\mathcal{N}}(z) = \int \frac{1}{z - E'} \, d\mathcal{N}(E').
\end{align}
\end{lemma}
The proof of Lemma~\ref{lem:tech-1} is deferred to the next subsection, as it involves several intermediate estimates and technical steps that require separate discussion.

\begin{lemma}\label{lem:tech-2}
Let \(B_{\mathcal{N}}(z)\) be as in Lemma~\ref{lem:tech-1}, and let \(E_0, D_0 > 0\) be the constants in \eqref{eqn:NE-bound}. For \(T > 1/(2E_0)\),
\begin{align}\label{eqn:borel-lower}
    \int_0^{\frac{1}{T}} \mathrm{Im} \, B_{\mathcal{N}}\left(E + \frac{i}{T}\right) dE \ge \arctan\left(\tfrac{1}{2}\right)D_0 T^{-1/2}.
\end{align}
\end{lemma}
The proof of Lemma~\ref{lem:tech-2} follows directly from the asymptotic behavior of the IDS given in \eqref{eqn:NE-bound} of Theorem~\ref{thm:NE-root}.
\begin{proof}
The asymptotic behavior \eqref{eqn:NE-bound} implies that there exist constants \(E_0, D_0 > 0\) such that
\begin{align*}
    \mathcal{N}([0, E]) \ge D_0 \sqrt{E}, \quad 0 < E < E_0,
\end{align*}
where we also use the fact that \(\mathcal{N}(0) = 0\).  Let \(S = [0, \frac{1}{T}]\) and \(S' = [0, \frac{1}{2T}] \subset S\). Observe that
\[
\big\{\, (E, x) \in S \times S'\, \big\} \supset \big\{\,(E, x) : x \in S', \, x \le E \le x + \tfrac{1}{2T}\, \big\}.
\]
Hence,
\begin{align*}
    \int_S \mathrm{Im} \, B_{\mathcal{N}}\left(E + \tfrac{i}{T}\right) dE 
    & = \int_S \int_{\mathbb{R}} \frac{T^{-1}}{(E - x)^2 + T^{-2}} \, d\mathcal{N}(x) \, dE \\
    & \ge \int_{S'} \left( \int_x^{x + \frac{1}{2T}} \frac{T^{-1}}{(E - x)^2 + T^{-2}} \, dE \right) d\mathcal{N}(x) \\
    & = \arctan\left(\tfrac{1}{2}\right) \cdot \mathcal{N}(S') \ge \arctan\left(\tfrac{1}{2}\right) D_0 T^{-1/2},
\end{align*}
provided that \(1/(2T) < E_0\).
\end{proof}

 \begin{proof}[Proof of Theorem~\ref{thm:lowerBd}]
     Combining \eqref{eqn:Mtq-borel} and \eqref{eqn:borel-lower}, we have 
     \begin{align*}
    \mathbb{E} M_T^q \ge c T^{\frac{q}{2} - \frac{3}{2}} \int_0^{\frac{1}{T}} \mathrm{Im} \, B_{\mathcal{N}}\left(E + \frac{i}{T}\right) dE\ge  c' T^{\frac{q}{2} - \frac{3}{2}} T^{-1/2}=c' T^{\frac{q}{2} - 2},   
     \end{align*}
     where \(c'=c\arctan\left(\tfrac{1}{2}\right) D_0\). 
 The bound holds for \(q > 0\), but we focus on the nontrivial regime where \(q \ge 4\).
 \end{proof}

%%%%%%%%%%%%%%%%%%%%%%%%%%%%%%%%%%%%%%%%%%%%%%%%%%%%%%%%%%%%%%%%

\subsection{Proof of Lemma~\ref{lem:tech-1}}

The upper bound on the transfer matrices in \eqref{eqn:Tnz-exp} implies a lower bound on the Green's function \(G^z\). Denote by 
\begin{align}\label{eqn:gn}
    g^z(n;\omega) := G^z(n, 0; \omega) = \langle \delta_n, (H_\omega - z)^{-1} \delta_0 \rangle.
\end{align}
When there is no ambiguity, we will suppress the dependence on \(\omega\) and simply write \(g^z(n)=g^z(n;\omega)\).

From the definition of \(H_\omega\) in \eqref{eqn:div-grad}, the \(n\)-th component of \((H_\omega - z) G^z\) satisfies
\begin{align}\label{eqn:h-z-g}
    \big((H_\omega - z) G^z\big)_n = -a_{n+1} g^z(n+1) + (a_{n+1} + a_n - z) g^z(n) - a_n g^z(n-1) = \delta_n(0).
\end{align}
In particular, at \(n = 0\),
\begin{align*}
    -a_1 g^z(1) + (a_1 + a_0 - z) g^z(0) - a_0 g^z(-1) = 1,
\end{align*}
which implies that for \(|z| \le 1\),
\begin{align}\label{eqn:g012}
    |g^z(1)|^2 + |g^z(0)|^2 + |g^z(-1)|^2 \ge \frac{1}{3(a_++1)^2} > 0,
\end{align}
where \(a_+ > 0\) is the constant in \eqref{eqn:an-bound}.

On the other hand, evaluating \eqref{eqn:h-z-g} at \(n \ge 1\) gives
\begin{align*}
    -a_{n+1} g^z(n+1) + (a_{n+1} + a_n - z) g^z(n) - a_n g^z(n-1) = 0,
\end{align*}
which is, in view of \eqref{eqn:u-cocycle}, 
\begin{align}\label{eqn:Tnz-n20}
    \begin{pmatrix}
        a_n g^z(n) \\
        g^z(n-1)
    \end{pmatrix}
    = T_n^z \begin{pmatrix}
        a_0 g^z(0) \\
        g^z(-1)
    \end{pmatrix}, \quad n\ge 1.
\end{align}
A direct consequence of \eqref{eqn:Tnz-exp} is that for \(z = E +  i T^{-1}\) with \(E \in (0, T^{-1})\),  
\(0 < n \le \sqrt{T}\) and \(\omega\) in a full measure set where \eqref{eqn:an-bound} holds, 
\begin{align}\label{eqn:Tnz-Cn}
    \|T_n^z(\omega)^{-1}\| = \|T_n^z(\omega)\| \le C_0 n, \quad C_0=\frac{8 a_+^2}{a_-}\exp\Big(\frac{8 a_+^2}{a_-^2}\Big), 
\end{align}
where we used the fact that \(\|A\| = \|A^{-1}\|\) for \(A \in SL(2, \mathbb{C})\).

Consequently, for the same range of parameters \(n,E,T,\omega\), and \(C_1=\frac{\max(a_+^2,1)}{\min(a_-^2,1)}\), 
\begin{align*}
    |g^z(0)|^2 + |g^z(-1)|^2 &\le C_1 \|T_n^z(\omega)^{-1}\|^2 \left( |g^z(n)|^2 + |g^z(n-1)|^2 \right)   \\
    &\le C_1C_0^2 n^2 \left( |g^z(n)|^2 + |g^z(n-1)|^2 \right),   
\end{align*}
and similarly,
\begin{align*}
    |g^z(1)|^2 + |g^z(0)|^2 \le C_1C_0^2 (n-1)^2 \left( |g^z(n)|^2 + |g^z(n-1)|^2 \right).  
\end{align*} 

Adding the above two inequalities gives 
\begin{align} 
    |g^z(n)|^2 + |g^z(n-1)|^2 \ge \frac{1}{2C_1C_0^2n^2} \left( |g^z(1)|^2 + |g^z(0)|^2 + |g^z(-1)|^2 \right).\label{eqn:gn>g0}
\end{align}
 
Multiplying both sides by \(n^q\) and summing over \(\sqrt{T}/2 < n \le \sqrt{T}\) yields a lower bound on the \(q\)-th moment of \(g^z\) for \(z = E + i/T\) with \(E \le 1/T\):
\begin{align}
    \sum_{\sqrt{T}/2 \le n \le \sqrt{T}} n^q |g^z(n)|^2 \ge & \frac{1}{4C_1C_0^2} \sum_{\sqrt{T}/2 < n \le \sqrt{T}} n^{q-2} \left( |g^z(1)|^2 + |g^z(0)|^2 + |g^z(-1)|^2 \right) \notag \\
    \ge & \frac{1}{4C_1C_0^2}\Big(\frac{1}{2}\sqrt T\Big)^{q-1} \left( |g^z(1)|^2 + |g^z(0)|^2 + |g^z(-1)|^2 \right) . \label{eqn:nq-gn-lower}
\end{align} 

Now we are ready to prove the lemma.  
\begin{proof}[Proof of Lemma~\ref{lem:tech-1}]
Substituting the lower bound on the initial data from \eqref{eqn:g012} into \eqref{eqn:nq-gn-lower} immediately yields, for some constant \(c_1 > 0\) (independent of \(E\) and \(T\)), and for \(z = E +  i T^{-1}\) with \(E \in (0, T^{-1})\), and \(\omega\) in a full measure set where \eqref{eqn:an-bound} holds, 
\begin{align}\label{eqn:327}
    \sum_{\sqrt{T}/2 \le n \le \sqrt{T}} n^q |G^z(n, 0)|^2 \ge \frac{1}{12C_1C_0^2(a_++1)^2}\Big(\frac{1}{2}\sqrt T\Big)^{q-1} := c_1 T^{\frac{q - 1}{2}}.
\end{align} 

On the other hand, recall that an important connection between the Green's function and the DOS of an ergodic operator (see, e.g., \cite[\S3.3]{aizenman2015random}) is 
\begin{align}\label{eqn:green-ids}
    \mathbb{E} \, G^z(0,0) = \mathbb{E}\left(\langle \delta_0, (z - H_\omega)^{-1} \delta_0 \rangle\right) = \int \frac{1}{E - z} \, d{\mathcal N}(E),
\end{align}
where the last term is the Borel transform \(B_{\mathcal N}(z)\) of the DOS \(d{\mathcal N}\) as in \eqref{eqn:borel-def}.

Dropping \(g^z(1)\) and \(g^z(-1)\) in \eqref{eqn:nq-gn-lower}, rewriting \(g^z(n)=G^z(n,0)\), and taking the expectation gives, by the Cauchy–Schwarz inequality, 
for \(c_2=\big[C_1C_0^22^{q+1}\big]^{-1}\), 
\begin{align}
    \mathbb{E} \sum_{\sqrt{T}/2 \le n \le \sqrt{T}} n^q |G^z(n, 0)|^2 
    \ge c_2\, T^{\frac{q - 1}{2}} \, \mathbb{E} \left( |G^z(0,0)|^2 \right) 
    &\ge c_2 T^{\frac{q - 1}{2}} \, \left| \mathbb{E} G^z(0, 0) \right|^2 \notag \\
    &\ge c_2 T^{\frac{q - 1}{2}} \, \left| \mathrm{Im} \, B_{\mathcal{N}}(z) \right|^2. \label{eqn:330}
\end{align}

Let \(c_3 =  \min\{c_1, c_2\}\), where \(c_1,c_2\) are the constants from \eqref{eqn:327} and \eqref{eqn:330}, respectively. Then for \(z = E + i/T\) with \(E \le 1/T\),
\begin{align*}
    \mathbb{E} \sum_{\sqrt{T}/2 \le n \le \sqrt{T}} n^q |G^z(n, 0)|^2 
    \ge \frac{c_3}{2} \left( 1 + \left| \mathrm{Im} \, B_{\mathcal{N}}(z) \right|^2 \right) T^{\frac{q - 1}{2}} 
    \ge c_3 \, \mathrm{Im} \, B_{\mathcal{N}}(z) \, T^{\frac{q - 1}{2}}.
\end{align*}

Using \eqref{eqn:M-G}, we obtain a lower bound for \(\mathbb{E} M_T^q\) in terms of the above partial sum:
\begin{align}\label{eqn:Mtq-partial-lower}
    \mathbb{E} M_T^q 
    \ge \frac{1}{\pi T} \int_0^{\frac{1}{T}}\sum_{\sqrt{T}/2 < n \le \sqrt{T}} |n|^q \, \mathbb{E} \left| G^{E + \frac{i}{T}}(n, 0; \omega) \right|^2 \, dE  
     \ge \frac{c_3}{\pi  }T^{\frac{q-3}{2}} \int_0^{\frac{1}{T}} \, \mathrm{Im} \, B_{\mathcal{N}}(z) \,     dE.
\end{align}
Together with \eqref{eqn:borel-lower}, this establishes Lemma~\ref{lem:tech-1}. The bound on \(T\) is imposed by the condition \(|z|\le \sqrt{2}/T \le 1\) needed for applying \eqref{eqn:g012}.
\end{proof}
%%%%%%%%%%%%%%%%%%%%%%%%%%%%%%%%%%%%%%%%%%%%%%%%%%%%%%%%%%%%%%%%%%%%%%%%%%%%%%%%%%%%%%%%%%%%%%%%%%%%%%%%%%%

%%%%%%%%%%%%%%%%%%%%%%%%%%%%%%%%%%%%%%%%%%%%%%%%%%%%%%%%%%%%

\subsection{Almost Sure Lower Bound}\label{sec:as-lower}
In the above estimate, we used the deterministic bound \eqref{eqn:Tnz-exp} on \(\|T_n^z\|\). The probabilistic version of this bound in \eqref{eqn:Tnz-exp} yields an almost-sure lower bound on the quantum transport exponent using similar arguments.

Fix \(0 < \alpha < \frac{2}{5}\) and \(T > 1\), and set \(N = \lfloor T^{\frac{2}{5} - \alpha} \rfloor\). For any \(N \le n \le 2N\) and \(E \in [T^{-2/5}, 2T^{-2/5}]\), we have
\begin{align*}
    n^{1+2\alpha}E \le (2T^{\frac{2}{5} - \alpha})^{1+2\alpha} \cdot 2T^{-2/5} = 2^{2+2\alpha} T^{-\frac{1}{5}\alpha-2\alpha^2} < 1,
\end{align*}
and
\begin{align*}
    n\le 2  T^{\frac{2}{5} - \alpha}=2\big(T^{-\frac{2}{5}})^{\frac{3}{2} }  T^{1 - \alpha}\le 2T^{-\alpha} E^{\frac{3}{2}}T\le E^{\frac{3}{2}}T,
\end{align*}
provided \(T>T_0(\alpha)\). 
Hence, the assumptions of Theorem~\ref{thm:ldt-Tn-norm} are satisfied. By \eqref{eqn:ldt-TnZ}, we obtain for \(N \le n \le 2N\),
\begin{align*}
    \mathbb{P}\left( \left\{ \omega : \|T_n^z(\omega)\| \text{ or } \|T_{n-1}^z(S\omega)\| > C E^{-\frac{3}{2}} \right\} \right) \le 2ne^{-n^\alpha} \le 4Ne^{-N^\alpha},
\end{align*}
where \(S\) is the measure-preserving shift defined in \eqref{eqn:shift}.

Define
\begin{align}\label{eqn:ome-TN-union}
    \Omega_N = \bigcap_{n = N}^{2N} \left\{ \omega : \|T_n^z(\omega)\|, \|T_{n-1}^z(S\omega)\| \le C E^{-\frac{3}{2}} \right\}.
\end{align}
Then, by the union bound,
\begin{align*} 
    \mathbb{P}(\Omega_N^c) \le \sum_{n = N}^{2N} \mathbb{P}\left( \left\{ \omega : \|T_n^z(\omega)\| \text{ or } \|T_{n-1}^z(S\omega)\| > C E^{-\frac{3}{2}} \right\} \right) \le 4N^2 e^{-N^\alpha}.
\end{align*}

For \(z = E + \frac{i}{T}\) with \(E \in [T^{-2/5}, 2T^{-2/5}]\) and \(N = \lfloor T^{\frac{2}{5} - \alpha} \rfloor\), a direct consequence of \eqref{eqn:ome-TN-union} is that 
for \(\omega \in \Omega_N\) and \(N \le n \le 2N\),
\begin{align}\label{eqn:Tnz-as}
    \|T_n^z(\omega)^{-1}\| = \|T_n^z(\omega)\| \le C E^{-\frac{3}{2}} \le C T^{3/5}.
\end{align}

We now use the probabilistic bound in \eqref{eqn:Tnz-as} to replace \eqref{eqn:Tnz-Cn} in estimating the lower bound. Combining \eqref{eqn:Tnz-as} with \eqref{eqn:Tnz-n20}, we obtain, in the same way as \eqref{eqn:gn>g0}, 
\begin{align*}
    |g^z(n)|^2 + |g^z(n-1)|^2 \ge \frac{1}{ 2C_1 C^2 T^{6/5}} \left( |g^z(1)|^2 + |g^z(0)|^2 + |g^z(-1)|^2 \right) \ge c_1T^{-\frac{6}{5}}.
\end{align*}
In the last inequality we used \eqref{eqn:g012}. Here \(C_1 > 0\) is the same constant in \eqref{eqn:gn>g0} and \(c_1=\big[{6C_1 C^2(a_++1)^2 }\big]^{-1}\). 

Multiplying both sides by \(n^q\) and summing over \(N \le n \le 2N\) yields a lower bound on the \(q\)-th moment for \(z = E + i/T\) with \(E \in [T^{-2/5}, 2T^{-2/5}]\), in the same way as \eqref{eqn:gn>g0} and \eqref{eqn:nq-gn-lower}: 
\begin{align}
    \sum_{N \le n \le 2N} n^q |G^z(n, 0)|^2 
     \ge  \frac{c_1}{2}\sum_{N \le n \le 2N} n^q T^{-\frac{6}{5}} 
     \ge \frac{c_1}{2 }T^{-\frac{6}{5}}  N^{q+1}\ge \frac{c_1}{2^{q+2}}T^{\frac{2}{5}q-\frac{4}{5}-\alpha(q+1)} , \label{eqn:348-as}
\end{align} 
where we used \(N = \lfloor T^{\frac{2}{5} - \alpha} \rfloor \ge \frac{1}{2} T^{\frac{2}{5} - \alpha}\) provided \(0<\alpha<\frac{2}{5}\) and \(T\ge 1\).

In this case, since \(\Omega_N\) in \eqref{eqn:ome-TN-union} does not have full probability, we cannot take the expectation and apply \eqref{eqn:green-ids} to obtain an analogue of \eqref{eqn:330}. Instead, we estimate directly using \eqref{eqn:M-G}. For \(\omega \in \Omega_N\),
\begin{align*}
    M_T^q 
     \ge \frac{1}{\pi T} \int_{T^{-2/5}}^{2T^{-2/5}} \sum_{N \le n \le 2N} |n|^q \left| G^{E + \frac{i}{T}}(n, 0; \omega) \right|^2 \, dE 
    &\ge \frac{1}{\pi T} \cdot \frac{c_1}{2^{q+2}}T^{\frac{2}{5}q-\frac{4}{5}-\alpha(q+1)} \cdot T^{-\frac{2}{5}}   \\
    &\ge \frac{c_1}{ \pi 2^{q+2}} T^{\frac{2q}{5} - \frac{11}{5} - \alpha(q + 1)}.
\end{align*}

Since \(\mathbb{P}(\Omega_N^c) \le 4N^2 e^{-N^\alpha}\), the Borel–Cantelli lemma implies that almost surely, there exists \(T_\omega\) such that for all \(T > T_\omega\), we have
\[
M_T^q \ge \frac{c_1}{ \pi 2^{q+2}} T^{\frac{2q}{5} - \frac{11}{5} - \alpha(q + 1)}.
\]
Hence, almost surely,
\begin{align*}
    \liminf_{T \to \infty} \frac{\log M_T^q}{q\log T } \ge \frac{2}{5} - \frac{11}{5q} - \alpha\left(1 + \frac{1}{q}\right).
\end{align*}
Since \(\alpha > 0\) can be chosen arbitrarily small, it follows that almost surely,
\[
\beta_q^{-, \mathrm{a.s.}} := \liminf_{T \to \infty} \frac{\log M_T^q}{q \log T } \ge \frac{2}{5} - \frac{11}{5q}.
\]

%%%%%%%%%%%%%%%%%%%%%%%%%%%%%%%%%%%%%%%%%%%%%%%%%%%%%%%%%%%%%%%%%%%%%%%%%%%%%%%%%%%%%%%%%%%%%%%%%%%%%%%%%%%%%%%%%%%%%%%%%%%%%%%%%%%%%%%%%%%%%%%%%%%%%%%%%%%%%%%%%%%%%%

\section{Upper Bound on the Quantum Dynamics}\label{sec:upper}

In this section, we establish the upper bound \eqref{eqn:beta+ave}. A more precise estimate, expressed in terms of the expectation value of \(\,M_T^q\), is stated below. 
\begin{theorem}\label{thm:upper}
For \(q \ge 1\) and \(0 < \alpha < \frac{1}{4}\), there exist constants \(C, T_0 > 0\), depending on \(\alpha\) and \(q\), such that for all \(T > T_0\),
\begin{align}\label{eqn:Mtq-upper}
    \mathbb{E} M_T^q \le C T^{q - \frac{1}{5} + 5\alpha q}.
\end{align}
Consequently,
\begin{align}\label{eqn:upper}
    \beta^+_q := \limsup_{T \to \infty} \frac{\log \mathbb{E} M_T^q}{q\log T } \le 1 - \frac{1}{5q}.
\end{align}
\end{theorem}
\begin{remark}
The restriction \(\alpha < \frac{1}{4}\) is not essential, provided that \eqref{eqn:Mtq-upper} holds for sufficiently small \(\alpha > 0\), since we ultimately take the limit \(\alpha \to 0^+\) after letting \(T \to \infty\). The condition \(q \ge 1\) arises from the hyperbolic energy region \(E \to 0^+\); see Appendix~\ref{sec:Mtq-negativeE}. Because the upper bound \(T^{q - \frac{1}{5}}\) is far from optimal, we do not attempt to refine the threshold \(q \ge q_0\) where a transition might occur.
\end{remark}

The main ingredient remains the formula \eqref{eqn:M-G}, which reduces the analysis of \(\,M_T^q\) to estimating \(G^z(n,0;\omega)\). Unlike the lower bound, where it suffices to focus on the ``most delocalized'' energy-frequency regime of \((E,n)\) (see, for instance, the partial sum lower bound in \eqref{eqn:Mtq-partial-lower}), obtaining an upper bound for \(\,M_T^q\) requires examining contributions from all possible \(n\) and \(E\) in the Green's function \(G^{E+i/T}(n,0;\omega)\) appearing in \eqref{eqn:M-G}. 

For convenience, we suppress the dependence on \(\omega\) and write \(G^z(n,0) = G^z(n,0;\omega)\) whenever no ambiguity arises. To facilitate regrouping the sum over \(n \in \mathbb{Z}\) and the integral over \(E\), for \(0 \le \alpha_0 < \alpha_1 \le \infty\) and an interval (or union of intervals) \(I \subseteq \R\), we define
\begin{align}\label{eqn:M-alpha}
    M_T^{q, \alpha_0, \alpha_1}(I) = \frac{1}{\pi T} \int_{I}\,  \sum_{T^{\alpha_0} \le |n| \le T^{\alpha_1}} |n|^q    \, |G^{E + i/T}(n, 0)|^2\, dE.
\end{align}
For simplicity, we denote the complete sum over all \(n \in \Z\) by
\begin{align}\label{eqn:fullM}
    M_T^q(I) = M_T^{q, 0, \infty}(I) = \frac{1}{\pi T} \int_{I}\,  \sum_{n \in \Z} |n|^q    \, |G^{E + i/T}(n, 0)|^2\, dE.
\end{align}

Decay estimates for the Green's function outside the spectrum, such as the classical Combes--Thomas bound \cite{combes} and its modern extension to general graph operators \cite{aizenman2015random}, together with sub-exponential decay near the spectrum under uniformly positive Lyapunov exponents, lead to strong bounds on quantum dynamics. These observations are formalized in the following lemma.

\begin{lemma}\label{lem:CT-logT}
Let \(\sigma_1 = [-E_1, E_1]\) and \(\sigma_2 = \sigma_1 \backslash [-E_0, E_0]\), where \(E_1 = 8ea_+\) and \(0 < E_0 < 1\) is as in \eqref{eqn:LE-bound}. For any \(\alpha, q > 0\), there exist \(T_0, C_1 > 0\) such that for \(T > T_0\),
\begin{align}\label{eqn:Mtq-CT}
    M_T^{q}(\sigma_1^c) + M_T^{q, 1+\alpha, \infty}(\sigma_1) \le C_1.
\end{align}
And for any \(q > 0\), there exists \(T_1> 0\) such that for \(T > T_1\),
\begin{align}\label{eqn:Mtq-logT}
     M_T^{q}(\sigma_2) \le 2 (\log T)^{3q}.
\end{align}
\end{lemma}

\begin{remark}

The boundedness in \eqref{eqn:Mtq-CT} follows from the Combes--Thomas estimate for \(z\) strictly away from the spectrum, where \(\mathrm{dist}(z, \sigma(H_\omega)) \gtrsim |E| + 1/T\), and the Green's function exhibits decay both in \(|n|\) and \(|E|\). See the shaded green region in Figure~\ref{fig:Mtq-split1}.  

When \({\rm Re} z\) is near the spectrum but remains strictly separated from the critical energy \(E = 0\), the uniform positivity of the Lyapunov exponent implies that, with high probability, one has a bound \(\|T_n^z\| \approx e^{\gamma \sqrt{|n|}}\) for some uniform constant \(\gamma > 0\). The logarithmic growth in \eqref{eqn:Mtq-logT} then arises from frequencies satisfying \(|n| \lesssim (\log T)^3\). In fact, the argument applies for \(|n| \lesssim (\log T)^\beta\) with any \(\beta > 2\), yielding an alternative bound of order \((\log T)^{\beta q}\). For simplicity, we choose \(\beta = 2\). See the shaded blue region in Figure~\ref{fig:Mtq-split1}. 

Both results were established in \cite{jitomirskaya2007upper} for general Jacobi operators and can be adapted to our div-grad model after carefully adjusting the frequency and energy regions. For the reader's convenience, we provide a direct proof for the region considered here for the div-grad model in Appendix~\ref{sec:CT-logT}.

\end{remark}

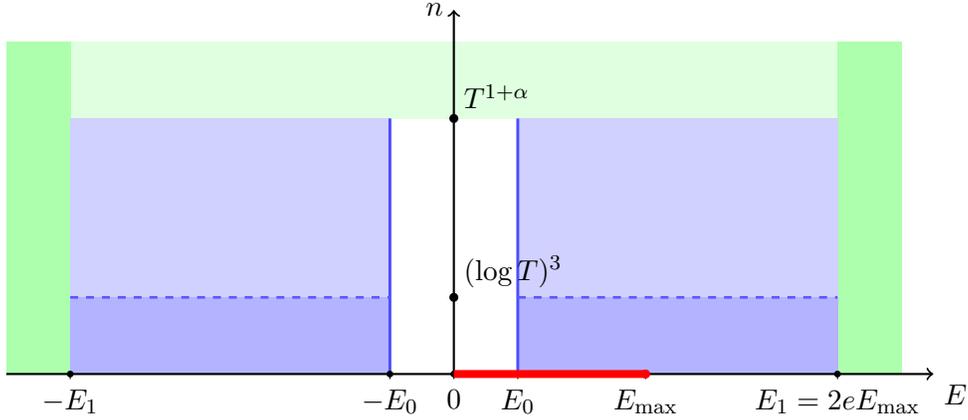
\begin{figure}[ht]
\centering
% Requires only: \usepackage{tikz}
\begin{tikzpicture}[scale=0.85]

% ---------- Tunable parameters ----------
\def\xmin{-7.0}
\def\xmax{7.0}
\def\ymax{5.2}

\def\yT{4.0}        % height for T^{1+\alpha}
\def\ylogT{1.2}     % height for (\log T)^3

\def\xzero{0.0}         % 0
\def\xEzero{1}        % E_0
\def\xEmax{3.0}         % E_{\max}
\def\xEone{6.0}     % E_1

% ======================================================================
% FILLED REGIONS ONLY (no slash lines)
% ======================================================================

% Green: Top (n \ge T^{1+\alpha})
\fill[green!12] (\xmin,\yT) rectangle (\xmax,\ymax);

% Green: Left (E \le -E_1)
\fill[green!32] (\xmin,0) rectangle (-\xEone,\ymax);
% Green: Right (E \ge E_1)
\fill[green!32] (\xEone,0) rectangle (\xmax,\ymax);

% Blue: Left block (-E_1 \le E \le -E_0)
% bottom (0 .. (\log T)^3) slightly darker
\fill[blue!30] (-\xEone,0) rectangle (-\xEzero,\ylogT);
% top ((\log T)^3 .. T^{1+\alpha}) lighter
\fill[blue!18] (-\xEone,\ylogT) rectangle (-\xEzero,\yT);
% boundary line
\draw[blue!70, line width=1.1pt] (-\xEzero,0) -- (-\xEzero,\yT);

% Blue: Right block (E_0 \le E \le E_1)
% bottom (0 .. (\log T)^3) slightly darker
\fill[blue!30] (\xEzero,0) rectangle (\xEone,\ylogT);
% top ((\log T)^3 .. T^{1+\alpha}) lighter
\fill[blue!18] (\xEzero,\ylogT) rectangle (\xEone,\yT);
% boundary line
\draw[blue!70, line width=1.1pt] (\xEzero,0) -- (\xEzero,\yT);

% Dashed blue line at n = (\log T)^3 (across both blue blocks)
\draw[dashed,blue!65, line width=1pt] (-\xEone,\ylogT) -- (-\xEzero,\ylogT);
\draw[dashed,blue!65, line width=1pt] (\xEzero,\ylogT) -- (\xEone,\ylogT);

% Ticks and labels on E-axis
\fill (-\xEone,0) circle (1.5pt);
\fill (-\xEzero,0) circle (1.5pt);
\fill (\xzero,0) circle (1.5pt);
\fill (\xEzero,0) circle (1.5pt);
\fill (\xEmax,0) circle (1.5pt);
\fill (\xEone,0) circle (1.5pt);

\draw (-\xEone,0) -- ++(0,-0.08) node[below] {$-E_1$};
\draw (-\xEzero,0) -- ++(0,-0.08) node[below] {$-E_{0}$};
\draw (\xzero,0) -- ++(0,-0.08) node[below] {$0$};
\draw (\xEzero,0) -- ++(0,-0.08) node[below] {$E_{0}$};   % corrected label
\draw (\xEmax,0) -- ++(0,-0.08) node[below] {$E_{\max}$};
\draw (\xEone,0) -- ++(0,-0.08) node[below] {$E_1=2eE_{\max}$};

% n-axis labels
\fill (0,\yT) circle (2pt);
\node[above right] at (0,\yT) {$T^{1+\alpha}$};
\fill (0,\ylogT) circle (2pt);
\node[ above right] at (0,\ylogT) {$(\log T)^{3}$};

% ---------- Axes (draw LAST so they appear above fills) ----------
\draw[->,thick] (\xmin,0) -- ({\xmax+0.5},0) node[below right] {$E$};
\draw[->,thick] (0,0) -- (0,{\ymax+0.5}) node[left] {$n$};

% Red bar on E-axis: from 0 to E_{\max}
\draw[red, line width=3pt] (\xzero,0) -- (\xEmax,0);
\fill[red] (\xEmax,0) circle (2pt); % mark right end

\end{tikzpicture}
\caption{Visualization of the regions contributing to the bounds in Lemma~\ref{lem:CT-logT}. The green shaded area corresponds to energies far from the spectrum (\(|E|\ge E_1\)) or large frequencies (\(n \ge T^{1+\alpha}\)), where the Combes--Thomas estimate ensures exponential decay of the Green's function. The blue shaded area represents energies near the spectrum but away from the critical point \(E = 0\), with frequencies up to \(T^{1+\alpha}\); the dashed line at \(n = (\log T)^3\) indicates the scale relevant for the logarithmic bound in \eqref{eqn:Mtq-logT}. The red segment on the \(E\)-axis marks the spectrum of the operator as given in \eqref{eqn:spe}, with its lower endpoint at \(E = 0\), the only critical energy where the Lyapunov exponent vanishes, and its upper endpoint at \(E_{\max} = 4a_+\).}
\label{fig:Mtq-split1}
\end{figure}

%%%%%%%%%%%%%%%%%%%%%%%%%%%%%%%%%%%%%%%%%%%%%%%%%%%%%%%%%%%%%%%%%%%%%%%%%%%%%%%%%%%%%%%%%%%%%%%%%%%%%%%%%%%%%%%%%%%%%%%%%%%%%%%%%%%%%%%%%%%%%%%%%%%%%%%%%%%%%%%%%%%%%%%%
The remaining contribution in \eqref{eqn:fullM} is \(M_T^{q, 0, 1+\alpha}(-E_0, E_0)\) for \(0 < E_0 < 1\), which corresponds to the region near the critical energy \(E = 0\). See the white region in Figure~\ref{fig:Mtq-split1}. It suffices to restrict the last term in \eqref{eqn:fullM} to the right half-interval \((0, E_0)\). The treatment of the left half \((-E_0, 0)\) is analogous but considerably simpler, as it lies outside the spectrum. One can show that for \(q\ge 1 \), \( \E M_T^{q,0,1+\alpha}(-E_0, 0)\le  C ' T^{q-\frac{1}{2}+\alpha q}\); see Appendix~\ref{sec:Mtq-negativeE}.

For \(\alpha > 0\), we split the interval \((0, E_0)\) into three parts: a low-energy region, a mild-energy region, and a high-energy region; see Figure~\ref{fig:E3parts}:
\begin{align}
    M_T^{q,0,1+\alpha}(0, E_0) = M_T^{q, 0, 1+\alpha}(0, E_L) + M_T^{q, 0, 1+\alpha}(E_L, E_R) + M_T^{q, 0, 1+\alpha}(E_R, E_0), \label{eqn:Mtq-split}
\end{align}
where
\begin{align}\label{eqn:ELR}
    E_L = T^{-\frac{2}{5}}, \quad \text{and} \quad E_R = T^{-\alpha}.
\end{align}

%%%%%%%%%%%%%%%%%%%%%%%%%%%%%%%%%%%%%%%%%%%%%%%%%%%%%%%%%%%%%%%%%%%%%%%%%%%%%%%%%%%%%%%%%%%%%%%%%%%

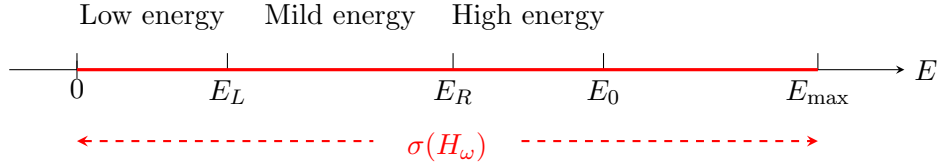
\begin{figure}[ht]
  \centering
  % Requires in preamble: \usepackage{tikz} \usetikzlibrary{calc}
  \begin{tikzpicture}[>=stealth]

    % ==== Tunable parameters ====
    \def\Ezero{7.0}        % position of E_0 on the x-axis
    \def\EL{2}           % position of E_L on the x-axis
    \def\ER{5}           % position of E_L on the x-axis
    \def\axisY{0}          % y-level of the horizontal energy axis

    % Spectrum range (red bar) — starts at 0 and extends a bit to the right:
    \def\specLeft{0.0}     % start of sigma(H_omega) exactly at 0
    \def\specRight{9.85}   % end of sigma(H_omega), extended further right

    % Axis extents (black axis) — extend accordingly:
    \def\leftEnd{-0.9}     % left end of axis
    \def\rightEnd{11.0}     % right end of axis (keep longer than red bar)

    % Visual sizes:
    \def\barH{0.22}        % height of tiny vertical bars above axis at E_L and E_0
    \def\tickH{0.12}       % half-height of the small tick at 0
    \def\labelY{0.70}      % y-level for “Low energy” and “Mild energy” labels

    % Position for the dashed <--- sigma(H) ---> line below the axis:
    \def\belowY{-0.95}     % vertical level for the dashed arrows + label below axis

    % ==== Axis (black) ====
    \draw[->] (\leftEnd,\axisY) -- (\rightEnd,\axisY) node[right] {$E$};

    % ==== Tiny vertical bars above the axis at E_L and E_0 ====
    \draw[black] (0,\axisY) -- (0,\axisY+\barH);
    \draw[black] (\EL,\axisY) -- (\EL,\axisY+\barH);
    \draw[black] (\ER,\axisY) -- (\ER,\axisY+\barH);
    \draw[black] (\Ezero,\axisY) -- (\Ezero,\axisY+\barH);
    \draw[black] (\specRight,\axisY) -- (\specRight,\axisY+\barH);

    % ==== Small central tick at 0 and numeric labels ====
    \draw[black] (0,\axisY-\tickH) -- (0,\axisY+\tickH);
    \node[below] at (0,\axisY) {$0$};
    \node[below] at (\EL,\axisY) {$E_L$};
    \node[below] at (\ER,\axisY) {$E_R$};
    \node[below] at (\Ezero,\axisY) {$E_0$};
    \node[below] at (\specRight,\axisY) {$E_{\max}$};

    % ==== Red spectrum line sigma(H_omega), extended right ====
    \draw[red,very thick] (\specLeft,\axisY) -- (\specRight,\axisY);

    % ==== Interval labels (placed at midpoints; use slight extension via specRight only) ====
    \node[black] at ({0.5*(\specLeft+\EL)},\labelY) {Low energy};
    \node[black] at ({0.5*(\EL+\ER)},\labelY) {Mild energy};
     \node[black] at ({0.5*(\ER+\Ezero)},\labelY) {High energy};

    % ==== Single dashed <---  sigma(H_omega)  ---> below the axis ====
    % Range must match the red bar: from 0 to specRight
    \draw[red,dashed,<-,thick] (\specLeft,\belowY) -- ({0.4*(\specLeft+\specRight)},\belowY);
    \draw[red,dashed,->,thick] ({0.6*(\specLeft+\specRight)},\belowY) -- (\specRight,\belowY);
    \node[red] at ({0.5*(\specLeft+\specRight)},\belowY-0.08) {$\sigma(H_\omega)$};

  \end{tikzpicture}
 \caption{Partition of the interval \((0, E_0)\) into three subregions: low-, mild-, and high-energy, separated by \(E_L=T^{-2/5}\), \(E_R=T^{-\alpha}\), and \(E_0\). The full spectrum \(\sigma(H_\omega)\), ranging from \(0\) to \(E_{\max} = 4a_+\), is shown in red.}
  \label{fig:E3parts}
\end{figure}
 
The contribution of each component to quantum transport decreases as the energy moves away from the critical point \(E = 0\). The low-energy part is the dominant contributor to the upper bound in \eqref{eqn:Mtq-upper} and is estimated as follows:
\begin{lemma}\label{lem:E-low}
For any \(q > 0\) and \(\alpha > 0\), there exists a constant \(C_2 > 0\) such that for all \(T \ge 1\),
\begin{align}\label{eqn:lowE}
    \mathbb{E} M_T^{q, 0, 1+\alpha}(0, E_L) \le C_2 T^{q - \frac{1}{5} + \alpha q}.
\end{align}
\end{lemma}

\begin{remark}
The estimate for the low-energy part is far from optimal and represents the main opportunity for improvement in order to make the transport behavior closer to diffusive.
\end{remark}

The mild-energy part in \eqref{eqn:Mtq-split} admits the following upper bound:
\begin{lemma}\label{lem:E-mild}
For any \(q \ge \frac{1}{2}\) and \(0 < \alpha < \frac{1}{4}\), there exist constants \(C_3 = C_3(q, \alpha)\), \(C_4 = C_4(q, \alpha)\), and \(T_2 = T_2(q, \alpha)\) such that for all \(T \ge T_2\),
\begin{align}\label{eqn:mildE}
    \mathbb{E} M_T^{q, 0, 1+\alpha}(E_L, E_R) \le C_3 + C_4 T^{\frac{2}{5}(q - \frac{1}{2}) + 5q\alpha}.
\end{align}
\end{lemma}

The high-energy part is controlled by the following estimate:
\begin{lemma}\label{lem:E-high}
For any \(q > 0\) and \(\alpha > 0\), there exist constants \(C_5 = C_5(q, \alpha)\) and \(T_3 = T_3(q, \alpha)\) such that for all \(T \ge T_3\),
\begin{align}\label{eqn:highE}
    \mathbb{E} M_T^{q,0,1+\alpha}(E_R, E_0) \le C_5 +  T^{4q\alpha}.
\end{align}
\end{lemma}

The proofs of Lemmas~\ref{lem:E-low}, \ref{lem:E-mild}, and \ref{lem:E-high} are provided in the next three subsections. We first combine these results with Lemma~\ref{lem:CT-logT} to complete the proof of Theorem~\ref{thm:upper}.

\begin{proof}[Proof of Theorem~\ref{thm:upper}]
Combining \eqref{eqn:Mtq-CT}, \eqref{eqn:Mtq-logT}, \eqref{eqn:lowE}, \eqref{eqn:mildE}, and Appendix~\ref{sec:Mtq-negativeE}, we obtain that for \(q \ge 1\) and \(T > \max(3, T_0, T_1, T_2, T_3)\),
\begin{align*}
    \mathbb{E} M_T^q 
    \le C_1 + 2 (\log T)^{3q} + C_2 T^{q - \frac{1}{5} + \alpha q} + C_3 + C_4 T^{\frac{2}{5}(q - \frac{1}{2}) + 5q\alpha} + C_5 + T^{4q\alpha}.
\end{align*}
Clearly, \(q - \frac{1}{5} > \frac{2}{5}(q - \frac{1}{2})\) for any \(q > 0\). If, in addition, we require \(T > T_4(q,\alpha)\) so that the constant and logarithmic terms are bounded by \(T^{\alpha q}\), then  
\begin{align*}
    \mathbb{E} M_T^q  
    &\le C_6 T^{q - \frac{1}{5} + 5\alpha q},
\end{align*}
for some constant \(C_6\) depending on \(q\) and \(\alpha\).

Therefore, for any \(q \ge 1\) and \(0 < \alpha < \frac{1}{4}\),
\begin{align*}
    \limsup_{T \to \infty} \frac{\log \mathbb{E} M_T^q}{\log T} \le q - \frac{1}{5} + 5\alpha q.
\end{align*}
Taking \(\alpha \searrow 0\) and dividing both sides by \(q\) completes the proof of Theorem~\ref{thm:upper}.
\end{proof}

%%%%%%%%%%%%%%%%%%%%%%%%%%%%%%%%%%%%%%%%%%%%%%%%%%%%%%%%%%%%%%%%%%%%%%%%%%%%%%%%%%%%%%%%%%%%%%%%%%%%%%%%%%%%%%%%%%%%%%%%%%%%%%%%%%%%%%%%%%%%%%%%%%%%%%%%%%%%%
\subsection{Low-Energy Regime}

The following technical lemma relates the Green's function to the imaginary part of the Borel transform of the integrated density of states (IDS). It will be repeatedly used to estimate contributions from low frequencies (\(n\)) and low energies (\(E\)).
\begin{lemma}\label{lem:GImB}
Let \(B_{\mathcal{N}}\) be as in \eqref{eqn:green-ids}. For any \(E \in \mathbb{R}\), \(T > 0\), and \(N \ge 1\), one has
\begin{align}\label{eqn:GImB}
    \sum_{1 \le |n| \le N} \frac{|n|^q}{\pi T} \, \mathbb{E} |G^z(n, 0)|^2 \le \frac{N^q}{\pi} \, \mathrm{Im}\, B_{\mathcal{N}}(z), \quad z = E + \frac{i}{T},
\end{align}
and
\begin{align}\label{eqn:GImB-intFull}
    \int_{\mathbb{R}} \sum_{1 \le |n| \le N} \frac{|n|^q}{\pi T} \, \mathbb{E} |G^z(n, 0)|^2 \, dE \le N^q.
\end{align}
Furthermore, let \(D_1 > 0\) and \(0 < E_0 < 1\) be as in \eqref{eqn:NE-bound}. Then there exists \(C > 0\), depending on \(D_1\) and \(E_0\), such that for any finite \(T >0\) and \( E_2>0\),
\begin{align}\label{eqn:GImB-int}
    \int_{-E_2}^{E_2} \sum_{1 \le |n| \le N} \frac{|n|^q}{\pi T} \, \mathbb{E} |G^z(n, 0)|^2 \, dE \le 2C N^q \sqrt{E_2}.
\end{align}
\end{lemma}

\begin{remark}
The first two estimates, \eqref{eqn:GImB} and \eqref{eqn:GImB-intFull}, were established in \cite[Lemma~5]{jitomirskaya2007upper}. The bound in \eqref{eqn:GImB-int} extends \cite[Proposition~3]{jitomirskaya2007upper}, where Lipschitz continuity of the IDS near the critical energy was proved for the random dimer model, to the div-grad model, where the IDS exhibits a square-root singularity as in \eqref{eqn:NE-bound}. The proof follows from a direct computation of the integral of the Borel transform. For completeness, we include the proofs for the general case in Appendix~\ref{sec:borel}.
\end{remark}

\begin{proof}[Proof of Lemma~\ref{lem:E-low}]
Let \(E_L = T^{-\frac{2}{5}}\) be as in \eqref{eqn:ELR}.   By \eqref{eqn:GImB-int},  for \(T \ge 1\)
\begin{align}
    \mathbb{E} M_T^{q, 0, 1+\alpha}(0, E_L) 
     = \int_0^{E_L} \left( \sum_{|n| \le T^{1+\alpha}} \frac{n^q}{\pi T} \, \mathbb{E} |G^{E + i/T}(n, 0)|^2 \right) dE
    &\le 2C T^{q(1+\alpha)} \sqrt{E_L} \notag \\
    &= 2C T^{q - \frac{1}{5} + \alpha q}. \label{eqn:Mtq-low}
\end{align}
\end{proof}

%%%%%%%%%%%%%%%%%%%%%%%%%%%%%%%%%%%%%%%%%%%%%%%%%%%%%%
\subsection{Mild-Energy Regime}

For \(\alpha > 0\), let \(E_L = T^{-\frac{2}{5}}\) and \(E_R = T^{-\alpha}\) be as in \eqref{eqn:ELR}. We now estimate the second term in \eqref{eqn:Mtq-split}, which pertains to mild energies between \(E_L\) and \(E_R\). This contribution is further divided into a low-frequency component (\(J_1\)) and a high-frequency component (\(J_2\)), as illustrated in the central portion of Figure~\ref{fig:J012}:
\begin{align*}
 \mathbb{E} M_T^{q, 0, 1+\alpha}(E_L, E_R) \le 
  &\frac{1}{\pi T} \int_{E_L}^{E_R} \sum_{0 \le |n| \le E^{-1} T^{5\alpha}} |n|^q \, \mathbb{E}|G^{E + i/T}(n, 0)|^2 \, dE \quad (:= J_1) \\
  &+ \frac{1}{\pi T} \int_{E_L}^{E_R} \sum_{E^{-1} T^{5\alpha} \le |n| \le T^{1+\alpha}} |n|^q \, \mathbb{E}|G^{E + i/T}(n, 0)|^2 \, dE \quad (:= J_2).
\end{align*} 

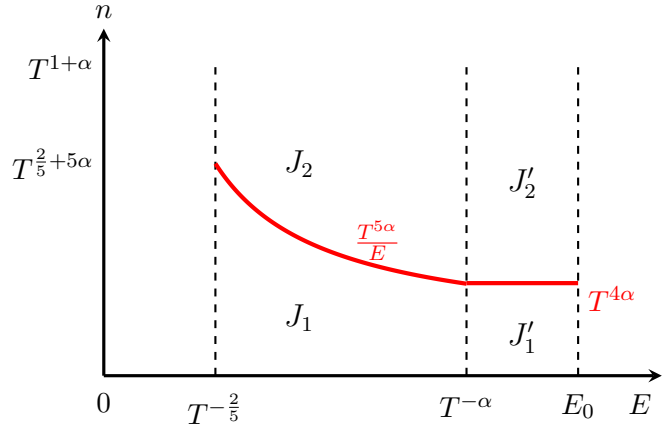
\begin{figure}[H]
  \centering
  \begin{tikzpicture}
    \begin{axis}[
      width=9cm, height=6.2cm,
      axis lines=left,
      xmin=0, xmax=10,
      ymin=0, ymax=9,
      tick style={opacity=0},
      xticklabels={},
      yticklabels={},
      axis line style={very thick},
      clip=false,  % ensure below-axis labels remain visible
    ]

      % --- Horizontal positions for special E-values ---
      \def\xET{2}       % E_L (corresponds to T^{-2/5}, but we only label E_L)
      \def\xTalpha{6.5}   % T^{-\alpha}
      \def\xEzero{8.5}    % E_0

      % --- Vertical positions for n-axis text labels (no lines drawn) ---
      \def\yTop{8}                      % T^{1+\alpha}
      \def\yMid{5.5}                      % T^{2/5 + 5\alpha}

      % --- Dashed vertical guides ---
      \draw[dashed, thick] (axis cs:\xET,0)     -- (axis cs:\xET,8);
      \draw[dashed, thick] (axis cs:\xTalpha,0) -- (axis cs:\xTalpha,8);
      \draw[dashed, thick] (axis cs:\xEzero,0)  -- (axis cs:\xEzero,8);

      % --- E-axis labels (clearly below the axis) ---
      \node[anchor=north, yshift=-3pt]           at (axis cs:0,0)       {$0$};
      \node[anchor=north, yshift=-3pt]           at (axis cs:\xET,0)    {$ T^{-\frac{2}{5}}$};
      \node[anchor=north, yshift=-3pt]           at (axis cs:\xTalpha,0){$ T^{-\alpha}$};
      \node[anchor=north, yshift=-3pt]           at (axis cs:\xEzero,0) {$E_0$};
      \node[anchor=north east, yshift=-3pt]      at (axis cs:10,0)      {$E$};

      % --- n-axis text labels (left side; no horizontal guide lines) ---
      \node[left] at (axis cs:0,\yTop) {$T^{1+\alpha}$};
      \node[left] at (axis cs:0,\yMid) {$T^{\frac{2}{5}+5\alpha}$};
       \node[above] at (axis cs:0,{\yTop+1}) {$n$};

      % --- Region labels ---
      \node[font=\large] at (axis cs:3.5,5.5) {$J_2$};   % upper middle band
      \node[font=\large] at (axis cs:3.5,1.5) {$J_1$};   % lower middle band (moved down)
      \node[font=\large] at (axis cs:7.5,5.0) {$J_2'$};  % right top band
       \node[font=\large] at (axis cs:7.5,1) {$J_1'$};  % right bot band

      % --- Red curve: n(E) = T^{5\alpha}/E, ONLY between E_L and T^{-\alpha} ---
      \def\A{9.0}  % visual scale for T^{5\alpha}; adjust to raise/lower the curve
      \addplot[
        domain=\xET:\xTalpha,
        samples=200,
        ultra thick,
        color=red
      ] {\A/x+1};
\node[color=red,left] at (axis cs:5.5,3.5) {$\frac{T^{5\alpha}}{E} $};
\draw[ultra thick,
        color=red] (axis cs:\xTalpha,2.4)     -- (axis cs:\xEzero,2.4);
        \node[color=red,right] at (axis cs:\xEzero,2) {$T^{4\alpha} $};
    \end{axis}
  \end{tikzpicture}
  \caption{Partition of the energy-frequency axis into subregions, illustrating the decomposition of contributions into \(J_1, J_2\) for mild energies and \(J_1', J_2'\) for high energies. The curve \(n(E) = T^{5\alpha}/E\) is shown for \(E \in [E_L, E_R]\), along with the cutoff \(n = T^{4\alpha}\) for the high-energy regime.}\label{fig:J012}
\end{figure}

%%%%%%%%%%%%%%%%%%%%%%%%%%%%%%%%%%%%%%%%%%%%%%%%%%%%%%%%%%%%%%%%%%%%%%%%%%%%%%%%%%  est of J1

\noindent \(\bullet\) {\bf Estimate of \(J_1\):} To estimate \(J_1\), we partition the energy interval into subintervals whose lengths grow by a factor of \(T^\alpha\) at each step. More precisely, set
\[
[E_L, E_R] = \bigcup_{j = 1}^{j_{\max}} I_j,
\]
where
\begin{align*}
    I_j = E_L T^{(j-1)\alpha}[1, T^\alpha] = [E_L T^{(j-1)\alpha}, E_L T^{j\alpha}], \quad j = 1, \dots, j_{\max}.
\end{align*}
Because \(E_L T^{j_{\max}\alpha} = E_R\), we obtain for \(0 < \alpha < \frac{2}{5}\):
\begin{align*}
    j_{\max} = \frac{\log E_R - \log E_L}{\alpha \log T} = \frac{-\alpha+\frac{2}{5}}{\alpha} \le \frac{1}{\alpha}.
\end{align*}
For each \(I_j\), we apply the approach used in estimating \eqref{eqn:Mtq-low}. For \(E \in I_j\), we have \(E_L T^{j\alpha} \ge E \ge E_L T^{(j-1)\alpha}\). Hence,
\begin{align*}
    |n| \le E^{-1} T^{5\alpha}\le E_L^{-1}T^{(-j+6)\alpha}, \quad {\rm and} \quad I_j\subset [-E_L T^{j\alpha}, E_L T^{j\alpha}]. 
\end{align*}
Therefore, we bound the integral over \(I_j\) from above by  
\begin{align*}
      \int_{I_j} \Big( \sum_{0 \le |n| \le E^{-1} T^{5\alpha}} \frac{|n|^q}{\pi T}  \, \mathbb{E}|G^z(n, 0)|^2\, \Big)\, dE  
     \le   \int_{-E_L T^{j\alpha}}^{E_L T^{j\alpha}} \Big( \sum_{0 \le |n| \le E_L^{-1} T^{(-j+6)\alpha}} \frac{|n|^q}{\pi T}  \, \mathbb{E}|G^z(n, 0)|^2\, \Big)\, dE  .
\end{align*}
Applying \eqref{eqn:GImB-int} with \(E_2=E_L T^{j\alpha}\le T^{-\alpha}\) gives 
\begin{align*}
         \int_{-E_L T^{j\alpha}}^{E_L T^{j\alpha}} \Big( \sum_{0 \le |n| \le E_L^{-1} T^{(-j+6)\alpha}} \frac{|n|^q}{\pi T}  \, \mathbb{E}|G^z(n, 0)|^2\, \Big)\, dE  
    &\le  2C \big(E_L^{-1} T^{(-j+6)\alpha}\big)^q \cdot  \big(E_L T^{j\alpha}\big)^{\frac{1}{2}} \\
    &=   2C  T^{\frac{2}{5}(q - \frac{1}{2})} T^{6q\alpha+j\alpha (\frac{1}{2}-q)} \\
     &\le   2C  T^{\frac{2}{5}(q - \frac{1}{2})} T^{6q\alpha}, 
\end{align*}
provided \(q\ge \frac{1}{2}\).
Hence, with the constant \(C\) given in \eqref{eqn:GImB-int}, 
\begin{align}
    J_1 \le \sum_{j = 1}^{j_{\max}} \int_{I_j} \frac{dE}{\pi T} \sum_{0 \le |n| \le E^{-1} T^{5\alpha}} |n|^q \, \mathbb{E}|G^z(n, 0)|^2
    &\le j_{\max} \cdot 2C T^{\frac{2}{5}(q - \frac{1}{2})} T^{6q\alpha} \notag  \\
    &\le \frac{2C}{\alpha} T^{\frac{2}{5}(q - \frac{1}{2}) + 6q\alpha}. \label{eqn:J1-est}
\end{align}  

\noindent \(\bullet\) {\bf Estimate of \(J_2\):} Damanik et al.~\cite{damanik2007upper} showed that decay estimates for the Green’s function in \(n\) at complex energies can be expressed in terms of transfer matrices for discrete one-dimensional Schrödinger operators. Jitomirskaya and Schulz-Baldes~\cite{jitomirskaya2007upper} established a similar result for general Jacobi matrices, with constants independent of energy. Below, we restate the result from \cite{jitomirskaya2007upper} in the context of the div-grad model \eqref{eqn:div-grad}.

\begin{proposition}[{\cite[Proposition~2]{jitomirskaya2007upper}}]
Let \(H_\omega\) be the Jacobi operator in \eqref{eqn:div-grad}. There exists a constant \(c > 0\), depending on \(a_-\) and \(a_+\) in \eqref{eqn:an-bound}, such that for any \(\omega\) in a full-measure set, any \(z = E + \frac{i}{T}\) with \(T \ge 1\), and any \(n \ge 1\),
\begin{align}\label{eqn:green-tail}
    \sum_{|m| > n} |G^z(m, 0; \omega)|^2 \le \frac{c T^6}{\max\limits_{0 \le |m| \le n} \|T_m^z(\omega)\|^2}.
\end{align}
\end{proposition}

For high-frequency terms in \(n\) and mild energies \(E > E_L\), the key step is to derive a lower bound on \(\|T_n^z\|\). The goal is to establish the following:

\begin{lemma}\label{lem:Gn-decay}
Let \(z = E + i/T\), \(E_L = T^{-\frac{2}{5}}\), and \(E_R = T^{-\alpha}\) be as in \eqref{eqn:ELR}. For \(0 < \alpha < \frac{1}{4}\), there exists \(T_0 = T_0(\alpha) > 0\) such that for \(T \ge T_0(\alpha)\), \(E^{-1} \le n \le T^{1+\alpha}\), and \(E_L \le E \le E_R\),
\begin{align}\label{eqn:1/T-exp}
    \mathbb{E}\!\left(\frac{1}{\max\limits_{1 \le m \le n} \|T_m^z(\omega)\|^2}\right) 
    \le 2 e^{-\frac{D_0}{8} E^{2\alpha + \frac{1}{2}} \sqrt{n}},
\end{align}
where \(D_0\) is the constant as in \eqref{eqn:LE-bound}.
\end{lemma}
As a consequence of \eqref{eqn:green-tail}, we have for \(E^{-1}\le |n| \le T^{1+\alpha}\) and \(T^{-\frac{2}{5}}\le E\le T^{-\alpha}\) \begin{align}
 \mathbb{E}|G^z(n, 0; \omega)|^2  + \mathbb{E}|G^z(-n, 0; \omega)|^2 \le  2c T^6 e^{-\frac{D_0}{8} E^{2\alpha + \frac{1}{2}} \sqrt{|n|-1}}  
\le   2c T^6 e^{-\frac{D_0}{10} E^{2\alpha + \frac{1}{2}} \sqrt{|n|}},\label{eqn:Green-upper}
\end{align}
provided \(|n|\ge E^{-1}\ge T^{\alpha}\ge 3\) so that \(\sqrt{|n|-1}\ge 0.8\sqrt{|n|}\).

 Since the index \(n\) in the sum for \(J_2\) satisfies \(E^{-1}< E^{-1} T^{5\alpha}\le |n| \le T^{1+\alpha}\), inequality \eqref{eqn:Green-upper} applies to \(\mathbb{E}|G^z(n, 0)|^2\). Set \(c_1=2c\) and \(c_2=\frac{D_0}{10}\). We can use \eqref{eqn:Green-upper} to bound \(J_2\) from above as 
\begin{align}\label{eqn:J2-upper}
    J_2 
     \le \frac{1}{\pi T}\int_{E_L}^{E_R} \sum_{|n| \ge E^{-1} T^{5\alpha}} |n|^q \, c_1 T^6 e^{-c_2 E^{2\alpha + \frac{1}{2}} |n|^{\frac{1}{2}}}  \, dE  . 
\end{align}

To estimate the sum on the right-hand side for \(|n| \ge E^{-1} T^{5\alpha}\), which is the tail of a sub-exponentially decaying series, we use the following quantitative estimate from \cite{jitomirskaya2007upper}:

\begin{lemma}[{\cite[Lemma 2]{jitomirskaya2007upper}}]\label{lem:geo-sum}
Let \(\Delta, \tau > 0\), \(q \ge 0\), and \(N \in \mathbb{N}\). Define \(p = \lfloor \frac{q + 1}{\tau} \rfloor\). Then
\begin{align}\label{eqn:geo-sum}
    \sum_{n \ge N} n^q e^{-\Delta n^\tau} \le C_{\tau, q} (N + \Delta^{-1})^p \frac{e^{-\Delta N^\tau}}{\Delta}.
\end{align}
\end{lemma}

Set \(N = E^{-1} T^{5\alpha}\), \(\Delta = c_2 E^{2\alpha + \frac{1}{2}}\), and \(\tau = \frac{1}{2}\). Then \(p = \lfloor \frac{q + 1}{\tau} \rfloor \le 2(q + 1)\). For \(E \ge E_L\), and \(0<\alpha<1/4\),
\begin{align*}
    \frac{(N + \Delta^{-1})^p}{\Delta} 
     = \frac{1}{c_2} E^{-2\alpha - \frac{1}{2}} \Big(E^{-1} T^{5\alpha} + c_2^{-1} E^{-2\alpha - \frac{1}{2}}\Big)^p  
    &\le \frac{1}{c_2^{p+1}} E_L^{-2\alpha - \frac{1}{2}} \Big(E_L^{-1} T^{5\alpha}\Big)^p \\
    &:= c_3 T^{q_0},
\end{align*}
where
\[
q_0 = \frac{2}{5}\big[2\alpha+1+2(q+1)\big]+5\alpha q, \quad {\rm and} \quad c_3=  c_2^{-(p+1)}  . 
\]
Moreover,
\begin{align*}
    \Delta N^\tau \ge c_2 E^{2\alpha + \frac{1}{2}} \Big(E^{-1} T^{5\alpha}\Big)^{\frac{1}{2}} 
    \ge c_2 E_L^{2\alpha} T^{\frac{5}{2}\alpha} 
    = c_2 T^{-\frac{4}{5} \alpha  + \frac{5}{2}\alpha} 
    \ge c_2 T^\alpha.
\end{align*}
Applying \eqref{eqn:geo-sum} to \eqref{eqn:J2-upper} using these parameters yields
\begin{align}
    J_2 
     \le   \frac{1}{\pi T}|E_R-E_L|\, (c_1 T^6)\,  C_{\tau, q}\,  ( c_3 T^{q_0})\, e^{-c_2 T^\alpha}  
    \le C', \label{eqn:J2-est}
\end{align}
where \(C'\) is a constant depending on \(\alpha\) and \(q\), valid for \(T > T_2'(\alpha, q)\).

\begin{proof}[Proof of Lemma~\ref{lem:E-mild}]
Combining \eqref{eqn:J1-est} and \eqref{eqn:J2-est} completes the proof of \eqref{eqn:mildE}.
\end{proof}

%%%%%%%%%%%%%%%%%%

The remainder of this section is devoted to proving Lemma~\ref{lem:Gn-decay}. The argument relies on a bootstrap large deviation approach developed in \cite{jitomirskaya2007upper} for the random dimer model. In our setting, the asymptotic behavior of the Lyapunov exponent differs, as described in \eqref{eqn:linearLE}. Moreover, we must address the singularity as \(E \to 0^+\) arising from \eqref{eqn:W}.

%%%%%%%%%%%%%%%%%%%%%%%%%%%%%%%%%%%%%%%%%%%%%%%%%%%%%%%%%%%%%%%%%%%%%%%%%%%%%%%%%%%%%%%%%%%%%%%%%%%%%%%%%
\subsubsection{Bootstrap LDT and Proof of Lemma~\ref{lem:Gn-decay}}\label{sec:bootstrap-ldt}
We employ a bootstrap argument: beginning with the probabilistic estimate \eqref{eqn:ldt-TnZ} for systems of size \(n_0 < 1/E\), we iteratively extend this bound to sizes exceeding \(1/E\).

\noindent \(\bullet\) {\bf Upper bounds on transfer matrices.} For \(\alpha > 0\), \(T \ge 1\), and \(E \ge E_L = T^{-\frac{2}{5}}\), let \(n_0 = E^{-1 + 2\alpha}\). Observe that
\[
    n_0^{1 + 2\alpha} E = E^{4\alpha^2} \le 1, \qquad
    n_0 E^{-\frac{3}{2}} = E^{-\frac{5}{2} + 2\alpha} \le T^{-\frac{2}{5}(-\frac{5}{2} + 2\alpha)} \le T.
\]
Hence, the conditions of Theorem~\ref{thm:ldt-Tn-norm} are satisfied. From \eqref{eqn:ldt-TnZ}, we obtain
\[
    \mathbb{P}\left\{ \omega : \|T_{n_0}^z(\omega)\| \le e^c E^{-\frac{3}{2}} \right\} \ge 1 - n_0 e^{-n_0^\alpha},
\]
where \(E_0\) and the constant \(C = e^c\) are as in \eqref{eqn:ldt-TnZ}.

On the other hand, assume \(T > T_0(\alpha)\) is large so that \(E \le E_R = T^{-\alpha}\) is small, which implies \(E^{-\frac{3}{2}} \le e^{n_0^{2\alpha}}\). Then the probability estimate becomes
\begin{align*}
    \mathbb{P}\left\{ \omega : \|T_{n_0}^z(\omega)\| \le e^{c + n_0^{2\alpha}} \right\} \ge 1 - n_0 e^{-n_0^\alpha}.
\end{align*}

Recall that the shift operator \((S\omega)(n) = \omega(n+1)\) in \eqref{eqn:shift} preserves the probability measure. Hence, for each \(j = 0, 1, \dots\),
\begin{align*} 
    \mathbb{P}\left\{ \omega : \|T_{n_0}^z(S^{j n_0} \omega)\| \le e^{c + n_0^{2\alpha}} \right\} \ge 1 - e^{-n_0^\alpha}.
\end{align*}

For \(n_1 \ge E^{-1 + 2\alpha}\), if \(\|T_{n_0}^z(S^{j n_0} \omega)\| \le e^{c + n_0^{2\alpha}}\) holds for all \(j = 0, \dots, \frac{n_1}{n_0} - 1\), then
\begin{align*}
    \|T_{n_1}^z(\omega)\| 
    &= \left\| \prod_{j = 0}^{\frac{n_1}{n_0} - 1} T_{n_0}^z(S^{j n_0} \omega) \right\| \\
    &\le e^{(c + n_0^{2\alpha}) \cdot \frac{n_1}{n_0}} 
    = \exp\left\{ \left(c E^{1 - 2\alpha} + E^{1 - 4\alpha + 4\alpha^2} \right) n_1 \right\} 
    \le e^{2 E^{1 - 4\alpha} n_1}, 
\end{align*}
provided \(c \le E^{-2\alpha}\). Therefore,
\begin{align*}
    \mathbb{P}\left\{ \omega : \|T_{n_1}^z(\omega)\| > e^{2 E^{1 - 4\alpha} n_1} \right\} 
    &\le \sum_{j = 0}^{\frac{n_1}{n_0} - 1} \mathbb{P}\left\{ \omega : \|T_{n_0}^z(S^{j n_0} \omega)\| > e^{c + n_0^{2\alpha}} \right\} \\
    &\le \frac{n_1}{n_0} n_0 e^{-n_0^\alpha} 
    \le n_1 e^{-E^{-\alpha/2}},
\end{align*}
where the last step uses \(-\alpha + 2\alpha^2 \le -\frac{\alpha}{2}\) for \(0 < E < 1\) and \(0 < \alpha < \frac{1}{4}\).

The same estimate applies to \(\|T_{n_1}^z(S^j \omega)\|\) for any \(j\). In conclusion, we obtain the following deviation estimate:

\begin{lemma}\label{lem:ldt-upper}
Let \(0 < \alpha < \frac{1}{4}\). Then there exists a constant \(E_0(\alpha)\) such that for all \(E_R > E \ge E_L = T^{-\frac{2}{5}}\), \(T \ge 1\), \(j \ge 0\), and \(n_1 \ge E^{-1 + 2\alpha}\), we have
\begin{align}\label{eqn:ldt-upper}
    \mathbb{P}\left\{ \omega : \|T_{n_1}^z(S^j \omega)\| \le e^{2 E^{1 - 4\alpha} n_1} \right\} 
    \ge 1 - n_1 e^{-E^{-\alpha/2}}.
\end{align}
\end{lemma}

%%%%%%%%%%%%%%%%%%%%%%%%%%%%%%%%%%%%%%%%%%%%%%%%%%%%%%%%%%%%%%%%%%%%
\noindent \(\bullet\) {\bf Lower bounds on transfer matrices.} Let \(D_0\) be as in \eqref{eqn:LE-bound}, so that for \(0 \le E < E_0\), we have \(L(E) \ge D_0 E\). Then, by \eqref{eqn:Lz-lower}, for any \(T > 0\),
\[
L(E + i/T) \ge L(E) \ge D_0 E.
\]
Combining this with the infimum in \eqref{eqn:Lyp}, we obtain for any \(n > 0\):
\begin{align*}
    \mathbb{E}\big(\log \|T^z_n(\omega)\|\big) \ge D_0 E n, \quad z = E + \frac{i}{T}.
\end{align*}
Now let \(n_1, E\) satisfy the conditions in Lemma~\ref{lem:ldt-upper}, and define
\begin{align}\label{eqn:p0def}
    p_0 = \mathbb{P}\big\{ \omega : \|T^z_{n_1}(S^j \omega)\| \ge e^{\frac{1}{2} D_0 E n_1} \big\}.
\end{align}
Recall the trivial upper bound \(\|T^z_n(\omega)\| \le e^{\gamma_1 |n|}\) for any \(|z| \le 1\), any \(n\), and any \(S^j \omega\) in a full-measure set, where \(\gamma_1\) depends explicitly on \(a_-, a_+\) in \eqref{eqn:an-bound}. Combining this with \eqref{eqn:ldt-upper} and \eqref{eqn:p0def}, we have
\begin{align}\label{eqn:430}
    D_0 E n_1 \le (1 - p_0) \frac{1}{2} D_0 E n_1  
    + p_0 \cdot (2 E^{1 - 4\alpha} n_1)  
    + n_1 e^{-E^{-\alpha/2}} \cdot (\gamma_1 n_1).
\end{align}
where in the last term we applied the trivial upper bound \(\log \|T_{n_1}^z\| \le \gamma_1 n_1\) to the complement set of \eqref{eqn:ldt-upper}. Dividing both sides of \eqref{eqn:430} by \(n_1\) gives
\begin{align*}
    D_0 E \le (1 - p_0) \frac{1}{2} D_0 E + p_0 \cdot 2 E^{1 - 4\alpha} + e^{-E^{-\alpha/2}} \gamma_1 n_1,  
\end{align*}
which implies
\begin{align}\label{eqn:p0}
    p_0 \ge \frac{D_0 E - 2 e^{-E^{-\alpha/2}} \gamma_1 n_1}{4 E^{1 - 4\alpha} - D_0 E}  
    = E^{4\alpha} \frac{D_0 - 2 E^{4\alpha-1} e^{-E^{-\alpha/2}} \gamma_1 n_1}{4 - D_0 E^{4\alpha}}  
    \ge \frac{D_0}{8} E^{4\alpha}. 
\end{align}
The last inequality can be guaranteed by taking \(T\) large so that the second term in the numerator is negligible:
\begin{align}\label{eqn:T-large}
   2 \gamma_1 n_1 E^{4\alpha-1} e^{-E^{-\alpha/2}} \le \frac{D_0}{2} 
   \Longleftrightarrow 2 \gamma_1 n_1 E^{4\alpha-1} \le \frac{D_0}{2} e^{E^{-\alpha/2}}. 
\end{align}
More precisely, for \(\alpha > 0\), there exists \(T_0 = T_0(\alpha) \ge 1\) such that for \(T \ge T_0\),
\[
2 \gamma_1 T^{1+\alpha} T^{-\frac{2}{5}(4\alpha-1)} \le \frac{D_0}{2} e^{T^{\alpha^2/2}}.
\]
If we also assume \(n_1 \le T^{1+\alpha}\) and \(T^{-\frac{2}{5}} \le E \le T^{-\alpha}\), then
\[
2 \gamma_1 n_1 E^{4\alpha-1} \le 2 \gamma_1 T^{1+\alpha} T^{-\frac{2}{5}(4\alpha-1)}, 
\quad {\rm and} \quad 
\frac{D_0}{2} e^{E^{-\alpha/2}} \ge \frac{D_0}{2} e^{T^{\alpha^2/2}}.
\]
Combining with \eqref{eqn:T-large} shows that for any \(j \in \mathbb{Z}\), \(T \ge T_0\), \(E^{-1 + 2\alpha} \le n_1 \le T^{1+\alpha}\), and \(T^{-\frac{2}{5}} \le E \le T^{-\alpha}\), inequality \eqref{eqn:p0} holds.

This probability estimate \eqref{eqn:p0} does not improve as the system size \(n_1\) increases and deteriorates as \(E \to 0^+\). Next, we bootstrap it for system sizes larger than the inverse localization length \(1/E\) by iteration.

\begin{lemma}\label{lem:ldt-lower}
Let \(0 < \alpha < \frac{1}{4}\), and let \(T_0 = T_0(\alpha)\) be as in \eqref{eqn:p0}. For \(T \ge T_0(\alpha)\), \(E^{-1} \le n \le T^{1+\alpha}\), and \(T^{-\frac{2}{5}} \le E \le T^{-\alpha}\),  
\begin{align}\label{eqn:ldt-mild}
    \mathbb{P}\Big\{\omega : \max_{1 \le m \le n} \|T_n^z(\omega)\| > e^{\frac{1}{4} D_0 E^{2\alpha + \frac{1}{2}}\sqrt n } \Big\} 
    \ge 1 - e^{-\frac{D_0}{8} E^{2\alpha + \frac{1}{2}} \sqrt n }.
\end{align}
\end{lemma}

\begin{proof}
Let \(n_1\) and \(p_0\) be as in \eqref{eqn:p0}. For \(n \ge n_1 \ge E^{-1 + 2\alpha}\), split \(n\) into approximately \(\frac{n}{n_1}\) segments of length \(n_1\). Then for \(j = 0, 1, \dots, \frac{n}{n_1} - 1\),
\begin{align}\label{eqn:prob-Tn1}
    \mathbb{P}\Big\{\omega : \|T_{n_1}^z(S^{j n_1} \omega)\| \ge e^{\frac{1}{2} D_0 E n_1}\Big\} \ge \frac{D_0}{8} E^{4\alpha}.
\end{align}

For each \(j\), write
\begin{align*}
    T_{n_1}^z(S^{j n_1} \omega) = T_{(j+1)n_1}^z(\omega) \big[T_{j n_1}^z(\omega)\big]^{-1}.
\end{align*}
Since \(T_n^z(S^j \omega) \in SL(2, \mathbb{C})\) for any \(n, j\), we have \(\|T_n^z(S^j \omega)\| = \big\|\big[T_n^z(S^j \omega)\big]^{-1}\big\|\). Hence, for \(j = 0, \dots, \frac{n}{n_1} - 1\), if \(\|T_{(j+1)n_1}^z(\omega)\| \le e^{\frac{1}{4} D_0 E n_1}\) and \(\|T_{j n_1}^z(\omega)\| \le e^{\frac{1}{4} D_0 E n_1}\), then \(\|T_{n_1}^z(S^{j n_1} \omega)\| \le e^{\frac{1}{2} D_0 E n_1}\). Therefore,
\begin{align}\label{eqn:550}
    \Big\{\omega : \max_{1 \le j \le \frac{n}{n_1}} \|T_{j n_1}^z(\omega)\| \le e^{\frac{1}{4} D_0 E n_1}\Big\} 
    \subset \bigcap_{j = 1}^{\frac{n}{n_1} - 1} \Big\{\omega : \|T_{n_1}^z(S^{j n_1} \omega)\| \le e^{\frac{1}{2} D_0 E n_1}\Big\}.
\end{align}
Clearly, \(\|T_{j n_1}^z(\omega)\| \le e^{\frac{1}{4} D_0 E n_1}\) implies \(\|T_{j n_1}^z(\omega)\| \le e^{\frac{1}{2} D_0 E n_1}\), so \(j = 0\) can also be included in the intersection above. Therefore, computing the probability in \eqref{eqn:550} gives
\begin{align*}
    \mathbb{P}\Big\{\omega : \max_{1 \le j \le \frac{n}{n_1}} \|T_{j n_1}^z(\omega)\| \le e^{\frac{1}{4} D_0 E n_1}\Big\} 
    &\le \prod_{j = 0}^{\frac{n}{n_1} - 1} \mathbb{P}\Big\{\omega : \|T_{n_1}^z(S^{j n_1} \omega)\| \le e^{\frac{1}{2} D_0 E n_1}\Big\} \\
    &\le \Big(1 - \frac{D_0}{8} E^{4\alpha}\Big)^{\frac{n}{n_1}} \le e^{-\frac{D_0}{8} E^{4\alpha} \frac{n}{n_1}}.
\end{align*}

Setting \(E n_1\) and \(E^{4\alpha} \frac{n}{n_1}\) equal gives \(n_1 = E^{2\alpha - \frac{1}{2}} n^{\frac{1}{2}}\), and
\[
E n_1 = E^{4\alpha} \frac{n}{n_1} = E^{2\alpha + \frac{1}{2}} n^{\frac{1}{2}}.
\]
To ensure \(E^{-1 + 2\alpha} \le n_1 \le T^{1+\alpha}\), we require \(E^{-1} \le n \le T^{1+\alpha}\). Thus,
\begin{align*}
    \mathbb{P}\Big\{\omega : \max_{1 \le m \le n} \|T_m^z(\omega)\| \le e^{\frac{1}{4} D_0 E^{2\alpha + \frac{1}{2}} n^{\frac{1}{2}}}\Big\} 
    &\le \mathbb{P}\Big\{\omega : \max_{1 \le j \le \frac{n}{n_1}} \|T_{j n_1}^z(\omega)\| \le e^{\frac{1}{4} D_0 E^{2\alpha + \frac{1}{2}} n^{\frac{1}{2}}}\Big\} \\
    &\le \exp\Big\{-\frac{D_0}{8} E^{2\alpha + \frac{1}{2}} n^{\frac{1}{2}}\Big\},
\end{align*}
which proves \eqref{eqn:ldt-mild}.
\end{proof}

\begin{proof}[Proof of Lemma~\ref{lem:Gn-decay}]
It suffices to compute the expectation in \eqref{eqn:1/T-exp} over the probability set in \eqref{eqn:ldt-mild} and its complement:
\begin{align*} 
    \mathbb{E}\Big(\frac{1}{\max\limits_{1 \le m \le n} \|T_m^z(\omega)\|^2}\Big) 
    \le e^{-\frac{D_0}{4} E^{2\alpha + \frac{1}{2}} \sqrt n} + e^{-\frac{D_0}{8} E^{2\alpha + \frac{1}{2}} \sqrt n} 
    \le 2 e^{-\frac{D_0}{8} E^{2\alpha + \frac{1}{2}} \sqrt n}.
\end{align*}
In the second term, where the event in \eqref{eqn:ldt-mild} fails, we use the trivial bound for the \(SL(2,\mathbb{C})\) transfer matrix: \(\|T^z_n\| \ge 1\) for any \(n\).
\end{proof}

\subsection{High energy regime}
Let \(\alpha > 0\) and \(E_R = T^{-\alpha}\) be as in \eqref{eqn:ELR}. We now estimate the last term in \eqref{eqn:Mtq-split} for high energies beyond \(E_R\). This term resembles the mild-energy regime, involving a splitting in the frequency \(n\). Here, the splitting is simpler, dividing into a low-frequency regime (\(J_1'\)) and a high-frequency regime (\(J_2'\)), as shown in the right portion of Figure~\ref{fig:J012}:
\begin{align*}
\E M_T^{q,0,1+\alpha}(E_R, E_0)\le 
  &\frac{1}{\pi T} \int_{E_R}^{E_0}  \sum_{0 \le |n| \le  T^{4\alpha}} |n|^q \, \mathbb{E}|G^{E + i/T}(n, 0)|^2\, dE \quad (:=J_1') \\
  &+\frac{1}{\pi T}\int_{E_R}^{E_0}  \sum_{|n| \ge T^{4\alpha}} |n|^q \, \mathbb{E}|G^{E + i/T}(n, 0)|^2\, dE. \quad (:=J_2')
\end{align*}

The estimate of \(J_1'\) is similar to \eqref{eqn:Mtq-low}. Applying \eqref{eqn:GImB-intFull} directly gives
\begin{align*}
   J_1' &\le \int_0^{E_0} \left( \sum_{|n| \le T^{4\alpha}} \frac{|n|^q}{\pi T} \, \mathbb{E}|G^{E + i/T}(n, 0)|^2 \right) dE  
     \le T^{4q\alpha}. 
\end{align*}

The estimate of \(J_2'\) follows the argument for \(J_2\), using a large-deviation estimate similar to \eqref{eqn:ldt-mild}. Here, the lower bound \(E \ge T^{-\alpha}\) is stronger, allowing a weaker LDT without invoking the probabilistic upper bound \eqref{eqn:ldt-upper}. Define
\begin{align}\label{eqn:p0def-prime}
    p_0' = \mathbb{P}\big\{ \omega : \|T^z_{n'_1}(S^j \omega)\| \ge e^{\frac{1}{2} D_0 E n'_1} \big\},
\end{align}
and use only the trivial uniform bound \(\|T^z_{n'_1}(\omega)\| \le e^{\gamma_1 n_1'}\). Then, as in \eqref{eqn:430}, for \(0 < E < E_0\),
\[
D_0 E n'_1 \le (1 - p_0') \frac{1}{2} D_0 E n'_1 + p_0' \cdot (\gamma_1 n_1')
\Longrightarrow p_0' \ge \frac{D_0}{2\gamma_1} E.
\]
This estimate holds for any \(n'_1\) and \(j\), since \eqref{eqn:ldt-TnZ} is not used and there is no restriction such as \(n_1 \ge E^{-1+2\alpha}\) as in \eqref{eqn:ldt-upper}. Repeating the proof of \eqref{eqn:ldt-mild} with \(p_0'\) from \eqref{eqn:p0def-prime} gives, for any \(n \ge n'_1\),
 \begin{align*}
    \mathbb{P}\Big\{\omega : \max_{1 \le j \le \frac{n}{n'_1}} \|T_{j n'_1}^z(\omega)\| \le e^{\frac{1}{4} D_0 E n'_1}\Big\} 
    &\le \prod_{  j = 0}^{\frac{n}{n'_1} - 1} \mathbb{P}\Big\{\omega : \|T_{n'_1}^z(S^{j n'_1} \omega)\| \le e^{\frac{1}{2} D_0 E n'_1}\Big\} \\
    &\le \Big(1 - \frac{D_0}{2\gamma_1} E\Big)^{\frac{n}{n'_1}}  \le e^{-\frac{D_0}{2\gamma_1} E \frac{n}{n'_1}}.
\end{align*}
Setting \(n'_1 = \sqrt{2n/\gamma_1}\) gives, for \(n \ge \gamma_1/2\),
\begin{align*}
    \mathbb{P}\Big\{\omega : \max_{1 \le m \le n} \|T_m^z(\omega)\| \le e^{\frac{\sqrt{2} D_0}{4 \sqrt{\gamma_1}} E \sqrt{n}}\Big\} 
    \le e^{-\frac{\sqrt{2} D_0}{4 \sqrt{\gamma_1}} E \sqrt{n}}.
\end{align*}
Combining this with \eqref{eqn:green-tail}, as in \eqref{eqn:Green-upper}, gives for \(|n| \ge T^{4\alpha} \ge 3\) and \(0 < E < E_0\),
\begin{align}\label{eqn:Green-upper-weak}
  \mathbb{E}|G^z(n, 0; \omega)|^2 + \mathbb{E}|G^z(-n, 0; \omega)|^2 \le 2c T^6 e^{-c_2' E \sqrt{|n|}}, \quad c_2' = \frac{\sqrt{2} D_0}{5 \sqrt{\gamma_1}}.
\end{align}
Finally, substituting this bound into \(J_2'\) gives
\begin{align*}
    J_2' 
     \le \int_{E_R}^{E_0} \Bigg(\sum_{|n| \ge T^{4\alpha}} |n|^q \, 4 e^{-c_2' E \sqrt{|n|}}\Bigg) \frac{dE}{\pi T}  
    &\le \frac{2c T^6}{\pi T} \sum_{|n| \ge T^{4\alpha}} |n|^q e^{-c_2' T^{-\alpha} |n|^{1/2}}.
\end{align*}
Applying Lemma~\ref{lem:geo-sum} with \(\Delta = c_2' T^{-\alpha}\) and \(\tau = \frac{1}{2}\) yields, for some explicit constants \(C', C''\) and \(q_0'\),
\begin{align*}
    J_2' \le C' T^{q_0'} e^{-C' T^{\alpha}} \le C'',
\end{align*}
provided \(T > T_0(q,\alpha)\). Combining the above estimates for $J_1'$ and $J_2'$ concludes the proof of Lemma \ref{lem:E-high}.

%%%%%%%%%%%%%%%%%%%%%%%%%%%%%%%%%%%%%%%%%%%%%%%%%%%%%%%%%%%%%%%%%%%%%%%%%%%%%%%%%%%%%%%%%%%%%%%%%%%%%%%%%%%%%

%%%%%%%%%%%%%%%%%%%%%%%%%%%%%%%%%%%%%%%%%%%%%%%%%%%%%%%%%%%%%%%%%%%%%%%%%%%%%%%%%%%%%%%%%%%%%%%%%%%%

%%%%%%%%%%%%%%%%%%%%%%%%%%%%%%%%%%%%%%%%%%%%%%%%%%%%%%%%%%%%%%%%%%%%%%%%% appendix 

\appendix 

\section{Spectrum of the div-grad model}\label{sec:spe-pf}
Let \(H_\omega\) be as in \eqref{eqn:div-grad}, where the coefficients \(a_n\) satisfy the essential bound \eqref{eqn:an-bound}: almost surely,
\begin{align*} 
{\rm supp}P_0=[a_-,a_+],\quad  0 < a_- \le a_n \le a_+ < \infty \quad \text{for all } n \in \mathbb{Z}.
\end{align*}
A direct computation shows that
\begin{align*} 
    \ipc{\varphi}{H_\omega \varphi} = \sum_{n \in \mathbb{Z}} a_n \, |\varphi_{n} - \varphi_{n-1}|^2,
\end{align*}
which implies
\begin{align*} 
  0 \le \ipc{\varphi}{H_\omega \varphi} \le a_+ \sum_{n \in \mathbb{Z}} |\varphi_{n} - \varphi_{n-1}|^2 = a_+ \ipc{\varphi}{-\Delta \varphi},
\end{align*}
where
\begin{align*} 
    (-\Delta)_n = -\phi_{n+1} + 2\phi_n - \phi_{n-1}, \quad n \in \mathbb{Z},
\end{align*}
is the one-dimensional discrete Laplacian, whose spectrum is \([0,4]\). Hence, almost surely,
\begin{align}\label{eqn:spe-app}
    \sigma(H_\omega) \subset a_+ \cdot \sigma(-\Delta) = [0,4a_+]=\sigma(-\Delta)\cdot {\rm supp}P_0.
\end{align}
We now prove the reverse inclusion in \eqref{eqn:spe-app}, inspired by the correspondence between the div–grad model and the isotopically disordered harmonic chain in \eqref{eqn:hu=eu} and \eqref{eqn:hv=ev}.
\begin{proposition}\label{prop:rev-spe}
Almost surely,
\begin{align} 
     a_+ \cdot \sigma(-\Delta) = \sigma(-\Delta) \cdot \mathrm{supp}\, P_0 \subset \sigma(H_\omega).
\end{align}    
\end{proposition}
\begin{proof}
  Let \(\lambda \in (0,4] \subset \sigma(-\Delta)\) and \(\mu \in \mathrm{supp}\, P_0\). By the Weyl criterion (see, e.g., \cite{kirsch2007invitation}), there exists a sequence of compactly supported approximate eigenfunctions of \(-\Delta\), denoted by \(\{v^{(k)}\}_{k=1}^{\infty}\), such that
\begin{align}\label{eqn:vk-weyl}
    \|v^{(k)}\| = 1, \qquad \|-\Delta v^{(k)} - \lambda v^{(k)}\| \le \frac{1}{k}.
\end{align}
By a standard Borel–Cantelli argument (see, e.g., \cite[Proposition 3.8]{kirsch2007invitation}, or \cite[Theorem 3.12]{aizenman2015random}), there exists a sequence \(j_k \to \infty\) such that
\begin{align*}
    \sup_{n \in \mathrm{supp}\, v^{(k)}} \big|a_{n+j_k} - \mu\big| \le \frac{1}{k}.
\end{align*}
Since \(0 < a_- \le \mu \le a_+ < \infty\), dividing by \(\mu\) gives
\begin{align}\label{eqn:an-mu}
    \sup_{n \in \mathrm{supp}\, v^{(k)} + j_k} \Big|\frac{a_n}{\mu} - 1\Big| \le \frac{1}{\mu k}.
\end{align}
 Because \(v^{(k)}\) is compactly supported and \(-\Delta\) is translation-invariant on \(\mathbb{Z}\), the shifted function \(\widetilde{v}^{(k)}_n = v^{(k)}_{n+j_k}\) also satisfies \eqref{eqn:vk-weyl}, with \(\mathrm{supp}\, \widetilde{v}^{(k)} = \mathrm{supp}\, v^{(k)} + j_k\). Moreover, \(\max_n |\widetilde{v}^{(k)}_n| \le \|\widetilde{v}^{(k)}\| \le 1\).
 
Define a sequence \(\{u^{(k)}\}\) by
\begin{align}\label{eqn:wt-u}
    u^{(k)}_n = -\frac{\widetilde{v}^{(k)}_{n+1} - \widetilde{v}^{(k)}_n}{\lambda \mu}.
\end{align}
Then
\begin{align*}
    u^{(k)}_{n+1} - u^{(k)}_n
    &= -\frac{\widetilde{v}^{(k)}_{n+2} - \widetilde{v}^{(k)}_{n+1}}{\lambda \mu}
       + \frac{\widetilde{v}^{(k)}_{n+1} - \widetilde{v}^{(k)}_n}{\lambda \mu}
    = -\frac{(\Delta \widetilde{v}^{(k)})_{n+1}}{\lambda \mu}.
\end{align*}
 Thus,
\begin{align*}
    (H_\omega u^{(k)})_n
    &= -a_{n+1}(u^{(k)}_{n+1} - u^{(k)}_n) + a_n(u^{(k)}_n - u^{(k)}_{n-1}) \\
    &= -\frac{a_{n+1}}{\lambda \mu} (\Delta \widetilde{v}^{(k)})_{n+1}
       + \frac{a_n}{\lambda \mu} (\Delta \widetilde{v}^{(k)})_n.
\end{align*}

Hence,
\begin{align*}
    (H_\omega u^{(k)})_n - \lambda \mu u^{(k)}_n
    &= -\frac{a_{n+1}}{\lambda \mu} \big[(\Delta \widetilde{v}^{(k)})_{n+1} - \lambda \widetilde{v}^{(k)}_{n+1}\big]
       + \frac{a_n}{\lambda \mu} \big[(\Delta \widetilde{v}^{(k)})_n - \lambda \widetilde{v}^{(k)}_n\big] \\
    &\quad - \widetilde{v}^{(k)}_{n+1}\Big[\frac{a_{n+1}}{\mu} - 1\Big]
       + \widetilde{v}^{(k)}_n\Big[\frac{a_n}{\mu} - 1\Big].
\end{align*}
Combining \eqref{eqn:vk-weyl} and \eqref{eqn:an-mu}, we obtain
\begin{align*}
    \|H_\omega u^{(k)} - \lambda \mu u^{(k)}\|
    \le \frac{2a_+}{\mu \lambda} \frac{1}{k} + \frac{2}{\mu k} \to 0 \quad \text{as } k \to \infty.
\end{align*}
From \eqref{eqn:wt-u}, we also have
\begin{align*}
    \widetilde{v}^{(k)}_n = \lambda \mu (u^{(k)}_n - u^{(k)}_{n-1}).
\end{align*}
Thus, \(\|\widetilde{v}^{(k)}\| = 1\) implies
\begin{align*}
  0<  \frac{1}{2\mu \lambda} \le \|u^{(k)}\| \le \frac{2}{\mu \lambda}<\infty,
\end{align*}
where the bounds are independent of \(k\). 

Therefore, \(u^{(k)}\) can be normalized to form a Weyl sequence associated with \(\lambda \mu\), which shows that \(\lambda \mu \in \sigma(H_\omega)\). Hence,
\[
(0,4] \cdot \mathrm{supp}\, P_0 \subset \sigma(H_\omega).
\]
Finally, note that \(0\) also belongs to the spectrum due to its compactness.
\end{proof}

%%%%%%%%%%%%%%%%%%%%%%%%%%%%%%%%%%%%%%%%%%%%%%%%%%%%%%%%%%%%%%%%%%%%%%%%%%%%%%%%%%%%%%%%%%%%%%%%%%%%%%%%%%%%%

%%%%%%%%%%%%%%%%%%%%%%%%%%%%%%%%%%%%%%%%%%%%%%%%%%%%%%%%%%%%%%%%%%%%%%%%%%%%%%%%%%%%%%%%%%%%%%%%%%%%

\section{Lyapunov exponent and quantum transport in the hyperbolic region}\label{sec:lyp-hyp}
As shown in \eqref{eqn:spe}, the spectrum of \(H_\omega\) is almost surely \([0,4a_+]\), where \(a_+ > 0\) is as in \eqref{eqn:an-bound}. In this section, we consider energies \(E < 0\) in the resolvent set approaching the critical energy \(E_c = 0\) from the left.

\subsection{Asymptotic behavior of the Lyapunov exponent as \(E \to 0^-\)}
\begin{proposition}\label{prop:Fn-hyp}
Consider a deterministic \(SL(2,\mathbb{C})\) cocycle of the form
\begin{align*}
    B^z_j = \begin{pmatrix}
        2 - \dfrac{z}{a_j} & -1 \\
        1 & 0
    \end{pmatrix}, \quad j \in \mathbb{Z}, \quad z \in \mathbb{C}.
\end{align*}
  Assume that
\begin{align}\label{eqn:bn-bound}
    0 < b_- \le \dfrac{1}{a_j} \le b_+ < \infty \quad \text{for all } j \in \mathbb{Z}.
\end{align}
Let \(F_n^z = B^z_{n-1} \cdots B^z_0\) for \(n \ge 1\). Then there exists \(0 < E_0 < 1\) such that for any \(z = E + i\delta\) and any \(n \ge 1\), if \(-E_0 < E < 0\), then
\begin{align}\label{eqn:Fn-hyp-lower}
    \log \|F_n^z\| \ge \frac{n}{2} \sqrt{-E b_-}
\end{align}
and
\begin{align}\label{eqn:Fn-hyp-upper}
    \log \|F_n^E\| \le 2n \sqrt{-E b_+} + \log \frac{3}{\sqrt{-E b_+}},
\end{align}
where \(\|F_n^z\| = \|\cdot\|_{\infty}\) denotes the entrywise maximum norm.
\end{proposition}
A direct application to the div–grad model, using the conjugacy relation (extended to complex energy \(z \neq 0\)) in \eqref{eqn:u-v-change}–\eqref{eqn:F-Tn-norm}, is as follows:
\begin{corollary}\label{cor:Lyp-hyper}
Let \(H_\omega\) be the div–grad model in \eqref{eqn:div-grad} with transfer matrix \(T_n^z\) as in \eqref{eqn:Tn-intro} and Lyapunov exponent defined in \eqref{eqn:Lyp}. There exist constants \(0 < E_0 < 1\) and \(C_0, C_1, D_0, D_1 > 0\), depending on \(a_-, a_+\) in \eqref{eqn:an-bound}, such that for any \(z = E + i\delta\) and any \(n \ge 1\), if \(-E_0 < E < 0\), then almost surely,
\begin{align}\label{eqn:Tn-hyp-lower}
    \|T_n^z\| \ge C_0 |E| e^{n D_0 \sqrt{|E|}}
\end{align}
and
\begin{align}\label{eqn:Tn-hyp-upper}
    \|T_n^E\| \le C_1 |E|^{-\frac{3}{2}} e^{n D_1 \sqrt{|E|}}.
\end{align}
Consequently, for \(-E_0 < E < 0\),
\begin{align}
    L(z) \ge D_0 \sqrt{-E}, \quad z = E + i\delta, \qquad \text{and} \qquad L(E) \le D_1 \sqrt{-E}.
\end{align}
\end{corollary}
\begin{remark}
The lower bound for the Lyapunov exponent at complex energies \(z\) with \({\rm Re} z < 0\) can also be derived from the Thouless formula \eqref{eqn:thouless} and the asymptotic behavior of the IDS in \eqref{eqn:NE-root}. The upper bound holds only for real energies \(E < 0\) and may fail for complex energies with \({\rm Im} z \neq 0\).
\end{remark}

\begin{proof}[Proof of Proposition~\ref{prop:Fn-hyp}]
For \(z = E + i\delta\) with \(E < 0\), set \(x = -E > 0\) and \(b_j = 1/a_j\), which is bounded away from zero and infinity as in \eqref{eqn:bn-bound}. Denote by
\[
B_j^{z} = \begin{pmatrix}
    2 + (x - i\delta)b_j & -1 \\
    1 & 0
\end{pmatrix}, \qquad
F_n^{z} = \begin{pmatrix}
    P_n(z) & Q_n(z) \\
    \widetilde{P}_n(z) & \widetilde{Q}_n(z)
\end{pmatrix}.
\]

From the recurrence
\[
\begin{pmatrix}
    P_{n+1} & Q_{n+1} \\
    \widetilde{P}_{n+1} & \widetilde{Q}_{n+1}
\end{pmatrix}
=
\begin{pmatrix}
    2 + (x - i\delta)b_n & -1 \\
    1 & 0
\end{pmatrix}
\begin{pmatrix}
    P_n & Q_n \\
    \widetilde{P}_n & \widetilde{Q}_n
\end{pmatrix},
\]
we obtain for \(n = 0, 1, \dots\),
\begin{align}\label{eqn:Pn-it}
\begin{cases}
    P_{n+1} = (2 + (x - i\delta)b_n) P_n - \widetilde{P}_n, \\
    \widetilde{P}_{n+1} = P_n,
\end{cases}
\quad \Longrightarrow \quad
P_{n+1} = (2 + (x - i\delta)b_n) P_n - P_{n-1},
\end{align}
with initial conditions
\[
P_1 = 2 + (x - i\delta)b_0, \quad P_0 = 1, \quad P_{-1} = 0.
\]

We estimate \(|P_n(z)|\) inductively. First, since \(x > 0\) and \(b_n \ge b_-\),
\[
|P_1| = |2 + x b_0 - i\delta b_0| > 2 + x b_- > 2 > |P_0|.
\]
Assume \(|P_n| > |P_{n-1}|\). Then
\begin{align}\label{eqn:Pn-recu}
|P_{n+1}| \ge |2 + (x - i\delta)b_n| \cdot |P_n| - |P_{n-1}|
\ge (2 + x b_-) |P_n| - |P_n| = (1 + x b_-) |P_n|.
\end{align}
Thus, \(|P_n|\) is strictly increasing and satisfies, for \(n \ge 0\),
\[
|P_n| > (1 + x b_-) |P_{n-1}| > \cdots > (1 + x b_-)^{n+1} |P_0| > (1 + x b_-)^{n+1} > e^{n x b_-}.
\]
This rough bound implies  
\(\frac{1}{n} \log \|F_n^z\| \ge \frac{1}{n} \log |P_n| \ge x b_-\) 
uniformly for \(n \ge 1\). To improve this to the order of \(\sqrt{x} = \sqrt{-E}\), we use the second-order recurrence in \eqref{eqn:Pn-recu}:
\[
|P_{n+1}| \ge (2 + x b_-) |P_n| - |P_{n-1}|.
\]
Denote by \(r>1\) the solution to \(r^2 - (2 + x b_-) r + 1 = 0\), given by
\begin{align*}
    r &= \frac{2 + x b_- + \sqrt{4 x b_- + (x b_-)^2}}{2} = 1 + \sqrt{x b_-} + O(x), \\
    r^{-1} &= \frac{2 + x b_- - \sqrt{4 x b_- + (x b_-)^2}}{2} = 1 - \sqrt{x b_-} + O(x).
\end{align*}
Then
\begin{align*}
    |P_{n+1}| - r |P_n| \ge r^{-1} \big[|P_n| - r |P_{n-1}|\big] \ge C_0 r^{-n}, \quad C_0 = |P_1| - r |P_0|,
\end{align*}
where \(C_0 = |P_1| - r |P_0| \ge 1 + x b_-\). Inductively,
\begin{align}\label{eqn:Pn-rec2}
    |P_{n+1}| \ge r |P_n| + C_0 r^{-n} \ge r |P_n| \ge r^{n+1} |P_0|.
\end{align}
Therefore, using the asymptotic expansion \(r = 1 + \sqrt{x b_-} + O(x)\), we obtain for \(x \to 0^+\),
\begin{align*}
    |P_n(z)| \ge r^n \ge e^{n \log\big(1 + \sqrt{x b_-} + O(x)\big)} \ge e^{n \frac{\sqrt{x b_-}}{2}}.
\end{align*}

This implies that for \(z = E + i\delta = -x + i\delta\), there exists \(0 < E_0 < 1\) such that if \(0 < x = -E < E_0\), then for any \(n \ge 0\) and any \(\delta \ge 0\),
\begin{align*}
    \log \|F_n^z\| \ge \log |P_n(z)| \ge n \frac{\sqrt{x b_-}}{2} = n \frac{\sqrt{-E b_-}}{2},
\end{align*}
which proves the lower bound in \eqref{eqn:Fn-hyp-lower}.

The above lower bounds hold for any complex energy \(\,z\in\C\). At a real energy \(\,x=-E>0\), by \eqref{eqn:Pn-it} and \eqref{eqn:Pn-recu}, \(P_n(E)=|P_n(E)|>0\) is real, strictly positive, and increasing, and satisfies
\begin{align*}
 P_{n+1}(E)= (2 +xb_n) P_{n}-P_{n-1} \le (2 +xb_+) P_{n}(E)-P_{n-1}(E).
\end{align*}
Thus, similar to \eqref{eqn:Pn-rec2}, we have
\begin{align}
      P_{n+1}(E)\le s P_n(E) + C_1s^{-n},
\end{align}
where \(\,s>1>s^{-1}\) solves \(s^2-(2+xb_+)s+1=0\) and as \(x\to 0^+\)
\begin{align*}
 \frac{1}{2}\le   C_1=P_1(E)-sP_0(E)=1-\sqrt{xb_+}+O(x)\le 2. 
\end{align*}  Inductively, for \(\,n\ge 0\),
\begin{align*}
P_{n}(E) \le s^{n} P_0 + C_1 \sum_{j=0}^{n-1} s^{2j-(n-1)}
=   s^{n} + C_1 s^{-n+1}\,\frac{1 - s^{2n}}{1 - s^2} 
\le   s^{n} +  \frac{2s}{s^2-1}s^{n}.
\end{align*}
Hence, for \(\,x\to 0^+\), using the fact that \(\,s>1\) solves \(s^2-(2+xb_+)s+1=0\) , one has \(s/(s^2-1)\le 1/\sqrt{xb_+}\), which implies
\begin{align*}
P_{n}(E)\le \Big(1+\frac{2}{\sqrt{xb_+}}\Big) s^{n} \le \frac{3}{\sqrt{-Eb_+}} e^{2n\sqrt{-Eb_+}}.
\end{align*}

For the upper right element \(\,Q_n(E)\) of \(F_n^E\), a similar recurrence relation holds
\[
Q_{n+1}(E)= (2 +xb_{n})Q_n(E)- Q_{n-1}(E),
\]
but with negative initial values \(Q_1=-1, Q_{0}=0\). Hence, \(Q_n<0\), is decreasing, and \(-Q_n\) satisfies the same recurrence inequality as \(P_n\):
\begin{align*}
0<-Q_{n+1}(E)\le (2 +xb_+)\big(-Q_n-(-Q_{n-1})\big).
\end{align*}
Therefore, the same upper bound holds for \(-Q_n(E)\): for \(C_1'=-Q_1-(-Q_0)=1\), 
\begin{align*}
0<-Q_{n}(E)\le \frac{3 }{\sqrt{-Eb_+}} e^{2n\sqrt{-Eb_+}}.
\end{align*}
Thus, there exists \(0<E_0'<1\) such that if \(0<x=-E<-E_0'\), then for any \(n\ge 0\), the max norm satisfies
\begin{align*}
\log \|F_n^E\|\le & \log\max\big\{|P_n(E)|,|P_{n-1}(E)|,|Q_n(E)|,|Q_{n-1}(E)|\big\} \\
\le & 2n \sqrt{-Eb_+}+\log \frac{3}{\sqrt{-Eb_+}}.
\end{align*}
Note that this upper bound only holds for real energies.

\end{proof}

\subsection{Quantum transport for \(E\to 0^-\)}\label{sec:Mtq-negativeE}

In Section~\ref{sec:upper}, we restrict the estimate of \(M_T^q\) to the energy interval \([-E_0,E_0]\) near the critical value \(E_c=0\) and examine the contribution from the right half \(M_T^{q}(0,E_0)\); see \eqref{eqn:Mtq-split} and Figure~\ref{fig:Mtq-split1}. Here, we outline the estimate for the left half \(M_T^{q}(-E_0,0)\), which is comparatively straightforward by Corollary~\ref{cor:Lyp-hyper}.

Combining the lower bound \eqref{eqn:Tn-hyp-lower} with \eqref{eqn:green-tail}, we obtain that for any \(|n|\ge 2\) and \(z=E+i/T\) with \(-E_0<E<0\), almost surely,
\begin{align}\label{eqn:Gn-negative}
    |G^z(n,0;\omega)|+|G^z(-n,0;\omega)|\le \frac{cT^6}{C_0|E|}e^{-\frac{1}{2}|n| D_0\sqrt{|E|}}.
\end{align}

We now split
\begin{align*}
  M_T^{q,0,1+\alpha}(-E_0,0)\le
  & \frac{1}{\pi T} \int_{-T^{-1}}^{0} \sum_{0 \le |n| \le T^{1+\alpha}} |n|^q \, |G^{E+i/T}(n,0)|^2\, dE \quad (:=X_1) \\
  & +\frac{1}{\pi T} \int_{-E_0}^{-T^{-1}} \sum_{0 \le |n| \le T^{\frac{1+\alpha}{2}}} |n|^q \, |G^{E+i/T}(n,0)|^2\, dE \quad (:=X_2) \\
  & +\frac{1}{\pi T} \int_{-E_0}^{-T^{-1}} \sum_{|n| \ge T^{\frac{1+\alpha}{2}}} |n|^q \, |G^{E+i/T}(n,0)|^2\, dE. \quad (:=X_3)
\end{align*}
Applying \eqref{eqn:GImB-int} in Lemma~\ref{lem:GImB} with \(E_2=T^{-1}\), we estimate \(\E X_1\):
\begin{align*}
 \E X_1 \le \frac{1}{\pi T} \int_{-T^{-1}}^{T^{-1}} \sum_{0 \le |n| \le T^{1+\alpha}} |n|^q \, \E |G^{E+i/T}(n,0)|^2\, dE \le 2C T^{q-\frac{1}{2}+\alpha q}.
\end{align*}
Similarly, using \eqref{eqn:GImB-intFull} in Lemma~\ref{lem:GImB}, we estimate \(\E X_2\):
\begin{align*}
 \E X_2 \le \frac{1}{\pi T} \int_{\R} \sum_{0 \le |n| \le T^{\frac{1+\alpha}{2}}} |n|^q \, \E |G^{E+i/T}(n,0)|^2\, dE \le T^{\frac{1+\alpha}{2}q}.
\end{align*}
The last term \(X_3\) is estimated using \eqref{eqn:Gn-negative} and Lemma~\ref{lem:geo-sum}. For \(E_0>|E|\ge T^{-1}\), \eqref{eqn:Gn-negative} implies
\begin{align*}
  |G^z(n,0;\omega)|+|G^z(-n,0;\omega)|\le \frac{cT^7}{C_0}e^{-\frac{1}{2}|n| D_0T^{-1/2}}.
\end{align*}
Then Lemma~\ref{lem:geo-sum} yields that for some constants \(C_2=C_2(q,\alpha)\) and \(q'=q'(q,\alpha)\),
\begin{align*}
   X_3\le \frac{E_0}{\pi T}\frac{cT^7}{C_0} \sum_{|n| \ge T^{\frac{1+\alpha}{2}}} |n|^q \, e^{-\frac{1}{2}|n| D_0T^{-1/2}}\le C_2T^{q'}e^{-\frac{1}{2}D_0T^{\alpha}}.
\end{align*}
Hence, almost surely, \(X_3\le C_3(\alpha,q)\) provided \(T>T(\alpha,q)\).

Combining the estimates for \(X_1,X_2,X_3\) gives, for \(\alpha>0,q>0\) and \(T>T(\alpha,q)\),
\begin{align*}
    \E M_T^{q,0,1+\alpha}(-E_0,0)\le (2C+1)\max\Big\{T^{q-\frac{1}{2}+\alpha q},T^{\frac{1+\alpha}{2}q}\Big\}+C_3.
\end{align*}
Thus, for \(q\ge 1\), 
\[
 \E M_T^{q,0,1+\alpha}(-E_0,0)\le (2C+1)T^{q-\frac{1}{2}+\alpha q}+C_3,
\]
which is dominated by the upper bound of \(\E M_T^{q}\) in \eqref{eqn:Mtq-upper}.
%%%%%%%%%%%%%%%%%%%%%%%%%%%%%%%%%%%%%%%%%%%%%%%%%%%%%%%%%%%%%%%%%%%%%%%%%%%%%%%%%%%%%%%%%%%%%%%%%%%%%%%%%%%%%%

\section{More Facts About Modified Pr\"ufer Variables}\label{sec:prufer-app}

In this section, we continue the discussion of the Figotin–-Pastur phase formalism (the modified Pr\"ufer variables) introduced in Section~\ref{sec:prufer-ldt}. This approach was first employed in \cite{pastur1992book} to study one-dimensional random Schr\"odinger operators and was later extended to other models in \cite{chulaev,bour00strongmix,jitomirskaya2003deloc}. We review additional fundamental facts for readers unfamiliar with these topics and then present a proof of Theorem~\ref{thm:NE-root}, which appears as an exercise in \cite[Problem 18, Page 183]{pastur1992book}. In addition, the second subsection provides further details on the iteration in Proposition~\ref{prop:rho-chi-Qn} and clarifies the argument underlying Theorem~\ref{thm:LE-linear}.

\subsection{Asymptotic Formulas for the Integrated Density of States}

Recall the conjugacy in \eqref{eqn:hu=eu}-\eqref{eqn:v-cocycle}.  
For \(E > 0\) and \(u_n\) satisfying \eqref{eqn:hu=eu}:
\begin{align}\label{eqn:hu=eu-app}
    -a_{n+1} u_{n+1} + (a_{n+1} + a_n) u_n - a_n u_{n-1} = E u_n, \quad n \in \mathbb{Z}.
\end{align}
Define \(v_n = a_n (u_n - u_{n-1})\) for \(n \in \mathbb{Z}\). Then \(v = \{v_j\}_{j \in \mathbb{Z}}\) satisfies
\begin{align}\label{eqn:hv=ev-app}
    -v_{n+1} + 2 v_n - v_{n-1} = \frac{E}{a_n} v_n.
\end{align}
The free Pr\"ufer variables for \(v_n\), with phases \(\theta_n(E) \in \R\) and amplitudes \(r_n(E) \ge 0\), are defined by 
\begin{align}\label{eqn:Prufer-free}
   r_n(E)\begin{pmatrix}
       \cos \theta_n(E) \\
       \sin \theta_n(E)
   \end{pmatrix}=\begin{pmatrix}
       v _n \\
       v _{n-1}
   \end{pmatrix}, \ \ n\ge 0.
\end{align}
We take the normalized initial value \((a_0u_0, u_{-1}) = (\cos \beta_0, \sin \beta_0)\). Then the initial value \((v_0, v_{-1})\) is given by 
\begin{align}
\begin{pmatrix}\label{eqn:v-ini-prufer}
        v_{0} \\v_{-1}
    \end{pmatrix}= r_0 \begin{pmatrix}
        \cos \theta_0 \\ \sin \theta_0
    \end{pmatrix}=\begin{pmatrix}
        1 & -a_0 \\
        1 & E - a_0
    \end{pmatrix}\begin{pmatrix}
        \cos \beta_0 \\ \sin \beta_0
    \end{pmatrix} .
\end{align}
A direct computation shows that there is a one-to-one correspondence between \(\beta_0 \in [0, \pi)\) and \(\theta_0 \in [0, \pi)\).

Restrict the equations \eqref{eqn:hu=eu-app} and \eqref{eqn:hv=ev-app} to the interval \([0, N-1]\), subject to the boundary conditions 
\begin{align} \label{eqn:uBC}
u_{-1} = a_0u_0 \tan \beta_0, \quad  u_N = u_{N-1}, \quad \beta_0 \in [0, \pi) . 
\end{align}
It is well known (see, e.g., \cite{pastur1992book}) that for i.i.d. \(\{a_n(\omega)\}\), the integrated density of states \(\mathcal N(E)\), as defined by the almost sure limit in \eqref{eqn:ids}, is independent of the boundary conditions \eqref{eqn:uBC}. We choose the convenient right boundary condition \(u_N = u_{N-1}\), which corresponds to \(v_N = 0\) for \(v_n\). For \(E > 0\), the pair \((E,\{u_n\}_{n=0}^{N-1})\) is an eigenpair of the system \eqref{eqn:hu=eu-app} with the boundary condition \eqref{eqn:uBC} if and only if \(\{v_n\}_{n=-1}^N\) is generated by the initial value \eqref{eqn:v-ini-prufer} under the iteration of \(B^E_n\) in \eqref{eqn:v-cocycle} and satisfies
\begin{align}\label{eqn:thetaN-ev}
   \cot \theta_N(E) = 0 \Longleftrightarrow \theta_N(E) = \frac{\pi}{2} \mod \pi.   
\end{align}

The well-known oscillation theorem for one-dimensional differential and second-order difference operators (see \cite{hartman1964ode,pastur1992book,simon2005osc}) relates the number of eigenvalues below a given energy to the zeros of the solution, which, in terms of Pr\"ufer variables, correspond to phase values that are multiples of \(\pi\). This connection ultimately enables the computation of the IDS, for example, for an ergodic operator, via the average of the phase variables.

More precisely, in the one-dimensional Jacobi operator setting, one can show (see, e.g., \cite{jitomirskaya2003deloc}) that  all eigenvalues \(E_j\) are simple and can therefore be arranged in strictly increasing order as \(0 \le E_1 < E_2 < E_3 < \cdots < E_N\), and for \(n > 0\),
\begin{align}\label{eqn:osc}
    \lim_{E \to -\infty} \theta_n(E) = 0, \ \ \frac{\partial}{\partial E}\theta_N(E) > 0. 
\end{align} 
Combining \eqref{eqn:osc} with \eqref{eqn:thetaN-ev} gives 
\begin{align} \label{eqn:Ej-exp}
  \theta_N(E_j) = \frac{\pi}{2} + \pi(j-1) = \frac{\pi}{2}, \ j = 1,\cdots,N.  
\end{align} 
See Figure~\ref{fig:thetaE}. Therefore, for any \(E \in \R\),
\begin{align*}
    \Big|\frac{1}{\pi}\theta_N(E) - \#\{\ {\rm eigenvalues}\  E_j \ {\rm such\ that}\ E_j \le E\}\Big| \le \frac{1}{2}.
\end{align*}
Note that in this case \(\theta_N(E_1) = \frac{\pi}{2}\) follows from the right boundary condition \eqref{eqn:thetaN-ev} and \eqref{eqn:uBC}. In general, if \(\beta_1 \neq \frac{\pi}{2}\), one can replace the bound \(1/2\) by \(2\).

Combining \eqref{eqn:ids} with the preceding discussion, we obtain the following equivalent definition of the IDS in terms of \(\theta_N\):
\begin{align}\label{eqn:ids-theta}
    \cN(E) = \lim_{N \to \infty}\frac{1}{\pi N}\E \theta_{N}(E).
\end{align}

The monotonicity in \eqref{eqn:osc} and the separation of \(\theta_N\) at \(E_j\) imply that all eigenvalues \(E_j\) strictly interlace with those energies for which \(\theta_N\) is a multiple of \(\pi\). In other words, for any \(j \in \{1, \cdots, N\}\), there exists \(E_c \in (E_j, E_{j+1})\) such that
\begin{align}
\theta_N(E_c) = 0 \mod \pi. \label{eqn:theta-interlace}
\end{align}
For the div-grad model, the min-max principle ensures that all eigenvalues are positive. Hence, for all \(n > 0\), the smallest eigenvalue satisfies \(E_1 > 0\). We can also extend the definition of \(\theta_n(E)\) to \(E = 0\). It follows from \eqref{eqn:v-ini-prufer} that
\begin{align}
v_n(0) = v_{n-1}(0) \Longrightarrow \theta_n(0) = \frac{\pi}{4}, \ \ n \ge 0. \label{eqn:theta-0}
\end{align}

\begin{figure} 
\centering
\begin{tikzpicture}[>=stealth,scale=0.6]

  % ---------------- extents & parameter ----------------
  \def\xmin{-1}       % extend x-axis to the left
  \def\xmax{10}       % right end of x-axis (square-ish aspect)
  \def\ymax{10.2}     % top of y-axis (just above 3*pi)

  % Choose 'a' so that 3*pi crossing ~ E=9
  \pgfmathsetmacro{\a}{ln(12)/9}

  % ---------------- intersection abscissas ----------------
  % y(E) = (pi/4) * exp(aE)
  \pgfmathsetmacro{\Epi}{ln(4)/\a}           % L = pi
  \pgfmathsetmacro{\E2pi}{ln(8)/\a}          % L = 2pi
  \pgfmathsetmacro{\E3pi}{ln(12)/\a}         % L = 3pi

  % For E1 (pi/2), E2 (3pi/2), E3 (5pi/2):
  \pgfmathsetmacro{\Ehalf}{ln(2)/\a}         % L = pi/2
  \pgfmathsetmacro{\Ethreehalf}{ln(6)/\a}    % L = 3pi/2
  \pgfmathsetmacro{\Efivehalf}{ln(10)/\a}    % L = 5pi/2

  % --- End the curve slightly above 3*pi (you can adjust this manually) ---
  \pgfmathsetmacro{\Eend}{ln(12.5)/\a}       % extend a bit to the right

  % Inline curve value: y(E) = (pi/4) * exp(a*E)
  \newcommand{\thetaof}[1]{(pi/4)*exp(\a*(#1))}

  % ---------------- axes ----------------
  \draw[thick,->] (\xmin,0) -- (\xmax,0) node[below right] {$E$};
  \draw[thick,->] (0,0) -- (0,\ymax); % (removed the y-axis node label)

  % ---------------- dashed reference lines ONLY at pi, 2pi, 3pi ----------------
  \draw[dashed] (\xmin,  pi) -- (\xmax,  pi);
  \draw[dashed] (\xmin,2*pi) -- (\xmax,2*pi);
  \draw[dashed] (\xmin,3*pi) -- (\xmax,3*pi);

  % ---------------- y-axis labels to the LEFT (only: pi/4, pi/2, pi, 2pi, 3pi) ----------------
  % Place pi/4 down-left; pi/2 up-left; others with slight vertical nudges.
  \node[below left] at (0,  pi/4) {$\frac{\pi}{4}$};
  \node[ left] at (0,  pi/2) {$\frac{\pi}{2}$};
  \node[anchor=east, yshift= 5pt] at (-0.1,    pi) {$\pi$};
  \node[anchor=east, yshift=5pt] at (-0.1,  2*pi) {$2\pi$};
  \node[anchor=east, yshift= 5pt] at (-0.1,  3*pi) {$3\pi$};

  % ---------------- small dots on the y-axis at pi/4 and pi/2 ----------------
  \fill (0, pi/4) circle (2.5pt);
  \fill (0, pi/2) circle (2.5pt);

  % ---------------- the curve y = (pi/4) * e^{a E} ----------------
  \draw[cyan!60!blue, very thick, domain=\xmin:\Eend, samples=400]
    plot (\x, {\thetaof{\x}});

  % ---------------- vertical guides at E1 (pi/2), E2 (3pi/2), E3 (5pi/2) ----------------
  \draw[red!70, dashed, thick] (\Ehalf,0) -- (\Ehalf, {\thetaof{\Ehalf}});
  \draw[red!70, dashed, thick] (\Ethreehalf,0) -- (\Ethreehalf, {\thetaof{\Ethreehalf}});
  \draw[red!70, dashed, thick] (\Efivehalf,0) -- (\Efivehalf, {\thetaof{\Efivehalf}});

  % x-axis labels under the guides
  \node[below] at (\Ehalf,0) {$E_1$};
  \node[below] at (\Ethreehalf,0) {$E_2$};
  \node[below] at (\Efivehalf,0) {$E_3$};

  % ---------------- intersection markers on the curve ----------------
  \fill[cyan!60!blue] (\Ehalf,      pi/2)   circle (1.4pt);
  \fill[cyan!60!blue] (\Ethreehalf, 3*pi/2) circle (1.4pt);
  \fill[cyan!60!blue] (\Efivehalf,  5*pi/2) circle (1.4pt);

  % ---------------- Ec at the pi intersection (green), label moved BELOW ----------------
  \fill[green!60!black] (\Epi, pi) circle (2.5pt);
  \node[anchor=north, text=green!60!black, yshift=-2pt] at (\Epi, pi) {$E_c$};

  % ---------------- curve label near the end of the curve ----------------
  \node[anchor=west, text=cyan!60!blue, font=\bfseries]
    at (\Eend - 0.4, {1.04*\thetaof{\Eend}}) {$\theta_N(E)$};

\end{tikzpicture}
\caption{Interlacing property of $\theta_N(E)$: 
$0 \le E_1 < E_2 < \cdots$, with $\theta_N(E_j)=\frac{\pi}{2}\pmod{\pi}$, and there exists $E_c\in(E_j,E_{j+1})$ such that 
$\theta_N(E_c)\equiv 0\pmod{\pi}$.}
\label{fig:thetaE}
\end{figure}
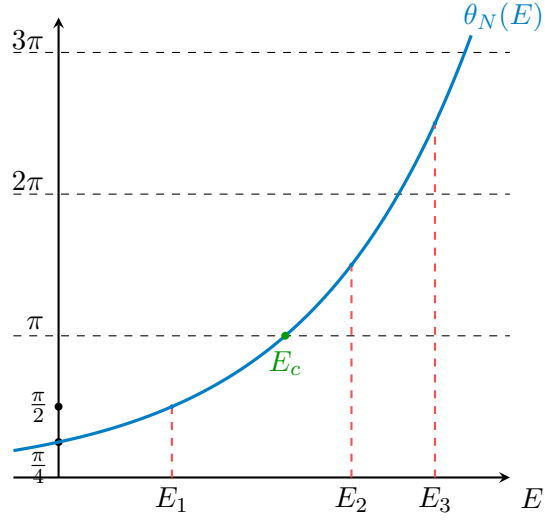

The relation \eqref{eqn:ids-theta} between the IDS and \(\theta_N(E)\) does not directly yield the asymptotic behavior as \(E \to 0^+\). To address this, we introduce an additional modification involving \(\sqrt{E}\) in the polar coordinates \eqref{eqn:Prufer-free}, as defined in \eqref{eqn:eta}-\eqref{eqn:modi-prufer}:
\begin{align}\label{eqn:modi-prufer-app}
    \rho_n(E)
    \begin{pmatrix}
        \cos \chi_n(E) \\
        \sin \chi_n(E)
    \end{pmatrix}
    = P
    \begin{pmatrix}
        v_n \\ v_{n-1}
    \end{pmatrix}, \quad n \ge 0,
\end{align}
where  
\begin{align}\label{eqn:eta-app}
    P = \begin{pmatrix}
        1 & -\cos \eta \\
        0 & \sin \eta
    \end{pmatrix}, 
      \quad 
    \eta(E) = \cos^{-1}\!\Big(1 - \frac{1}{2\kappa} E\Big) = \frac{\sqrt{E}}{\sqrt{\kappa}} + O(E^{3/2}),
\end{align}
and \(\kappa = \big[\mathbb{E}(a_0^{-1})\big]^{-1}\).  
The initial variables are determined by 
\begin{align*}
     \rho_{0}(E)\begin{pmatrix}
        \cos \chi_0(E) \\
        \sin \chi_0(E)
    \end{pmatrix} = r_0\begin{pmatrix}
        \cos \theta_0 - \cos \eta \sin \theta_0 \\ \sin \eta \sin \theta_0
    \end{pmatrix}, \ \ {\rm with}\  \chi_0(E) \in [0,\pi].
\end{align*}
For any \(E > 0\), there is a one-to-one correspondence between \(\theta_0 \in [0,\pi)\) and \(\chi_0 \in [0,\pi)\), and hence between \(\beta_0 \in [0,\pi)\) and \(\chi_0 \in [0,\pi)\).

\begin{lemma} \label{lem:ids-chi}
Let \(\cN(E)\) denote the IDS of \(H_\omega\) in \eqref{eqn:hu=eu}. For \(E > 0\),
\begin{align}\label{eqn:IDS-chi}
\cN(E) = \lim_{N \to \infty} \frac{1}{\pi N}\E \big[\chi_N(E)\big].
\end{align}
\end{lemma}

\begin{remark}
This result was stated in \cite{lifshits88book} for a chain with random force constants (the div-grad model) without proof. We sketch the argument, which essentially follows \cite[\S7.2]{lifshits88book} for an isotopically disordered chain.
\end{remark}

\begin{proof}[Proof of Lemma~\ref{lem:ids-chi}]
The free and modified Pr\"ufer variables, \((r_n,\theta_n)\) and \((\rho_n,\chi_n)\)  respectively,  are linked through \eqref{eqn:Prufer-free} and \eqref{eqn:modi-prufer-app} as 
\begin{align*}
     \rho_{n}(E)\begin{pmatrix}
        \cos \chi_n(E) \\
        \sin \chi_n(E)
    \end{pmatrix}= r_n(E) \, P \begin{pmatrix}
       \cos \theta_n(E) \\
       \sin \theta_n(E)
   \end{pmatrix}.  
\end{align*}
Comparing the phase variables on each side gives 
\begin{align}
    \cot \chi_n = \frac{\cos \theta_n - \cos \eta \, \sin \theta_n}{\sin \eta \, \sin \theta_n}.
    \label{eqn:chi-theta}
\end{align}
Recall that \(\theta_n(E)\) is analytic in \(E\) and satisfies \(\theta_n(0) = \pi/4\) as in \eqref{eqn:theta-0}. Expanding near \(E \to 0^+\) gives 
\begin{align*}
   \cot \chi_n(E) = \frac{\cot \theta_n - \cos \eta}{\sin \eta} 
   = \frac{\cot \frac{\pi}{4} + O(E) - \big(1 - \frac{\kappa}{2}E\big)}{\sqrt{\kappa}\,\sqrt{E} + O(E)} 
   = O(\sqrt{E}).
\end{align*}
Hence, \(\cot \chi_n(E) \to 0\) as \(E \to 0^+\). Define \(\chi_n(0) := \lim_{E \to 0^+}\chi_n(E)\). Then for all \(n > 0\),
\begin{align*}
    \chi_n(0) = \frac{\pi}{2} \mod \pi.  
\end{align*}

From \eqref{eqn:Ej-exp}, it follows that at an eigenvalue, \(\chi_n\) satisfies 
\begin{align*}
  \cot \chi_n(E_j) = -\cot \eta(E_j) \Longrightarrow \chi_n(E_j) = \pi - \eta(E_j) \mod \pi.
\end{align*}

Since \(\sin \eta > 0\), \eqref{eqn:chi-theta} implies that if \(\chi_n = 0 \ {\rm mod}\ \pi\), then \(\theta_n = 0 \ {\rm mod}\ \pi\), and vice versa. Hence, at \(\chi_n = 0 \ {\rm mod}\ \pi\), one has  
\begin{align}
    \frac{\sin \chi_n}{\sin \theta_n}\Big|_{\chi_n = 0 \ {\rm mod}\ \pi} = \Big[\sin \eta \, \frac{\cos \chi_n}{\cos \theta_n} + \cos \eta \, \frac{\sin \chi_n}{\cos \theta_n}\Big]\Big|_{\chi_n = 0 \ {\rm mod}\ \pi} = \pm \sin \eta. \label{eqn:chi_nmodpi}
\end{align}

A direct computation by differentiating \eqref{eqn:chi-eta} with respect to \(E\), combined with \eqref{eqn:chi_nmodpi}, gives 
\begin{align}
  \frac{d\chi_n}{dE}\Big|_{\chi_n = 0 \ {\rm mod}\ \pi} = \sin \eta \, \frac{d\theta_n}{dE}\Big|_{\chi_n = 0 \ {\rm mod}\ \pi}. \label{eqn:D-chi-theta}
\end{align}
Combined with \eqref{eqn:theta-interlace}, \eqref{eqn:D-chi-theta} implies that \(E_j\) also interlaces with those energies for which \(\chi_n\) is a multiple of \(\pi\); see Figure~\ref{fig:chi-inter}.

\begin{figure}
    \centering
  
\begin{tikzpicture}[>=stealth,scale=0.8]

  % --- extents (tweak to taste) ---
  \def\xmin{-1}
  \def\xmax{10}
  \def\ymax{10.2} % just above 3*pi

  % --- axes ---
  \draw[thick,->] (\xmin,0) -- (\xmax,0) node[below right] {$E$};
  \draw[thick,->] (0,0) -- (0,\ymax);

  % --- dashed reference lines ONLY at pi, 2pi, 3pi ---
  \draw[dashed] (\xmin,  pi) -- (\xmax,  pi);
  \draw[dashed] (\xmin,2*pi) -- (\xmax,2*pi);
  \draw[dashed] (\xmin,3*pi) -- (\xmax,3*pi);

  % --- y-axis labels ONLY at pi/2, pi, 2pi, 3pi (to the left of the axis) ---
  % Slight vertical nudges to improve readability around dashed lines.
  \node[anchor=east, yshift= 5pt] at (-0.1,  pi/2) {$\dfrac{\pi}{2}$};
  \node[anchor=east, yshift= 5pt] at (-0.1,    pi) {$\pi$};
  \node[anchor=east, yshift=5pt] at (-0.1,  2*pi) {$2\pi$};
  \node[anchor=east, yshift= 5pt] at (-0.1,  3*pi) {$3\pi$};
 \fill (0, pi/2) circle (2.5pt);

 % --- parameters for the wiggly exponential curve chi_N(E) ---
% Choose 'a' so that 3*pi is reached near E≈9 (like your previous figure):
\pgfmathsetmacro{\a}{ln(6)/9}     % base growth
\pgfmathsetmacro{\eps}{0.16}      % wiggle amplitude (0.10–0.20 works well)
\pgfmathsetmacro{\w}{2.4}         % wiggle frequency (radians per E unit)
\pgfmathsetmacro{\lam}{0.10}      % damping of wiggles
\pgfmathsetmacro{\Ethreepi}{ln(6)/\a} % ≈ 9 (where (pi/2)*exp(aE)=3*pi)
\pgfmathsetmacro{\Eend}{\Ethreepi + 0.25} % extend just a bit beyond 3*pi

% Inline function: chi_N(E) = (pi/2)*exp(a*E)*(1 + eps*sin(ωE)*exp(-λE))
% (deg(...) converts radians to degrees for TikZ's sin())
\newcommand{\chiN}[1]{%
  (pi/2)*exp(\a*(#1))*(1 + \eps*sin(deg(\w*(#1)))*exp(-\lam*(#1)))%
}

% --- the blue curve (wiggly exponential), passing (0, pi/2) ---
\draw[cyan!60!blue, very thick, domain=\xmin:\Eend, samples=600, smooth]
  plot (\x, {\chiN{\x}});

% ===== Label at the blue curve end =====
% choose the x-location to place the label; default to the plotted end \Eend
% (if \Eend is not defined in this figure, see the fallback note below)
\pgfmathsetmacro{\xBlueLbl}{\Eend}
\pgfmathsetmacro{\yBlueLbl}{\chiN{\xBlueLbl}}

\node[text=cyan!60!blue, anchor=west, xshift=4pt, yshift=2pt, font=\bfseries]
  at (\xBlueLbl, \yBlueLbl) {$\chi_N(E)$};

% --- RED dashed curve under pi: y = pi - c * E^2 ---
% Choose a small c; 0.02–0.06 works well. You can tune it.
% Pick a small c
\pgfmathsetmacro{\ceta}{0.05}

% End the red curve somewhere before the right edge (no need for \Epi)
\pgfmathsetmacro{\Eredend}{4.0} % or e.g., \pgfmathsetmacro{\Eredend}{min(\Eend, 3.0)}

\draw[red!70, dashed, thick, domain=0:\Eredend, samples=400, smooth]
  plot (\x, {pi - \ceta*\x*\x});

% --- RED dashed curve under 2*pi: y = 2*pi - c2 * E^2 ---
\pgfmathsetmacro{\cetaTwo}{0.025}   % adjust this
\pgfmathsetmacro{\EredendTwo}{7.0}  % how far to draw this red curve
\draw[red!70, dashed, thick, domain=0:\EredendTwo, samples=400, smooth]
  plot (\x, {2*pi - \cetaTwo*\x*\x});

% (optional) tiny label near the 2*pi red curve
% \node[red!70] at (1.2, {2*pi - \cetaTwo*(1.2)*(1.2) + 0.18}) {$2\pi - \eta$};

% --- RED dashed curve under 3*pi: y = 3*pi - c3 * E^2 ---
\pgfmathsetmacro{\cetaThree}{0.015}   % adjust this
\pgfmathsetmacro{\EredendThree}{9.2}  % how far to draw this red curve
\draw[red!70, dashed, thick, domain=0:\EredendThree, samples=400, smooth]
  plot (\x, {3*pi - \cetaThree*\x*\x});

% (optional) tiny label near the 3*pi red curve
% \node[red!70] at (1.4, {3*pi - \cetaThree*(1.4)*(1.4) + 0.18}) {$3\pi - \eta$};

% ===== Labels at the red curve ends =====
% y-values at the chosen domain ends
\pgfmathsetmacro{\yredend}{pi - \ceta*\Eredend*\Eredend}
\pgfmathsetmacro{\yredendTwo}{2*pi - \cetaTwo*\EredendTwo*\EredendTwo}
\pgfmathsetmacro{\yredendThree}{3*pi - \cetaThree*\EredendThree*\EredendThree}

% place the labels a hair to the right of the end points
\node[red!70, anchor=west, xshift=-10pt,yshift=-8pt] at (\Eredend,      \yredend)      {$\pi - \eta(E)$};
\node[red!70, anchor=west, xshift=-20pt,,yshift=-8pt] at (\EredendTwo,   \yredendTwo)   {$2\pi - \eta(E)$};
\node[red!70, anchor=west, xshift=-15pt,yshift=-8pt] at (\EredendThree, \yredendThree) {$3\pi - \eta(E)$};

% ===== Approximate E1 intersection and vertical guide =====
% Pick a guess for the intersection abscissa (adjust by eye as needed)
\pgfmathsetmacro{\EoneGuess}{2.7}  % <-- tweak this until the dot visually sits on the crossing

% y-value on the blue curve at E1 (we'll use that as the intersection point)
\pgfmathsetmacro{\Yone}{\chiN{\EoneGuess}}

% (Optional) red curve value at the same x, just to sanity-check (not used)
%\pgfmathsetmacro{\YoneRed}{pi - \ceta*\EoneGuess*\EoneGuess}

% Draw the vertical dashed red guide from the "intersection" down to x-axis
\draw[red!70, dashed, thick] (\EoneGuess, 0) -- (\EoneGuess, \Yone);

% Mark the intersection point with a small dot (use red to match the guide or blue to match the curve)
\fill[red!70] (\EoneGuess, \Yone) circle (1.5pt);

% Label E_1 on the x-axis under the guide
\node[below] at (\EoneGuess, 0) {$E_1$};

% ===== Approximate E2 (intersection with the 2*pi - c2*E^2 red curve) =====
% Pick your guess visually; adjust as needed
\pgfmathsetmacro{\EtwoGuess}{5.8}   % <-- tweak this by eye
% y on the blue curve at EtwoGuess
\pgfmathsetmacro{\Ytwo}{\chiN{\EtwoGuess}}
% Draw the vertical dashed guide and mark the intersection
\draw[red!70, dashed, thick] (\EtwoGuess, 0) -- (\EtwoGuess, \Ytwo);
\fill[red!70] (\EtwoGuess, \Ytwo) circle (1.5pt);
% Label on x-axis
\node[below] at (\EtwoGuess,0) {$E_2$};

% ===== Approximate E3 (intersection with the 3*pi - c3*E^2 red curve) =====
\pgfmathsetmacro{\EthreeGuess}{8.2} % <-- tweak this by eye
% y on the blue curve at EthreeGuess
\pgfmathsetmacro{\Ythree}{\chiN{\EthreeGuess}}
% Draw the vertical dashed guide and mark the intersection
\draw[red!70, dashed, thick] (\EthreeGuess, 0) -- (\EthreeGuess, \Ythree);
\fill[red!70] (\EthreeGuess, \Ythree) circle (1.5pt);
% Label on x-axis
\node[below] at (\EthreeGuess,0) {$E_3$};

% ===== Green Ec markers at chi_N(E) = pi and chi_N(E) = 2*pi =====
% Initial analytic guesses (adjust by eye if needed)
\pgfmathsetmacro{\EcgOne}{0.6/\a}  % ~ intersection with pi
\pgfmathsetmacro{\EcgTwo}{1.46/\a}  % ~ intersection with 2*pi

% Place green dots "on top of" the intersections (use y = pi and y = 2*pi)
\fill[green!60!black] (\EcgOne, pi)   circle (1.8pt);
\node[anchor=south west, text=green!60!black, yshift=1.5pt,xshift=-15pt]
  at (\EcgOne, pi) {$E_c$};

\fill[green!60!black] (\EcgTwo, 2*pi) circle (1.8pt);
\node[anchor=south west, text=green!60!black, yshift=1.5pt,xshift=-15pt]
  at (\EcgTwo, 2*pi) {$E_c$};
  
\end{tikzpicture}

\caption{
Interlacing property of $\chi_N(E)$: $\chi_N$ increases at each $E_c$. 
There is a unique $E_j$ with $\chi_N(E_j)=-\eta(E_j)\in(j\pi,(j+1)\pi)$, 
and for $E_j<E<E_{j+1}$, $\chi_N(E)$ lies between $j\pi-\eta(E)$ and $(j+1)\pi-\eta(E)$, 
so $|\chi_N(E)-j\pi|<\pi$.
}
   \label{fig:chi-inter}
\end{figure}
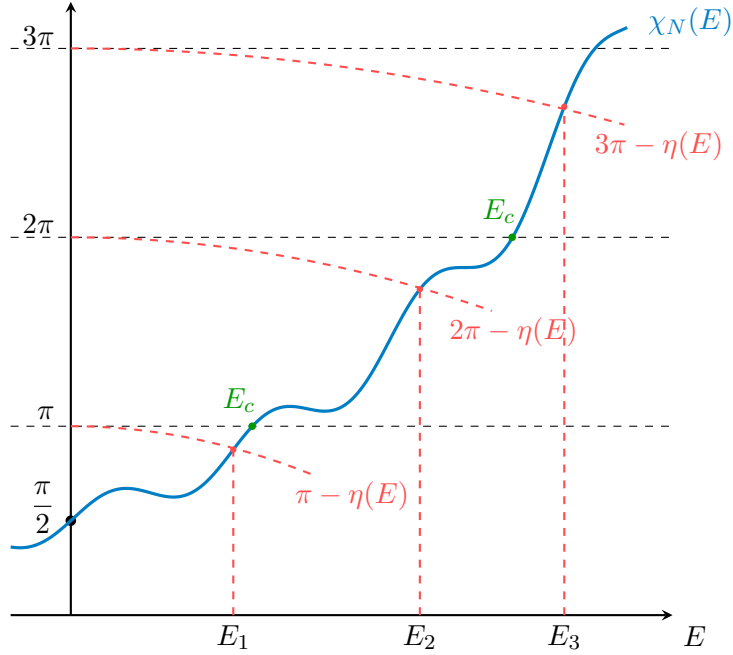

Since the smallest eigenvalue \(E_1 > 0\), for any \(n > 0\), \(\chi_n(0) = \frac{\pi}{2}\).
In particular, at \(n = N\),
\begin{align*}
     \chi_{N}(E_1) = \pi - \eta(E_1) \in (0,\pi).
\end{align*}
Inductively, for each \(j \ge 1\), there is exactly one eigenvalue \(E_j\) such that 
\begin{align*}
    \chi_{N}(E_j) = -\eta(E_j) \in (j\pi,(j+1)\pi),
\end{align*}
which implies that if \(E_j < E < E_{j+1}\), then 
\begin{align*}
    j\pi - \eta(E) < \chi_N(E) < (j+1)\pi - \eta(E)
\Longleftrightarrow |\chi_N(E) - j\pi| < \max\{\eta(E), \pi - \eta(E)\} < \pi.
\end{align*}
Therefore,
\begin{align*}
    \Big|\frac{1}{\pi}\chi_N(E) - \#\{\ {\rm eigenvalues}\ E_j\ {\rm such\ that}\ E_j \le E\}\Big| < 1.
\end{align*}
Combining this with \eqref{eqn:ids} proves \eqref{eqn:IDS-chi} for \(E > 0\).
\end{proof}

%%%%%%%%%%%%%%%%%%%%%%%%%%%%%%%%%%%%%%%%%%%%%%%%%%%%%%%%%%%%%%%%%%%%%%%%%%%%%%%%%%%%%%%%%%%%%%%%%%%%%%%%%%%%%%%

Now we study the asymptotic behavior of \(\chi_N(E)\) as \(E \to 0^+\). We rewrite the angle variables as
\begin{align}\label{eqn:chi-eta}
 \chi_n(E) = n \eta(E) + \big[\chi_n(E) - n \eta(E)\big], \ \ 0 \le n \le N.
\end{align}
Then \eqref{eqn:IDS-chi} and \eqref{eqn:eta-app} imply
\begin{align*}
    \mathcal N(E) = \lim_{N \to \infty}\frac{\E(\chi_N)}{\pi N} 
    = \frac{1}{\pi \sqrt{\kappa}}\sqrt{E} + O(E^{3/2}) + \lim_{N \to \infty}\frac{\E(\chi_N - N \eta)}{\pi N}.
\end{align*}
Theorem~\ref{thm:NE-root} then follows from the following lemma.
\begin{lemma}\label{lem:zeta-order}
For \(E \to 0^+\),
\begin{align}\label{eqn:zeta-order}
\lim_{N \to \infty}\frac{1}{N}\E\big[\chi_N(E) - N \eta(E)\big] = O(E).
\end{align}
\end{lemma}
The remainder of this section is devoted to the proof of this lemma.
%%%%%%%%%%%%%%%%%%%%%%%%%%%%%%%%%%%%%%%%%%%%%%%%%%%%%%%%%%%%%%%%%%%%%%%%%%%%%%%%%%%%%%%%%%%%%%%%%%%%%%%%
\begin{proof}[Proof of Lemma~\ref{lem:zeta-order}]
Let
\[
    B^E_j =
    \begin{pmatrix}
        2 - \frac{E}{a_j} & -1 \\
        1 & 0
    \end{pmatrix},
    \quad \text{and} \quad
    F_n^E = B^E_{n-1} \cdots B^E_0, \quad n \ge 1,\; j \in \mathbb{Z}.
\]
be the transfer matrix for \(v_n\) as in \eqref{eqn:v-cocycle}. For i.i.d. \(a_n\), denote the expectation of \(B_0^E\) by
\begin{align*}
  \overline{B} = \E(B^E_0) = \begin{pmatrix}
        2 - \E(a_0^{-1}) \cdot E & -1 \\
        1 & 0
    \end{pmatrix} = \begin{pmatrix}
        2 - \kappa^{-1} \cdot E & -1 \\
        1 & 0
    \end{pmatrix}.
\end{align*}

For \(E > 0\), a direct computation shows that \(\overline{B}\) can be conjugated by \(P\) in \eqref{eqn:eta-app} to a rotation with angle \(\eta = \eta(E)\):
\begin{align*}
P \overline{B} P^{-1} = R_{\eta} := \begin{pmatrix}
        \cos \eta(E) & -\sin \eta(E) \\
        \sin \eta(E) & \cos \eta(E)
    \end{pmatrix}.
\end{align*}
Conjugating \(B_n^E\) by \(P\) yields
\begin{align}\label{eqn:Bn-conj}
   P B_n^E P^{-1} = P \overline{B} P^{-1} + P(B_n^E - \overline{B})P^{-1} = R_{\eta} + Y_n^E,
\end{align}
where
\begin{align*}
    Y_n^E = P(B_n^E - \overline{B})P^{-1} 
    =& \begin{pmatrix}
        1 & -\cos{\eta} \\
        0 & \sin{\eta}
    \end{pmatrix} \begin{pmatrix}
          E (\kappa^{-1} - a_n^{-1}) & 0 \\
        0 & 0
    \end{pmatrix} \begin{pmatrix}
        1 & \frac{\cos{\eta}}{\sin{\eta}} \\
        0 & \frac{1}{\sin{\eta}}
    \end{pmatrix} \\
    =& E (\kappa^{-1} - a_n^{-1}) \begin{pmatrix}
        1 & \frac{\cos{\eta}}{\sin{\eta}} \\
        0 & 0
    \end{pmatrix}.
\end{align*}
From \eqref{eqn:Bn-conj} and the iteration of \(v_n\) in \eqref{eqn:v-cocycle}, it follows that
\begin{align}
   P \begin{pmatrix}
        v_{n+1} \\ v_n
    \end{pmatrix} = P B_n^E \begin{pmatrix}
        v_n \\ v_{n-1}
    \end{pmatrix} 
    = (P B_n^E P^{-1})\, P \begin{pmatrix}
        v_n \\ v_{n-1}
    \end{pmatrix} 
    = (R_{\eta} + Y_n^E)\, P \begin{pmatrix}
        v_n \\ v_{n-1}
    \end{pmatrix}. \label{eqn:vn-conj}
\end{align}
Expressing \eqref{eqn:vn-conj} in terms of Pr\"ufer variables using \eqref{eqn:modi-prufer-app} leads to
\begin{align*}
     \rho_{n+1}\begin{pmatrix}
        \cos \chi_{n+1} \\
        \sin \chi_{n+1}
    \end{pmatrix} =& R_\eta \, \rho_n \begin{pmatrix}
        \cos \chi_n \\
        \sin \chi_n
    \end{pmatrix} + Y_n^E \, \rho_n \begin{pmatrix}
        \cos \chi_n \\
        \sin \chi_n
    \end{pmatrix} \\
    =& \rho_n \begin{pmatrix}
        \cos (\chi_n + \eta) \\
        \sin (\chi_n + \eta)
    \end{pmatrix} + \rho_n \frac{E}{\sin \eta} (\kappa^{-1} - a_n^{-1}) \begin{pmatrix}
         \sin (\eta + \chi_n)  \\
       0
    \end{pmatrix}.
\end{align*}
Let \(Q_n = \frac{E}{\sin \eta(E)} (\kappa^{-1} - a_n^{-1})\) as in \eqref{eqn:Qn}. The above equations are exactly the iteration in Proposition~\ref{prop:rho-chi-Qn}. 

Taking the ratio of both sides of \eqref{eqn:rho-n-iter} yields the recurrence relation for the phase variables:
\begin{align}\label{eqn:chin-rec}
    \cot \chi_{n+1} = \cot (\chi_n + \eta) + Q_n.
\end{align}
From \eqref{eqn:chin-rec}, it follows that
\begin{align*}
      \tan\big[\chi_{n+1} - (\chi_n + \eta)\big]  
      = \frac{\cot(\chi_n + \eta) - \cot \chi_{n+1}}{1 + \cot \chi_{n+1}\cot (\chi_n + \eta)} 
        = -\frac{Q_n}{2}\frac{1 - \cos 2(\chi_n + \eta)}{1 + \frac{Q_n}{2}\sin 2(\chi_n + \eta)}.
\end{align*}
Since \(Q_n \sim \sqrt{E}\) (uniformly in \(n\)) as \(\sqrt{E} \to 0\), as discussed in Proposition~\ref{prop:rho-chi-Qn}, we expand the last term in \(Q_n\) up to first order:
\begin{align*}
    \frac{1 - \cos 2(\chi_n + \eta)}{1 + \frac{Q_n}{2}\sin 2(\chi_n + \eta)} &= \big[1 - \cos 2(\chi_n + \eta)\big]\big[1 + O(Q_n \sin 2(\chi_n + \eta))\big] \\
    &= 1 - \cos 2(\chi_n + \eta) + O(Q_n).
\end{align*}
Hence,
\begin{align*}
   \tan\big[\chi_{n+1} - (\chi_n + \eta)\big] = \frac{Q_n}{2}\big[\cos 2(\chi_n + \eta) - 1\big] + O(Q_n^2).
\end{align*}
The remainder \(O(Q_n^2) \sim O(E)\) is uniform in \(n\) as \(E \to 0\), since all coefficients of \(Q_n^k, k \ge 2\) are uniformly bounded in \(n\). As a result, together with the expansion \(\tan^{-1}(x) = x + O(x^3)\), this implies
\begin{align}
     \chi_{n+1} - (\chi_n + \eta) =& \tan^{-1}\Big[\frac{1}{2}Q_n\big(\cos 2(\chi_n + \eta) - 1\big) + O(Q_n^2)\Big] \notag \\
  =& \frac{1}{2}Q_n\big[\cos 2(\chi_n + \eta) - 1\big] + O(E). \label{eqn:zeta-diff}
\end{align}
Taking the expectation over all random variables in \eqref{eqn:zeta-diff} yields
\begin{align*}
    \E\big[\chi_{n+1} - (\chi_n + \eta)\big] = O(E), \quad (\text{uniformly in } n),
\end{align*}
where we used the independence of \(Q_n\) and \(\chi_n\), \(E Q_n=0\), together with \eqref{eqn:Qn} and \eqref{eqn:Qn-bound} as discussed in Proposition~\ref{prop:rho-chi-Qn}.
Finally, summing over \(1 \le n \le N-1\) gives
\begin{align*}
     \E(\chi_N - N\eta) = \E(\chi_0) + \sum_{n=0}^{N-1}\E\big[\chi_{n+1} - (\chi_n + \eta)\big] = \chi_0 + N O(E),
\end{align*}
since \(\chi_0 \in [0,\pi]\) is nonrandom. Dividing by \(N\) and taking the limit as \(N \to \infty\) proves \eqref{eqn:zeta-order}.
\end{proof}

\subsection{Asymptotic formulas for the Lyapunov exponent }\label{sec:lyp-asym}
As noted in Remark~\ref{rmk:lyp-asym}, the expansion in \cite[Theorem 14.6, Part (ii)]{pastur1992book}, i.e., Theorem~\ref{thm:LE-linear}, for the div-grad model was derived via a brief substitution in the Schr\"odinger case formulas. Here, using the modified Pr\"ufer variables \eqref{eqn:modi-prufer-app}, we supply additional steps to make the dependence on the small energy parameter \(E\) explicit throughout the derivation.

Let \(\rho_n(E)\) denote the radial variable as in \eqref{eqn:modi-prufer-app}. It follows from \eqref{eqn:Lyp}, \eqref{eqn:F-Tn-norm}, and \eqref{eqn:Fn-rhon} that for \(E > 0\), the Lyapunov can be computed alternatively through \(\rho_n\) as 
\begin{align}\label{eqn:lyp-rhon}
    L(E) = \lim_{n \to \infty} \frac{1}{n} \E \log \frac{\rho_n}{\rho_0}.
\end{align}
We use the expansion for \(\log \frac{\rho_n}{\rho_0}\) in \eqref{eqn:rhon-rho0}–\eqref{eqn:rho-2} to estimate the asymptotic behavior of \(L(E)\) as \(E \to 0^+\). Taking expectations in \eqref{eqn:rho-0}–\eqref{eqn:rho-2} gives
\begin{align} 
    \frac{1}{n} \E \log \frac{\rho_n}{\rho_0}
    &= \frac{1}{8n} \sum_{i=0}^{n-1} \E Q_i^2 \label{eqn:rho-0app} \\
    &\quad + \frac{1}{2n} \sum_{i=0}^{n-1} \E \big[Q_i \sin 2(\chi_i + \eta)\big] \label{eqn:rho-1app} \\
    &\quad + \frac{1}{8n} \sum_{i=0}^{n-1} \E\big[-2 Q_i^2 \cos 2(\chi_i + \eta)
    + Q_i^2 \cos 4(\chi_i + \eta)\big] \label{eqn:rho-2app} \\
    &\quad + O(E^{3/2}). \label{eqn:rho-3app}
\end{align}
Recall the properties of \(Q_n\) in Proposition~\ref{prop:rho-chi-Qn}. The term in \eqref{eqn:rho-1app} vanishes because \(Q_i\) is independent of \(\chi_i\) and \(\E Q_i = 0\).

A direct computation using the expression for \(Q_n\) in \eqref{eqn:Qn} gives
\begin{align}\label{eqn:ave-Qsquare}
    \E Q_i^2 = \kappa E \cdot \E\big[(\kappa^{-1} - a_i^{-1})^2\big] + O(E^3).
\end{align}
The averaging factor \(\frac{1}{8n}\), together with the summation over \(0 \le i \le n-1\) in \eqref{eqn:rho-0app}, determines the leading coefficient in \eqref{eqn:linearLE} for the linear term in \(E\).

Finally, the expectation values of the two terms in \eqref{eqn:rho-2app} can be reduced to
\[
    \E\Bigg[\sum_{i=0}^{n-1}\big(-2Q_i^2\cos 2(\chi_i+\eta)+ Q_i^2 \cos 4(\chi_i+\eta)\big)\Bigg]
    = \E Q_0^2 \, \E\Bigg[\sum_{i=0}^{n-1}\big(-2\cos 2(\chi_i+\eta)+ \cos 4(\chi_i+\eta)\big)\Bigg],
\]
using again that \(Q_i\) is independent of \(\chi_i\). Since \(\E Q_i^2 = O(E)\), if the expectation of the above sum after factoring out \(\E Q_i^2\) is of order \(n\sqrt{E}\), then \eqref{eqn:rho-2app} will be of higher order \(O(E^{3/2})\).

Hence, Theorem~\ref{thm:LE-linear} follows from \eqref{eqn:lyp-rhon}, \eqref{eqn:rho-0app}–\eqref{eqn:rho-3app}, \eqref{eqn:ave-Qsquare}, and the following lemma.
\begin{lemma}\label{lem:rho2-est}
There exist constants \(C_1, C_2 > 0\) such that for sufficiently small \(E > 0\) and any \(n > 1/E\),
\begin{align}\label{eqn:rho2-est}
    \Bigg|\frac{1}{n}\E\Bigg[\sum_{i=0}^{n-1}\cos 2(\chi_i+\eta)\Bigg]\Bigg| \le C_1\sqrt{E}, \ \ \text{and}\ \ 
    \Bigg|\frac{1}{n}\E\Bigg[\sum_{i=0}^{n-1}\cos 4(\chi_i+\eta)\Bigg]\Bigg| \le C_2\sqrt{E}.
\end{align}
\end{lemma}
\begin{remark}
For the one-dimensional Schr\"odinger operator \(-\Delta + gV_\omega\) with a small coupling constant \(g > 0\), similar oscillatory terms as in \eqref{eqn:rho2-est} were estimated by \(O(g/\eta)\) in \cite[Theorem 14.6, Part (i)]{pastur1992book}. In the Schr\"odinger case, this term is of order \(O(g)\) since \(\eta\) does not depend on the small parameter \(g\). This approach does not directly apply to the div-grad case, where the small coupling \(g\) is replaced by the energy parameter \(\sqrt{E}\) in \cite[Theorem 14.6, Part (ii)]{pastur1992book} and \(\eta \gtrsim \sqrt{E}\), making \(O(g/\eta)\) of order \(O(1)\). We therefore provide complementary details leading to the refined bounds in \eqref{eqn:rho2-est}. These estimates clarify the dependence on \(E\) and supplement the argument underlying Theorem~\ref{thm:LE-linear}.
\end{remark}
\begin{proof}[Proof of Lemma~\ref{lem:rho2-est}]
Let \(\chi_j, \eta, Q_j\) be as in the recurrence relation \eqref{eqn:chin-rec}. Define
\begin{align*} 
    \zeta_j = e^{2i \chi_j}, \quad \mu = e^{2i \eta}, \quad 0 \le j \le n.
\end{align*}
Then \eqref{eqn:chin-rec} is equivalent to
\begin{align*}
    \zeta_{j+1} = \mu \zeta_j + \frac{i}{2} Q_j \frac{(\mu \zeta_j - 1)^2}{1 - \frac{i}{2} Q_j (\mu \zeta_j - 1)}, \quad 0 \le j \le n-1.
\end{align*}
Using \(Q_j = O(\sqrt{E})\) and \(|\mu \zeta_j - 1| \le 2\) (both uniform in \(j\)) to expand the last term in powers of \(Q_j\), we obtain
\begin{align}\label{eqn:zeta-rec}
    \zeta_{j+1} = \mu \zeta_j + \frac{i}{2} (\mu \zeta_j - 1)^2 Q_j + O(Q_j^2), \quad 0 \le j \le n-1.
\end{align}
Summing both sides over \(0 \le j \le n-1\) and dividing by \(n\) gives
\begin{align*}
    \frac{\zeta_n - \zeta_0}{n} + \frac{1}{n}\sum_{j=0}^{n-1}\zeta_j
    = \frac{\mu}{n}\sum_{j=0}^{n-1}\zeta_j + \frac{i}{2} \frac{1}{n}\sum_{j=0}^{n-1}(\mu \zeta_j - 1)^2 Q_j + \frac{1}{n}\sum_{j=0}^{n-1}O(Q_j^2).
\end{align*}
Combining this with \(Q_j = O(\sqrt{E})\) in \eqref{eqn:Qn} and 
\[
1 - \mu = 1 - \cos(2\eta) - i\sin(2\eta) = -2ie^{i\eta}\sin \eta
\]
implies
\begin{align*}
    \frac{-2ie^{i\eta}\sin \eta}{n} \sum_{j=0}^{n-1}\zeta_j 
    = \frac{i}{2n}\sum_{j=0}^{n-1}(\mu \zeta_j - 1)^2 Q_j + O(E) + \frac{\zeta_0 - \zeta_n}{n}.
\end{align*}

By the definition of \(\eta\) in \eqref{eqn:eta-app}, for \(0<E<2\kappa\)
\begin{align*}
\sqrt{\frac{E}{2\kappa}} \le \sin \eta = \sqrt{\frac{E}{\kappa}\Big(1 - \frac{E}{4\kappa}\Big)} \le \sqrt{\frac{E}{\kappa}}.
\end{align*}
Thus,
\begin{align}
   \frac{1}{n}\sum_{j=0}^{n-1}\zeta_j
   =& -\frac{e^{-i\eta}}{4n\sin \eta}\sum_{j=0}^{n-1}(\mu \zeta_j - 1)^2 Q_j + \frac{O(E)}{\sin \eta} + \frac{\zeta_0 - \zeta_n}{n\sin \eta} \notag \\
   =& -\frac{e^{-i\eta}}{4n\sin \eta}\sum_{j=0}^{n-1}(\mu \zeta_j - 1)^2 Q_j + O(\sqrt{E}), \label{eqn:sum-zeta}
\end{align}
provided \(n > 1/E\). Since \(\cos 2(\chi_j + \eta) = {\rm Re}(\mu \zeta_j)\), multiplying both sides of \eqref{eqn:sum-zeta} by \(\mu\) and taking the real part gives
\begin{align*}
   \frac{1}{n}\sum_{i=0}^{n-1}\cos 2(\chi_i + \eta)
   =& -{\rm Re}\Bigg[\frac{\mu e^{-i\eta}}{4n\sin \eta}\sum_{j=0}^{n-1}(\mu \zeta_j - 1)^2 Q_j\Bigg] + O(\sqrt{E}) \\
   =& -\frac{1}{4n\sin \eta}\sum_{j=0}^{n-1}{\rm Re}\big(\mu e^{-i\eta}(\mu \zeta_j - 1)^2\big) Q_j + O(\sqrt{E}).
\end{align*}
Taking expectations and using the independence of \(Q_j\) and \(\chi_j\) (hence \(\zeta_j\)), together with \(\E[Q_j] = 0\), we obtain
\begin{align*}
   \frac{1}{n}\E\Bigg[\sum_{i=0}^{n-1}\cos 2(\chi_i + \eta)\Bigg]
   =& -\frac{1}{4n\sin \eta}\sum_{j=0}^{n-1}\E{\rm Re}\big(\mu e^{-i\eta}(\mu \zeta_j - 1)^2\big)\E Q_j + O(\sqrt{E}) \\
   =& O(\sqrt{E}),
\end{align*}
which proves the first inequality in \eqref{eqn:rho2-est}.

Squaring \eqref{eqn:zeta-rec} gives
\begin{align}\label{eqn:zeta-square}
    \zeta_{j+1}^2 = \mu^2 \zeta_j^2 + i\mu \zeta_j(\mu \zeta_j - 1)^2 Q_j + O(Q_j^2), \quad 0 \le j \le n-1.
\end{align}
Similar arguments, together with
\[
1 - \mu^2 = -2i\mu\sin(2\eta), \quad \text{and} \quad \cos 4(\chi_j + \eta) = {\rm Re}(\mu^2 \zeta_j^2),
\]
prove the second inequality in \eqref{eqn:rho2-est}.
\end{proof}

\section{Quantum transport for large energies or frequencies}\label{sec:CT-logT}
Recall the definitions in \eqref{eqn:M-alpha} and \eqref{eqn:fullM}: 
\begin{align} 
    M_T^{q, \alpha_0, \alpha_1}(I) = \frac{1}{\pi T} \int_{I}\,  \sum_{T^{\alpha_0} \le |n| \le T^{\alpha_1}} |n|^q    \, |G^{E + i/T}(n, 0)|^2\, dE,
\end{align}
and 
\begin{align} 
    M_T^q(I) = M_T^{q, 0, \infty}(I) = \frac{1}{\pi T} \int_{I}\,  \sum_{n\in \Z} |n|^q    \, |G^{E + i/T}(n, 0)|^2\, dE.
\end{align}
We estimate the quantum transport either when the frequency \(\lvert n\rvert\) is large or when the energy \(E\) lies away from the critical value \(E_c = 0\).
%%%%%%%%%%%%%%%%%%%%%%%%%%%%%%%%%%%%%%%%%%%%%%%%%%%%%%%%%%%%%%%%%%%%%%%%%%%%%%%%%%%%%%%%%%%%%%%%%%%% Combes–Thomas estimate

\subsection{Combes–Thomas estimate and the proof of \eqref{eqn:Mtq-CT}}
Recall that the almost-sure spectrum of \(H_\omega\) is given by
\(\sigma(H_\omega) = [0,E_{\max}]\),
where \(E_{\max} = 4a_+\) in \eqref{eqn:spe}. Let \(\sigma_1 = [-E_1,E_1]\) where \(E_1 = 2eE_{\max} \ge 2E_{\max}\). Then for \(z = E + \frac{i}{T}\) with \(E \notin \sigma_1\), we have
\[
\rho = {\rm dist}(z,\sigma(H_\omega)) \ge |E| - E_{\max} \ge \frac{|E|}{2}.
\]
By the Combes–Thomas estimate (see, e.g., \cite[Lemma 1]{jitomirskaya2007upper}),  
for \(C_1 = (4a_+)^{-1} = E_{\max}^{-1}\),
\begin{align}\label{eqn:CT}
    |G^z(n,0)| \le \frac{2}{\rho}\exp\big\{-{\rm arcsinh}(C_1\rho)|n|\big\}.
\end{align}
Using the fact that \({\rm arcsinh}(x)\) is monotonically increasing and \({\rm arcsinh}(x) \ge \ln(x) \ge 1\) for \(x \ge e\), one has
\[
{\rm arcsinh}(C_1\rho) \ge {\rm arcsinh}\big(C_1\frac{|E|}{2}\big) \ge \ln\big(\tfrac{1}{2}C_1|E|\big) \ge 0,
\]
since \(C_1|E| \ge 2e\) for \(|E| \ge 2eE_{\max}\). Combining this with the Combes–Thomas estimate, one obtains for \(|E| \ge 2eE_{\max}\) and \(|n| \ge 1\),
\begin{align}
    |G^z(n,0)| \le \frac{4}{|E|}\exp\Big\{-\ln\big(\tfrac{1}{2}C_1|E|\big)|n|\Big\}.
\end{align}
Then applying \eqref{eqn:geo-sum} with \(\tau = 1\) and \(\Delta = \ln\Big(\tfrac{1}{2}C_1|E|\Big) \ge 1\) implies that for any \(q > 0\) there exists a constant \(C_q > 0\) such that
\begin{align}
    \sum_{|n| \ge 1}|n|^q |G^z(n,0)|^2 \le \frac{4}{|E|}C_q\exp\Big\{-\ln\big(\tfrac{1}{2}C_1|E|\big)\Big\} = \frac{8C_q}{C_1E^2}.
\end{align}
Hence,
\begin{align}
    \frac{1}{\pi T} \int_{|E| \ge E_1} \sum_{|n| \ge 1}|n|^q |G^z(n,0)|^2 \, dE \le \frac{1}{\pi T} \int_{|E| \ge E_1} \frac{8C_q}{C_1E^2} \, dE = \frac{2C_q}{C_1ea_+}\frac{1}{T}.
\end{align}
This term contributes to the first term in \eqref{eqn:Mtq-CT} and is bounded by \(\frac{2C_q}{C_1ea_+}\) for any \(T \ge 1\).

When the real part of \(z\) is not large enough, i.e., \(z = E + i/T\) with
\(E \in [-E_1,E_1]\), we bound from below as \(\rho = {\rm dist}(z,\sigma(H_\omega)) \ge \frac{1}{T}\). In this case, we use the fact that \({\rm arcsinh}(x)\) is monotonically increasing and \({\rm arcsinh}(x) \ge x/2\) for \(0 < x < 4\). Then the Combes–Thomas estimate \eqref{eqn:CT} gives
\[
|G^z(n,0)| \le \frac{2}{\rho}\exp\{-{\rm arcsinh}(C_1\rho)|n|\} \le 2T\exp\big\{-\tfrac{C_1}{2T}|n|\big\},
\]
provided \(T \ge 4C_1\). Similarly, for any \(\alpha > 0\) and \(q > 0\), applying \eqref{eqn:geo-sum} with \(\tau = 1\), \(\Delta = \tfrac{C_1}{2T}\), and \(N = T^{1+\alpha}\) implies that there exists a constant \(C'_q > 0\) and \(q' > 0\) such that for any \(T > 4C_1\),
\begin{align}
    \sum_{|n| \ge T^{1+\alpha}}|n|^q |G^z(n,0)|^2 \le C'_q T^{q'}e^{-\tfrac{C_1}{2}T^{\alpha}}.
\end{align}
Hence, there exists \(C = C(q,\alpha,a_+) > 0\) and \(T_0 = T_0(\alpha,q,a_+) > 0\) such that for \(T > T_0\), one has
\begin{align}
    \frac{1}{\pi T} \int_{|E| \le E_1} \sum_{|n| \ge T^{1+\alpha}}|n|^q |G^z(n,0)|^2 \, dE \le \frac{2}{\pi T}|E_1|C'_q T^{q'}e^{-\tfrac{C_1}{2}T^{\alpha}} \le C.
\end{align}
This contribution corresponds to the second term in \eqref{eqn:Mtq-CT} and completes its proof.

%%%%%%%%%%%%%%%%%%%%%%%%%%%%%%%%%%%%%%%%%%%%%%%%%%%%%%%%%%%%%%%%%%%%%%%%%%%%%%%%%%%%%%%%%%%%%%%%%%%%%%%%%%%%%%%%%%%%%%%%%%%%%%%%%%%%%%%%%%%%%%%%%%%%%%%%%%%%%%%%%%%%%%%%%%%%%%%%%%%%%%%%   log T term
 \subsection{Logarithmic growth of the quantum transport due to positive Lyapunov exponent}

Let \(\sigma_2 = \{E : E_0 \le |E| \le E_1\}\). In this part, we estimate \(M_T^{q,0,1+\alpha}(\sigma_2)\). The logarithmic bound in \eqref{eqn:Mtq-logT} actually holds for any \(E_1 > E_0 > 0\) with a constant depending on \(E_0, E_1\). We use the choice of \(E_0, E_1\) from Lemma~\ref{lem:CT-logT} for simplicity. Since \(E_c = 0\) is the only critical energy such that \(L(E_c) = 0\), by continuity of the Lyapunov exponent, there exist \(\gamma_0, \gamma_1 > 0\) such that
\begin{align}\label{eqn:LE-gamma0}
\gamma_1 \ge L(z) \ge L(E) \ge \gamma_0 > 0, \ \ {\rm for}\ \ z = E + \frac{i}{T}, \ \ E_0 \le |E| \le E_1.
\end{align}
The contribution for \(E_0 \le |E| \le E_1\) is at most logarithmic, as in \eqref{eqn:Mtq-CT}, due to the uniform lower bound \eqref{eqn:LE-gamma0}. The proof essentially follows \cite[Theorem 1]{jitomirskaya2007upper}. For completeness, we include a self-contained proof for \(E_0 < E < E_1\); the case \(-E_1 < E < -E_0\) can be treated in exactly the same way.

The goal is to obtain bootstrap large deviation estimates similar to those in Section~\ref{sec:bootstrap-ldt}. The difference here is that both the upper and lower bounds in \eqref{eqn:LE-gamma0} are independent of \(E\), which makes the proof much simpler. Using the argument for \eqref{eqn:p0def}–\eqref{eqn:p0}, we have for any \(n_1 > 0\) and \(j \in \Z\),
\begin{align*} 
     \mathbb{P}\big\{\omega : \|T^z_{n_1}(S^j \omega)\| \ge e^{\frac{1}{2}\gamma_0 n_1}\big\} \ge \frac{\gamma_0}{2\gamma_1 - \gamma_0} := p_1 > 0.
\end{align*}
Then, using the splitting argument for \eqref{eqn:prob-Tn1}–\eqref{eqn:550}, we have for any \(n > n_1\),
\begin{align*} 
    \Big\{\omega : \max_{1 \le j \le \frac{n}{n_1}} \|T_{j n_1}^z(\omega)\| \le e^{\frac{1}{4}\gamma_0 n_1}\Big\} 
    \subset \bigcap_{j = 0}^{\frac{n}{n_1} - 1} \Big\{\omega : \|T_{n_1}^z(S^{j n_1}\omega)\| \le e^{\frac{1}{2}\gamma_0 n_1}\Big\},
\end{align*}
which implies
\begin{align*}
\mathbb{P}\Big\{\omega : \max_{1 \le j \le \frac{n}{n_1}} \|T_{j n_1}^z(\omega)\|^2 \le e^{\frac{1}{2}\gamma_0 n_1}\Big\} 
     \le (1 - p_1)^{\frac{n}{n_1}}  
    \le e^{-c \frac{n}{n_1}}, \ \ c = |\log(1 - p_1)|.
\end{align*}
Setting \(\tfrac{1}{2}\gamma_0 n_1 = c  {n}/{n_1}\) gives \(n_1 = \sqrt{ {2cn}/{\gamma_0}}\). We have
\begin{align*}
    \mathbb{P}\Big\{\omega : \max_{1 \le j \le n} \|T_{j}^z(\omega)\|^2 \le e^{\sqrt{\gamma_0 c/2}\sqrt{n}}\Big\} 
    \le e^{-\sqrt{\gamma_0 c/2}\sqrt{n}}.
\end{align*}
 As a consequence, for \(z = E + i/T\) with \(E_0 < E < E_1\),
\begin{align*} 
    \mathbb{E}\Big(\frac{1}{\max\limits_{1 \le j \le n} \|T_j^z(\omega)\|^2}\Big) 
    \le 2e^{-\sqrt{\gamma_0 c/2}\sqrt{n}}.
\end{align*}
Combining this with \eqref{eqn:green-tail}, we obtain for some constant \(c_1\) depending on \(a_-, a_+\) in \eqref{eqn:an-bound},
\begin{align}\label{eqn:green-app}
    \mathbb{E}|G^z(n, 0; \omega)|^2 \le 2c_1 T^6 e^{-\sqrt{\gamma_0 c/4}\sqrt{n}}.
\end{align}
Then, using \eqref{eqn:green-app} in place of \eqref{eqn:Green-upper} in the proof of \eqref{eqn:J2-est}, we apply \eqref{eqn:geo-sum} to conclude that for any \(q>0\) there exist constants \(C\) and \(T_1\) such that for \(T \ge T_1\),
\begin{align*}
  \int_{E_0}^{E_1} \sum_{|n| \ge (\log T)^3} |n|^q \, \mathbb{E}|G^z(n,0)|^2 \frac{dE}{\pi T} 
  \le \frac{|E_1|}{\pi T} \sum_{|n| \ge (\log T)^3} |n|^q \, 2c_1 T^6 e^{-\sqrt{\gamma_0 c/4}\sqrt{|n|}}  
  \le C.
\end{align*}
Finally, by \eqref{eqn:GImB-intFull},
\begin{align}
    \int_{E_0}^{E_1} \sum_{|n| \le (\log T)^3} |n|^q \, \mathbb{E}|G^z(n,0)|^2 \frac{dE}{\pi T} \le (\log T)^{3q}.
\end{align}
Combining these two parts proves \eqref{eqn:Mtq-logT} provided \((\log T)^{3q}\ge C\).

\section{Estimates of the Borel Transform of a Measure}\label{sec:borel}
In this section, we provide quantitative estimates for the Borel transform of a measure and use them to prove \eqref{eqn:GImB-int} in Lemma~\ref{lem:GImB}. The following result generalizes \cite[Proposition 3]{jitomirskaya2007upper}.
\begin{proposition}\label{prop:ImB-upper}
Consider a Borel measure \(\mu\) on \(\R\) normalized so that \(\mu(\R) = 1\). Define its Borel transform by
\begin{align}
  B_{\mu}(z)=\int \frac{1}{E-z}\, d\mu(E), \quad z\in \C\backslash \R.     
\end{align}  

For \(0<m\le 1\) and \(E_c \in \R\), assume there exist constants \(C_0 > 0\) and \(E_0 > 0\) such that
\begin{align}\label{eqn:mu-m-holder}
    \mu([E_c-E,E_c+E])<C_0E^m, \quad 0<E<E_0.
\end{align}
Then there exists \(C_1 = C_1(m,E_0,C_0) > 0\) such that for any finite \(\delta>0\),
\begin{align}\label{eqn:imB-bound}
    {\rm Im\, } B_{\mu}(E_c+i\delta)<C_1 \, \delta^{m-1}. 
\end{align}
Consequently, for any \(0<m\le 1\) and any finite \(\delta,E_1>0\),
\begin{align}\label{eqn:int-ImB}
    \int_{E_c-E_1}^{E_c+E_1} {\rm Im\, } B_{\mu}(E+i\delta)\, dE\le 2\pi C_1\,E_1^{m}.
\end{align}
\end{proposition}

\begin{proof}
A direct computation using Fubini’s theorem shows that for \(\delta>0\),
\begin{align*}
    {\rm Im\, } B_{\mu}(E_c+i\delta)=&\int_{\R} \frac{\delta}{(E_c-E)^2+\delta^2}\, d\mu(E)\\
    =& \delta \int_{0}^{ \frac{1}{ \delta^2}}\mu\Big(\big\{E:|E-E_c| <\sqrt{\frac{1}{t}-\delta^2}\big\} \Big)  dt. 
\end{align*}
Since
\begin{align*}
    \sqrt{\frac{1}{t}-\delta^2}<E_0\Longleftrightarrow  t>\frac{1}{E_0^2+\delta^2},
\end{align*}
we split the last integral into two parts:
\begin{align*}
    {\rm Im\, } B_{\mu}(E_c+i\delta) 
    =& \delta \int_{ \frac{1}{E_0^2+\delta^2}}^{ \frac{1}{\delta^2}}\mu\big(\{E:|E-E_c| <\sqrt{\frac{1}{t}-\delta^2}\} \big)  dt  \\
    &+\delta \int_{0}^{  \frac{1}{E_0^2+\delta^2}}\mu\big(\{E:|E-E_c| <\sqrt{\frac{1}{t}-\delta^2}\} \big)  dt.
\end{align*}
We bound the measure in the first part using \eqref{eqn:mu-m-holder} and in the second part by the total mass \(\mu(\R)=1\). This gives
\begin{align*}
    {\rm Im\, } B_{\mu}(E_c+i\delta) 
    \le & \delta \int_{ \frac{1}{E_0^2+\delta^2}}^{ \frac{1}{\delta^2}} C_0\Big( \sqrt{\frac{1}{t}-\delta^2}\Big)^m  dt + \delta \int_{0}^{  \frac{1}{E_0^2+\delta^2}}\,\mu(\R)\,  dt \\
    = & C_0\delta \int_{ 0}^{ \frac{1}{\delta^2}}  \Big( \frac{1}{t}-\delta^2\Big)^{\frac{m}{2}}  dt + \frac{\delta}{E_0^2+\delta^2}.
\end{align*}
The first integral is explicitly computable and strictly positive for any \(\delta>0\) and \(0<m<2\):
\begin{align*}
    \int_{ 0}^{ \frac{1}{\delta^2}}  \Big( \frac{1}{t}-\delta^2\Big)^{\frac{m}{2}}  dt=\delta^{\,m-2} \, \Gamma\!\big( \tfrac{m}{2}+1 \big)\, \Gamma\!\big( 1 - \tfrac{m}{2} \big)
    = \delta^{\,m-2} \, \frac{m\pi}{2}\, \frac{1}{\sin(m\pi/2)}  >0,
\end{align*}
where \(\Gamma(x)\) is the gamma function. The second term satisfies, for any \(\delta>0,E_0>0\), and \(m>0\),
\begin{align*}
    \frac{\delta}{E_0^2+\delta^2}\le \max( {E_0^{-2}},1)\cdot \delta^{m-1}.
\end{align*}
 Thus, for any \(0<m\le 1\) and \(E_0,\delta>0\),
\begin{align*}
    {\rm Im\, } B_{\mu}(E_c+i\delta) 
    \le  C_1\, \delta^{m-1}, \quad C_1=C_0\frac{m\pi}{2}\, \csc(\frac{m\pi}{2})+\max({E_0^{-2}},1).
\end{align*}
   
The proof of \eqref{eqn:int-ImB} follows by integrating \eqref{eqn:imB-bound}. A direct computation shows that for any \(E_1>0,E'\in \R\),
\begin{align*}
    \frac{2E_1^2}{E_1^2+(E'-E_c)^2}>0,  &\quad {\rm if} \quad   |E'-E_c|>E_1, \\ 
    \frac{2E_1^2}{E_1^2+(E'-E_c)^2}\geq 1, &\quad  {\rm if} \quad   |E'-E_c|\le E_1.
\end{align*}
Thus, as a function of \(E'\):
\[
    \frac{2E_1^2}{E_1^2+(E'-E_c)^2}\geq \chi_{[E_c-E_1,E_c+E_1]}(E').
\]
Therefore,
\begin{align}
    \int_{E_c-E_1}^{E_c+E_1} {\rm Im\, } B_{\mu}(E'+i\delta)\, dE'=&\int_\R \chi_{[E_c-E_1, E_c+E_1]}(E'){\rm Im\, } B_{\mu}(E'+i\delta)\, dE'  \notag  \\
    \le &\int_\R \frac{2E_1^2}{E_1^2+(E'-E_c)^2}{\rm Im\, } B_{\mu}(E'+i\delta)\, dE'  \notag \\
    =&\int_\R \frac{2E_1^2}{E_1^2+(E'-E_c)^2}\,   \Big(\int_\R   \frac{\delta}{(E'-E)^2+\delta^2}d\mu(E) \Big) \,  dE'   \notag \\ 
    =& 2E_1^2\delta \int_\R \Big(\int_\R \frac{1}{(E'-E_c)^2+E_1^2} \frac{1}{(E'-E)^2+\delta^2}dE' \Big) \, d\mu(E)  \label{eqn:imB-pf}
\end{align}
By the Fourier transform and Parseval identity:
\begin{equation}\label{eq:fourierIdty}
    \int_{-\infty}^\infty\frac{1}{(x-a)^2+A^2}\cdot \frac{1}{(x-b)^2+B^2}dx=\frac{\pi}{AB}\frac{A+B}{(a-b)^2+(A+B)^2}.
\end{equation}
Applying \eqref{eq:fourierIdty} to the inner integral of \eqref{eqn:imB-pf} with \(a=E_c,A=E_1,b=E,B=\delta\) gives    
\begin{align*}
    \int_{E_c-E_1}^{E_c+E_1}\im B_\mu(E'+i\delta)dE'\leq& 2\pi E_1\int_\R   \frac{\delta+E_1}{(E_c-E)^2+(\delta+E_1)^2}    \, d\mu(E) \\
    =& 2\pi E_1\, {\rm Im\,}B_\mu\big(E_c+i(\delta+E_1)\big) \\
    (\eqref{eqn:imB-bound}\Longrightarrow)\quad\le &  2\pi E_1\, C_1(\delta+E_1)^{m-1}\\
    \le &  2\pi E_1\, C_1(E_1)^{m-1}=2\pi C_1    E_1 ^{m}
\end{align*}
provided \(m-1\le 0\) and \(\delta,E_1>0\).
\end{proof}

We now apply the above proposition to the density of states measure \(d\mathcal N(E)\) with \(m=\frac{1}{2}\) and give the proof of:

\begin{proof}[Proof of Lemma~\ref{lem:GImB}]
As noted, the first two estimates, \eqref{eqn:GImB} and \eqref{eqn:GImB-intFull} in Lemma~\ref{lem:GImB}, were established in \cite[Lemma~5]{jitomirskaya2007upper}. We sketch the proof for completeness.

Let \(z=E+i/T\). Recall the definition of \(G^z(n,m;\omega)\) in \eqref{eqn:gn}:
\[
    G^z(n, m; \omega) = \langle \delta_n, (H_\omega - z)^{-1} \delta_m \rangle.
\]
Hence,
\begin{align*}
    \sum_{n\in \Z} |G^z(n, 0; \omega)|^2 =&\sum_{n\in \Z}\langle \delta_0, (H_\omega -   z)^{-1} \delta_n \rangle \langle \delta_n, (H_\omega - \bar z)^{-1} \delta_0 \rangle \\
    =& \langle \delta_0, \Big[(H_\omega -E)^2+T^{-2} \Big]^{-1} \delta_0 \rangle.
\end{align*}
Combining this with the spectral theorem and \eqref{eqn:green-ids} gives
\begin{align*} 
    \sum_{1 \le |n| \le N} \frac{|n|^q}{\pi T} \, \mathbb{E} |G^z(n, 0)|^2 \le &  \frac{N^q}{\pi T} \, \sum_{n\in \Z} \E |G^z(n, 0; \omega)|^2 \\
    \le &   \frac{N^q}{\pi  } \, \E\langle \delta_0, \frac{T^{-1}}{(H_\omega -E)^2+T^{-2}} \delta_0 \rangle=\frac{N^q}{\pi  } {\rm Im}B_{\mathcal N}(E+\frac{i}{T}),
\end{align*}
which proves \eqref{eqn:GImB}. For any \(E_0<E_1\), a direct computation shows that
\begin{align} 
    \int_{E_0}^{E_1} \sum_{1 \le |n| \le N} \frac{|n|^q}{\pi T} \, \mathbb{E} |G^z(n, 0)|^2 \, dE \le &  \frac{N^q}{\pi  } \int_{\R}  {\rm Im}B_{\mathcal N}(E+\frac{i}{T}) \, dE  \label{eqn:E7}  \\
    \le & \frac{N^q}{\pi  } \int_{\R} \Big(\int \frac{T^{-1}}{(E-E')^2+T^{-1}}\, d\mathcal N(E')\Big) \, dE \notag \\
    =& \frac{N^q}{\pi  }  \Big(\int_{\R}\, \pi\, d\mathcal N(E')\Big) =N^q, \notag
\end{align}
which proves \eqref{eqn:GImB-intFull}. 

Finally, let \(D_1,E_0\) be as in \eqref{eqn:NE-bound} such that
\begin{align*}
    0\le   \mathcal{N} ( E ) \le D_1 \sqrt{E}, \quad 0<E<E_0.
\end{align*}
Note that \(\mathcal N (E)=0\) for \(E\le 0\). Applying \eqref{eqn:int-ImB} to \(\mathcal N\) with \(E_c=0,m=\frac{1}{2}\) gives that there exists \(C>0\), depending on \(D_1,E_0\), such that for any \(E_2>0\) and \(T>0\),
\begin{align*} 
    \int_{-E_2}^{E_2} {\rm Im\, } B_{\mathcal N}(E+i/T)\, dE\le 2\pi C \,E_2^{\frac{1}{2}}.
\end{align*}
Then, similar to \eqref{eqn:E7}, we obtain
\begin{align*} 
    \int_{-E_2}^{E_2} \sum_{1 \le |n| \le N} \frac{|n|^q}{\pi T} \, \mathbb{E} |G^z(n, 0)|^2 \, dE \le &  \frac{N^q}{\pi  }  \int_{-E_2}^{E_2}   {\rm Im}B_{\mathcal N}(E+\frac{i}{T}) \, dE    \\
    \le & \frac{N^q}{\pi  } 2\pi C \,E_2^{\frac{1}{2}},
\end{align*}
which proves \eqref{eqn:GImB-int}.
\end{proof}

%%%%%%%%%%%%%%%%%%%%%%%%%%%%%%%%%%%%%%%%%%%%%%%%%%%%%%%%%%%%%%%%%%%%%%%%%%%%%%%%%%%%%%%%%%%%%%%%%%%%%%%%%%%%%%%%%%%%%%%%%%%%%%%%%%%%%%%%%%%%%%%%%%%%%%%%%%%%%%%%%%%%%%%

\noindent\textbf{Acknowledgments.}
\phantomsection
\addcontentsline{toc}{section}{Acknowledgments}
W. Wang is supported in part by the National Key R\&D Program of China (2024YFA1012302). L. Li is supported by AMS-Simons Travel Grant 2024-2026. S. Zhang is supported by the NSF grant DMS-2418611.

%%%%%%%%%%%%%%%%%%%%%%%%%%%%%%%%%%%%%%%%%%%%%%%%%%%%%%%%%%%%%%%%%%%%%%%%%%%%%%%%%%%%%%%%%%%%%%%%

\bibliographystyle{abbrv}
\bibliography{divgrad.bbl} % bibliography items are in separate bibtex file divgrad.bib
%\nocite{*} % prints all references in bibliography, temporary for now

{
  \bigskip
  \vskip 0.08in \noindent --------------------------------------

\footnotesize
\medskip
L.~Li, {Department of Mathematics, 
Texas A\&M University, 155 Ireland Street,
College Station, TX 77843}\par\nopagebreak
    \textit{E-mail address}:  \href{mailto:longli@tamu.edu }{longli@tamu.edu }

\vskip 0.4cm

 W. ~Wang, {SKLMS, Academy of Mathematics and Systems Science, Chinese Academy of Sciences, Beijing 100190, China}\par\nopagebreak
  \textit{E-mail address}: \href{mailto:ww@lsec.cc.ac.cn}{ww@lsec.cc.ac.cn}
  
\vskip 0.4cm

S.~Zhang, {Department of Mathematics and Statistics, University of Massachusetts Lowell, 
Southwick Hall, 
11 University Ave.
Lowell, MA 01854
 }\par\nopagebreak
  \textit{E-mail address}: \href{mailto:shiwen\_zhang@uml.edu}{shiwen\_zhang@uml.edu}
}

\end{document}